\documentclass[letterpaper, 11pt]{article}
\usepackage[utf8]{inputenc}
\usepackage{amsmath,amsthm}
\usepackage{amsfonts}
\usepackage{amssymb}
\usepackage{graphicx}
\usepackage[margin=1in]{geometry}
\usepackage{algorithmic}
\usepackage{mathtools}
\usepackage[ruled, linesnumbered]{algorithm2e}

\SetKwComment{Comment}{$\triangleright$\ }{}
\SetKwFor{Loop}{Loop}{}{}

\usepackage{hyperref}
\makeatletter
\def\namedlabel#1#2{\begingroup
    #2%
    \def\@currentlabel{#2}%
    \phantomsection\label{#1}\endgroup
}
\makeatother

\newtheorem{theorem}{Theorem}[]
\newtheorem{corollary}{Corollary}[section]
\newtheorem{claim}{Claim}[section]
\newtheorem{definition}{Definition}[section]
\newtheorem{lemma}{Lemma}[section]

\newcommand{\eps}{\varepsilon}
\newcommand{\lis}{\mathsf{LIS}}
\newcommand{\lcs}{\mathsf{LCS}}
\newcommand{\ed}{\mathsf{ED}}
\newcommand{\lns}{\mathsf{LNS}}
\newcommand{\lnst}{\mathsf{LNST}}
\newcommand{\N}{\mathbb{N}}
\newcommand{\dis}{\mathsf{DIS}}
\newcommand{\opt}{\mathsf{OPT}}

\newcommand{\tB}{\tilde{B}}
\newcommand{\bB}{\bar{B}}

\newcommand{\polylog}{\operatorname{polylog}}

\usepackage{graphicx}

\begin{document}
\begin{titlepage}
\def\thepage{}

\title{Lower Bounds and Improved Algorithms for Asymmetric Streaming Edit Distance and Longest Common Subsequence}


\author{Xin Li \thanks{Supported by NSF CAREER Award CCF-1845349.}   \\  Department of Computer Science, \\
	 Johns Hopkins University. \\lixints@cs.jhu.edu
	
	\and 
	
	Yu Zheng \thanks{Supported by NSF CAREER Award CCF-1845349.} \\
	Department of Computer Science \\
	Johns Hopkins University. \\
	yuzheng@cs.jhu.edu }

\date{}

\maketitle \thispagestyle{empty}
\begin{abstract}
In this paper, we study \emph{edit distance} ($\ed$) and \emph{longest common subsequence} ($\lcs$) in the asymmetric streaming model, introduced by Saks and Seshadhri \cite{saks2013space}. As an intermediate model between the random access model and the streaming model, this model allows one to have streaming access to one string and random access to the other string. Meanwhile, $\ed$ and $\lcs$ are both fundamental problems that are often studied on large strings, thus the (asymmetric) streaming model is ideal for studying these problems.  

Our first main contribution is a systematic study of space lower bounds for $\ed$ and $\lcs$ in the asymmetric streaming model. Previously, there are no explicitly stated results in this context, although some lower bounds about $\lcs$ can be inferred from the lower bounds for \emph{longest increasing subsequence} ($\lis$) in \cite{sun2007communication, gal2010lower, ergun2008distance}. Yet these bounds only work for large alphabet size. In this paper, we develop several new techniques to handle $\ed$ in general and $\lcs$ for small alphabet size, thus establishing strong lower bounds for both problems. In particular, our lower bound for $\ed$ provides an \emph{exponential} separation between edit distance and Hamming distance in the asymmetric streaming model. Our lower bounds also extend to $\lis$ and \emph{longest non-decreasing subsequence} (\emph{$\lns$}) in the standard streaming model. Together with previous results, our bounds provide an almost complete picture for these two problems.

As our second main contribution, we give improved algorithms for $\ed$ and $\lcs$ in the asymmetric streaming model.\ For $\ed$, we improve the space complexity of the constant factor approximation algorithms in \cite{farhadi2020streaming, cheng2020space} from $\tilde{O}(\frac{n^\delta}{\delta})$ to $O(\frac{d^\delta}{\delta} \; \polylog(n))$, where $n$ is the length of each string and $d$ is the edit distance between the two strings. For $\lcs$, we give the first $1/2+\eps$ approximation algorithm with space $n^{\delta}$ for any constant $\delta>0$, over a binary alphabet. Our work leaves a plethora of intriguing open questions, including establishing lower bounds and designing algorithms for a natural generalization of $\lis$ and $\lns$, which we call \emph{longest non-decreasing subsequence with threshold} ($\lnst$).
\end{abstract}

\if

In this paper, we study edit distance (ED) and longest common subsequence (LCS) in the asymmetric streaming model, introduced by Saks and Seshadhri [SS13]. As an intermediate model between the random access model and the streaming model, this model allows one to have streaming access to one string and random access to the other string. ED and LCS are both fundamental problems that are often studied on large strings, thus the (asymmetric) streaming model is ideal for studying these problems. 
Our first main contribution is a systematic study of space lower bounds for ED and LCS in the asymmetric streaming model. Previously, there are no explicitly stated results in this context, although some lower bounds about LCS can be inferred from the lower bounds for longest increasing subsequence (LIS) in [SW07][GG10][EJ08]. Yet these bounds only work for large alphabet size. In this paper, we develop several new techniques to handle ED in general and LCS for small alphabet size, thus establishing strong lower bounds for both problems. In particular, our lower bound for ED provides an exponential separation between edit distance and Hamming distance in the asymmetric streaming model. Our lower bounds also extend to LIS and longest non-decreasing sequence (LNS) in the standard streaming model. Together with previous results, our bounds provide an almost complete picture for these two problems.
As our second main contribution, we give improved algorithms for ED and LCS in the asymmetric streaming model. For ED, we improve the space complexity of the constant factor approximation algorithms in [FHRS20][CJLZ20] from $\tilde{O}(\frac{n^\delta}{\delta})$ to $O(\frac{d^\delta}{\delta}\;\polylog(n))$, where $n$ is the length of each string and $d$ is the edit distance between the two strings. For LCS, we give the first $1/2+\epsilon$ approximation algorithm with space $n^{\delta}$ for any constant $\delta>0$, over a binary alphabet. Our work leaves a plethora of intriguing open questions, including establishing lower bounds and designing algorithms for a natural generalization of $\lis$ and $\lns$, which we call longest non-decreasing subsequence with threshold (LNST).

\fi
\end{titlepage}


\section{Introduction}
Edit distance ($\ed$) and longest common subsequence ($\lcs$) are two classical problems studied in the context of measuring similarities between two strings. Edit distance is defined as the smallest number of edit operations (insertions, deletions, and substitutions) to transform one string to the other, while longest common subsequence is defined as the longest string that appears as a subsequence in both strings. These two problems have found wide applications in areas such as  bioinformatics, text and speech processing, compiler design, data analysis, image analysis and so on. In turn, these applications have led to an extensive study of both problems.

With the era of information explosion, nowadays these two problems are often studied on very large strings. For example, in bioinformatics a human genome can be represented as a string with $3$ billion letters (base pairs). Such data provides a huge challenge to the algorithms for $\ed$ and $\lcs$, as the standard algorithms for these two problems using dynamic programming need $\Theta(n^2)$ time and $\Theta(n)$ space where $n$ is the length of each string. These bounds quickly become infeasible or too costly as $n$ becomes large, such as in the human genome example. Especially, some less powerful computers may not even have enough memory to store the data, let alone processing it.

One appealing approach to dealing with big data is designing \emph{streaming algorithms}, which are algorithms that process the input as a data stream. Typically, the goal is to compute or approximate the solution by using sublinear space (e.g., $n^{\alpha}$ for some constant $0<\alpha<1$ or even $\polylog(n)$) and a few (ideally one) passes of the data stream. These algorithms have become increasingly popular, and attracted a lot of research activities recently.

Designing  streaming algorithms for $\ed$ and $\lcs$, however, is not an easy task. For $\ed$, only a couple of positive results are known. In particular, assuming that the edit distance between the two strings is bounded by some parameter $k$, \cite{Chakraborty2015LowDE} gives a randomized one pass algorithm achieving an $O(k)$ approximation of $\ed$, using linear time and $O(\log n)$ space, in a variant of the streaming model where one can scan the two  strings simultaneously in  a  coordinated way. In the same model  \cite{Chakraborty2015LowDE} also give randomized  one pass algorithms computing $\ed$ exactly, using space $O(k^6)$  and  time $O(n+k^6)$. This was later  improved to space $O(k)$  and  time $O(n+k^2)$ in \cite{ChakrabortyGK2016, BelazzouguiZ16}. Furthermore, \cite{BelazzouguiZ16} give a randomized one pass algorithm computing $\ed$ exactly, using space $\tilde{O}(k^8)$  and  time $\tilde{O}(k^2 n)$, in the standard streaming model. We note that all of these algorithms are only interesting if $k$ is small, e.g., $k \leq n^{\alpha}$ where $\alpha$ is some small constant, otherwise the space complexity can be as large as $n$. For $\lcs$, strong lower bounds are given in \cite{liben2006finding, sun2007communication}, which show that for exact computation, even constant pass randomized algorithms need space $\Omega(n)$; while any constant pass deterministic algorithm achieving a $\frac{2}{\sqrt{n}}$ approximation of $\lcs$ also needs space $\Omega(n)$, if the alphabet size is at least $n$.

Motivated by this situation and inspired by the work of \cite{AKO10}, Saks and Seshadhri \cite{saks2013space} studied the asymmetric data streaming model. This model is a relaxation of the standard streaming model, where one has streaming access to one string (say $x$), and random access to the other string (say $y$). In this model, \cite{saks2013space} gives a deterministic one pass algorithm achieving a $1+\eps$ approximation of $n-\lcs$ using space $O(\sqrt{(n \log n)/\eps})$, as well as a randomized one pass algorithm algorithm achieving an $\eps n$ \emph{additive} approximation of $\lcs$  using space $O(k \log^2 n/\eps)$ where $k$ is the maximum number of times any symbol appears in $y$. Another work by Saha \cite{saha17} also gives an algorithm in this model that achieves an $\eps n$ \emph{additive} approximation of $\ed$ using space $O(\frac{\sqrt{n}}{\epsilon})$.

The asymmetric streaming model is interesting for several reasons. First, it still inherits the spirit of streaming algorithms, and is particularly suitable for a distributed setting. For example, a local, less powerful computer can use the streaming access to process the string $x$, while sending queries to a remote, more powerful server which has access to $y$. Second, because it is a relaxation of the standard streaming model, one can hope to design better algorithms for $\ed$ or to beat the strong lower bounds for $\lcs$ in this model. The latter point is indeed verified by two recent works \cite{farhadi2020streaming, cheng2020space} (recently accepted to ICALP as a combined paper \cite{ChengFHJLRSZ}), which give a deterministic one pass algorithm achieving a $O(2^{1/\delta})$ approximation of $\ed$, using space $\tilde{O}(n^{\delta}/\delta)$ and time $\tilde{O}_\delta(n^4)$ for any constant $\delta>0$, as well as deterministic one pass algorithms achieving $1 \pm \eps$ approximation of $\ed$ and $\lcs$, using space $\tilde{O}(\frac{\sqrt{n}}{\eps})$ and time $\tilde{O}_{\eps}(n^{2})$.

A natural question is how much we can improve these results. Towards answering this question, we study both lower bounds and upper bounds for the space complexity of $\ed$ and $\lcs$ in the asymmetric streaming model, and we obtain several new, non-trivial results.

\subparagraph*{Related work.}  On a different topic, there are many works that study the time complexity of $\ed$ and $\lcs$.\ In particular, while \cite{BI15, ABW15} showed that $\ed$ and $\lcs$ cannot be computed exactly in truly sub-quadratic time unless the strong Exponential time hypothesis \cite{IPZ01} is false, a successful line of work \cite{CDGKS18, brakensiek2019constant, koucky2019constant, AndoniN20, hajiaghayi2019lcs, RSSS19, rubinstein2020reducing} has led to randomized algorithms that achieve constant approximation of $\ed$ in near linear time, and randomized algorithms that provide various non-trivial approximation of $\lcs$ in linear or sub-quadratic time. Another related work is \cite{AKO10}, where the authors proved a lower bound on the \emph{query complexity} for computing $\ed$ in the \emph{asymmetric query} model, where one have random access to one string but only limited number of queries to the other string.   

\subsection{Our Contribution}
We initiate a systematic study on lower bounds for computing or approximating $\ed$ and $\lcs$ in the asymmetric streaming model.\ To simplify notation we always use $1+\eps$ approximation for some $\eps>0$, i.e., outputting an $\lambda$ with $\opt \leq \lambda \leq (1+\eps)\opt$, where $\opt$ is either $\ed(x, y)$ or $\lcs(x, y)$. We note that for $\lcs$, this is equivalent to a $1/(1+\eps)$ approximation in the standard notation.

Previously, there are no explicitly stated space lower bounds in this model, although as we will discuss later, some lower bounds about $\lcs$ can be inferred from the lower bounds for \emph{longest increasing subsequence} $\lis$ in \cite{sun2007communication, gal2010lower, ergun2008distance}. As our first contribution, we prove strong lower bounds for $\ed$ in the asymmetric streaming model. 

\begin{theorem}\label{thm:ED1}
There is a constant $c>1$ such that for any $k, n \in \N$ with $n \geq ck$, given an alphabet $\Sigma$, any $R$-pass randomized algorithm in the asymmetric streaming model that decides if $\ed(x, y) \geq k$ for two strings $x, y \in \Sigma^n$ with success probability $\geq 2/3$ must use space $\Omega(\mathsf{min}(k, |\Sigma|)/R)$.
\end{theorem}

This theorem implies the following corollary. 

\begin{corollary}\label{cor:ED2}
Given an alphabet $\Sigma$, the following space lower bounds hold for any constant pass randomized algorithm with success probability $\geq 2/3$ in the asymmetric streaming model.
\begin{enumerate}
\item $\Omega(n)$ for computing $\ed(x, y)$ of two strings $x, y \in \Sigma^n$ if $|\Sigma| \geq n$.
\item $\Omega(\frac{1}{\eps})$ for $1+\eps$ approximation of $\ed(x, y)$ for two strings $x, y \in \Sigma^n$ if $|\Sigma| \geq 1/\eps$.
\end{enumerate}
\end{corollary}

Our theorems thus provide a justification for the study of \emph{approximating} $\ed$ in the asymmetric streaming model. Furthermore, we note that previously, unconditional lower bounds for $\ed$ in various computational models are either weak, or almost identical to the bounds for Hamming distance. For example, a simple reduction from the equality function implies the deterministic two party communication complexity (and hence also the space lower bound in the standard streaming model) for computing or even approximating $\ed$ is $\Omega(n)$.\footnote{We include this bound in the appendix for completeness, as we cannot find any explicit statement in the literature.}  However the same bound holds for Hamming distance. Thus it has been an intriguing question to prove a rigorous, unconditional separation of the complexity of $\ed$ and Hamming distance. To the best of our knowledge the only previous example achieving this is the work of \cite{AndoniK2010} and \cite{AndoniJP2010}, which showed that the randomized two party communication complexity of achieving a $1+\eps$ approximation of $\ed$ is $\Omega(\frac{\log n}{(1+\eps)\log \log n})$, while the same problem for Hamming distance has an upper bound of $O(\frac{1}{\eps^2})$. Thus if $\eps$ is a constant, this provides a separation of $\Omega(\frac{\log n}{\log \log n})$ vs. a constant. However, this result also has some disadvantages: (1) It only works in the randomized setting; (2) The separation becomes obsolete when $\eps$ is small, e.g., $\eps=1/\sqrt{\log n}$; and (3) The lower bound for $\ed$ is still weak and thus it does not apply to the streaming setting, as there even recoding the index needs space $\log n$. 

Our result from Corollary~\ref{cor:ED2}, on the other hand, complements the above result in the aforementioned aspects by providing another strong separation of $\ed$ and Hamming distance. Note that even exact computation of the Hamming distance between $x$ and $y$ is easy in the asymmetric streaming model with one pass and space $O(\log n)$. Thus our result provides an \emph{exponential} gap between edit distance and Hamming distance, in terms of the space complexity in the asymmetric streaming model (and also the communication model since our proof uses communication complexity), even for deterministic exact computation. 

Next we turn to $\lcs$, which can be viewed as a generalization of $\lis$. For example, if the alphabet $\Sigma=[n]$, then we can fix the string $y$ to be the concatenation from $1$ to $n$, and it's easy to see that $\lcs(x, y)=\lis(x)$. Therefore, the lower bound of computing $\lis$ for randomized streaming in \cite{sun2007communication} with $|\Sigma| \geq n$ also implies a similar bound for $\lcs$ in the asymmetric streaming model. However, the bound in \cite{sun2007communication} does not apply to the harder case where $x$ is a \emph{permutation} of $y$, and their lower bound where $|\Sigma| < n$ is actually for \emph{longest non-decreasing subsequence}, which does not give a similar bound for $\lcs$ in the asymmetric streaming model. \footnote{One can get a similar reduction to $\lcs$, but now $y$ needs to be the sorted version of $x$, which gives additional information about $x$ in the asymmetric streaming model since we have random access to $y$.} Therefore, we first prove a strong lower bound for $\lcs$ in general.

\begin{theorem}\label{thm:LCS1}
	There is a constant $c>1$ such that for any $k, n\in \N$ with $n \geq ck$, given an alphabet $\Sigma$, any $R$-pass randomized algorithm in the asymmetric streaming model that decides if $\lcs(x,y)\geq k$ for two strings $x, y \in \Sigma^n$ with success probability $\geq 2/3$ must use space $\Omega\big(\mathsf{min}(k, |\Sigma|)/R\big) $.\ Moreover, this holds even if $x$ is a permutation of $y$ when $|\Sigma| \geq n$ or $|\Sigma| \leq k$.
\end{theorem}

Similar to the case of $\ed$, this theorem also implies the following corollary.

\begin{corollary}\label{cor:LCS2}
Given an alphabet $\Sigma$, the following space lower bounds hold for any constant pass randomized algorithm with success probability $\geq 2/3$ in the asymmetric streaming model.
\begin{enumerate}
\item $\Omega(n)$ for computing $\lcs(x, y)$ of two strings $x, y \in \Sigma^n$ if $|\Sigma| \geq n$.
\item $\Omega(\frac{1}{\eps})$ for $1+\eps$ approximation of $\lcs(x, y)$ for two strings $x, y \in \Sigma^n$ if $|\Sigma| \geq 1/\eps$.
\end{enumerate}
\end{corollary}

We then consider \emph{deterministic} approximation of $\lcs$. Here, the work of \cite{gal2010lower, ergun2008distance} gives a lower bound of $\Omega\left (\frac{1}{R}\sqrt{\frac{n}{\eps}}\log \left ( \frac{|\Sigma|}{\eps n}\right ) \right )$ for any $R$ pass streaming algorithm achieving a $1+\eps$ approximation of $\lis$, which also implies a lower bound of $\Omega\left (\frac{1}{R}\sqrt{\frac{n}{\eps}}\log \left ( \frac{1}{\eps}\right ) \right )$ for asymmetric streaming $\lcs$ when $|\Sigma| \geq n$. These bounds match the upper bound in \cite{gopalan2007estimating} for $\lis$ and $\lns$, and in \cite{farhadi2020streaming, cheng2020space} for $\lcs$. However, a major drawback of this bound is that it gives nothing when $|\Sigma|$ is small (e.g., $|\Sigma| \leq \eps n$). For even smaller alphabet size, the bound does not even give anything for exact computation. For example, in the case of a binary alphabet, we know that $\lis(x) \leq 2$ and thus taking $\eps=1/2$ corresponds to exact computation. Yet the bound gives a negative number.

This is somewhat disappointing as in most applications of $\ed$ and $\lcs$, the alphabet size is actually a fixed constant. These include for example the English language and the  human DNA sequence (where the alphabet size is $4$ for the $4$ bases). Therefore, in this paper we focus on the case where the alphabet size is small, and we have the following theorem.

\begin{theorem}\label{thm:LCS3}
	Given an alphabet $\Sigma$, for any $\eps>0$ where $\frac{|\Sigma|^2}{\eps} =O(n)$, any $R$-pass deterministic algorithm in the asymmetric streaming model that computes a $1+\eps$ approximation of $\lcs(x,y)$   for two strings $x, y \in \Sigma^n$ must use space $\Omega\left (\frac{|\Sigma|}{\eps}/R \right )$.
\end{theorem}

Thus, even for a binary alphabet, achieving $1+\eps$ approximation for small $\eps$ (e.g., $\eps=1/n$ which corresponds to exact computation) can take space as large as $\Omega(n)$ for any constant pass algorithm. Further note that by taking $|\Sigma|=\sqrt{\eps n}$, we recover the $\Omega\left (\frac{\sqrt{n}}{\eps}/R \right )$ bound with a much smaller alphabet.

Finally, we turn to $\lis$  and longest non-decreasing subsequence ($\lns$), as well as a natural generalization of $\lis$  and $\lns$  which we call \emph{longest non-decreasing subsequence with threshold} ($\lnst$). Given a string $x \in \Sigma^n$ and a threshold $t \leq n$, $\lnst(x, t)$ denotes the length of the longest non-decreasing subsequence in $x$ such that each symbol appears at most $t$ times. It is easy to see that the case of $t=1$ corresponds to $\lis$ and the case of $t=n$ corresponds to $\lns$. Thus $\lnst$ is indeed a generalization of both $\lis$  and $\lns$. It is also a special case of $\lcs$ when $|\Sigma|t \leq n$ as we can take $y$ to be the concatenation of $t$ copies of each symbol, in the ascending order (and possibly padding some symbols not in $x$). How hard is $\lnst$? We note that in the case of $t=1$ ($\lis$) and $t=n$ ($\lns$) a simple dynamic programming can solve the problem in one pass with space $O(|\Sigma| \log n)$, and $1+\eps$ approximation can be achieved in one pass with space $\tilde{O}(\sqrt{\frac{n}{\eps} })$ by \cite{gopalan2007estimating}. Thus one can ask what is the situation for other $t$. Again we focus on the case of a small alphabet and have the following theorem.

\begin{theorem}
	\label{thm:LNSTmain}
	Given an alphabet $\Sigma$, for deterministic $(1+\eps)$ approximation of $\lnst(x, t)$ for a string $x\in \Sigma^n$ in the streaming model with $R$ passes, we have the following space lower bounds:
	\begin{enumerate}
		
		\item $\Omega(\min(\sqrt{n}, |\Sigma|)/R)$ for any constant $t$ (this includes $\lis$), when $\eps$ is any constant.
		
		\item $\Omega(|\Sigma|\log(1/\eps)/R)$ for $t\geq n/|\Sigma|$ (this includes $\lns$), when $|\Sigma|^2/\eps = O(n)$.
		
		\item $\Omega\left (\frac{\sqrt{|\Sigma|}}{\eps}/R \right ) $ for $t = \Theta(1/\eps)$, when $|\Sigma|/\eps = O(n)$.
\end{enumerate}
\end{theorem}

Thus, case 1 and 2 show that even for any constant approximation, any constant pass streaming algorithm for $\lis$ and $\lns$ needs space $\Omega(|\Sigma|)$ when $|\Sigma| \leq \sqrt{n}$, matching the $O(|\Sigma| \log n)$ upper bound up to a logarithmic factor. Taking $\eps=1/\sqrt[3]{n}$ and $|\Sigma| \leq \sqrt[3]{n}$ for example, we further get a lower bound of $\Omega(|\Sigma|\log n)$ for approximating $\lns$ using any constant pass streaming algorithm. This matches the $O(|\Sigma| \log n)$ upper bound. These results complement the bounds in \cite{gal2010lower, ergun2008distance, gopalan2007estimating} for the important case of small alphabet, and together they provide an almost complete picture for $\lis$ and $\lns$. Case 3 shows that for certain choices of $t$ and $\eps$, the space we need for $\lnst$ can be significantly larger than those for $\lis$ and $\lns$. It is an intriguing question to completely characterize the behavior of $\lnst$ for all regimes of parameters.
	
We also give improved algorithms for asymmetric streaming $\ed$ and $\lcs$. For $\ed$, \cite{farhadi2020streaming, cheng2020space} gives a $O(2^{1/\delta})$-approximation algorithm with $\tilde{O}(n^\delta)$ space for any constant $\delta\in(0,1)$. We further reduced the space needed from $\tilde{O}(\frac{n^\delta}{\delta})$ to $O(\frac{d^\delta}{\delta}\;\polylog(n)) $ where $d=\ed(x, y)$. Specifically, we have the following theorem.

\begin{theorem}
	\label{thm:ED_Algo}
	Assume $\ed(x,y) = d$, in the asymmetric streaming model, there are one-pass deterministic algorithms in polynomial time with the following parameters:
	\begin{enumerate}
		\item A $(3+\eps)$-approximation of $\ed(x,y)$ using $O(\sqrt{d}\; \polylog(n))$ space.
		\item For any constant $\delta\in (0,1/2)$, a $2^{O(\frac{1}{\delta})}$-approximation of $\ed(x,y)$ using $O(\frac{d^\delta}{\delta} \polylog(n))$ space.
	\end{enumerate}
\end{theorem}

For $\lcs$ over a large alphabet, the upper bounds in \cite{farhadi2020streaming, cheng2020space} match the lower bounds implied by \cite{gal2010lower, ergun2008distance}. We thus again focus on small alphabet. Note that our Theorem~\ref{thm:LCS3} does not give anything useful if $|\Sigma|$ is small and $\eps$ is large (e.g., both are constants). Thus a natural question is whether one can get better bounds. In particular, is the dependence on $1/\eps$ linear as in our theorem, or is there a threshold beyond which the space jumps to say for example $\Omega(n)$? We note that there is a trivial one pass, $O(\log n)$ space algorithm even in the standard streaming model that gives a $|\Sigma|$ approximation of $\lcs$ (or $1/|\Sigma|$ approximation in standard notation), and no better approximation using sublinear space is known even in the asymmetric streaming model. Thus one may wonder whether this is the threshold. We show that this is not the case, by giving a one pass algorithm in the asymmetric streaming model over the binary alphabet that achieves a $2-\eps$ approximation of $\lcs$ (or $1/2+\eps$ approximation in standard notation), using space $n^{\delta}$ for any constant $\delta>0$. 

\begin{theorem}
	\label{thm:LCS_Algo}
	 For any constant $\delta\in(0,1/2)$, there exists a constant $\eps>0$ and a one-pass deterministic algorithm that outputs a $2-\eps$ approximation of $\lcs(x,y)$ for any two strings $x,y\in \{0,1\}^n$, with $\tilde{O}(n^\delta/\delta)$ space and polynomial time in the asymmetric streaming model.
\end{theorem}

Finally, as mentioned before, we now have an almost complete picture for $\lis$ and $\lns$, but for the more general $\lnst$ the situation is still far from clear. Since $\lnst$ is a special case of $\lcs$, if $|\Sigma|t=O(n)$ then the upper bound of $\tilde{O}(\frac{\sqrt{n}}{\eps})$ in \cite{farhadi2020streaming, cheng2020space} still applies and this matches our lower bound in case 3, Theorem~\ref{thm:LNSTmain} by taking $|\Sigma|=\eps n$. One can then ask the natural question of whether we can get a matching upper bound for the case of small alphabet. We are not able to achieve this, but we provide a simple algorithm that can use much smaller space for certain regimes of parameters in this case.

\begin{theorem}
	\label{thm:lnstupper}
	Given an alphabet $\Sigma$ with $\lvert \Sigma \rvert = r$.\ For any $\eps>0$ and $t \geq 1$, there is a one-pass streaming algorithm that computes a $(1+\eps)$ approximation of $\lnst(x, t)$ for any $x\in \Sigma^n$ with $\tilde{O}\Big(\big(\min(t, r/\eps )+1\big)^r\Big)$ space. 
\end{theorem}

\subsection{Overview of our Techniques}
Here we provide an informal overview of the techniques used in this paper. 

\subsubsection{Lower Bounds}
Our lower bounds use the general framework of communication complexity. To limit the power of random access to the string $y$, we always fix $y$ to be a specific string, and consider different strings $x$. In turn, we divide $x$ into several blocks and consider the two party/multi party communication complexity of $\ed(x, y)$ or $\lcs(x, y)$, where each party holds one block of $x$. However, we need to develop several new techniques to handle edit distance and small alphabets.

\subparagraph*{Edit distance.} We start with edit distance. One difficulty here is to handle substitutions, as with substitutions edit distance becomes similar to Hamming distance, and this is exactly one of the reasons why strong complexity results separating edit distance and Hamming distance are rare. Indeed, if we define $\ed(x, y)$ to be the smallest number of insertions and deletions (without substitutions) to transform $x$ into $y$, then $\ed(x, y)=2n-2\lcs(x, y)$ and thus a lower bound for exactly computing $\lcs$ (e.g., those implied from  \cite{gal2010lower, ergun2008distance}) would translate directly into the same bound for exactly computing $\ed$. On the other hand, with substitutions things become more complicated: if $\lcs(x, y)$ is small (e.g., $\lcs(x, y) \leq n/2$) then in many cases (such as examples obtained by reducing from  \cite{gal2010lower, ergun2008distance}) the best option to transform $x$ into $y$ is just replacing each symbol in $x$ by the corresponding symbol in $y$ if they are different, which makes $\ed(x, y)$ exactly the same as their Hamming distance.

To get around this, we need to ensure that $\lcs(x, y)$ is large. We demonstrate our ideas by first describing an $\Omega(n)$ lower bound for the deterministic two party communication complexity of $\ed(x, y)$, using a reduction from the equality function which is well known to have an $\Omega(n)$ communication complexity bound. Towards this, fix $\Sigma=[3n] \cup \{a\}$ where $a$ is a special symbol, and fix $y=1 \circ 2 \circ \cdots \circ 3n$. We divide $x$ into two parts $x=(x_1, x_2)$ such that $x_1$ is obtained from the string $(1, 2, 4, 5, \cdots, 3i-2, 3i-1, \cdots, 3n-2, 3n-1)$ by replacing some symbols of the form $3j-1$ by $a$, while $x_2$ is obtained from the string $(2, 3, 5, 6, \cdots, 3i-1, 3i, \cdots, 3n-1, 3n)$ by replacing some symbols of the form $3j-1$ by $a$. Note that the way we choose $(x_1, x_2)$ ensures that $\lcs(x, y) \geq 2n$ before replacing any symbol by $a$.

Intuitively, we want to argue that the best way to transform $x$ into $y$, is to delete a substring at the end of $x_1$ and a substring at the beginning of $x_2$, so that the resulted string becomes an increasing subsequence as long as possible. Then, we insert symbols into this string to make it match $y$ except for those $a$ symbols. Finally, we replace the $a$ symbols by substitutions. If this is true then we can finish the argument as follows. Let $T_1, T_2 \subset [n]$ be two subsets with size $t=\Omega(n)$, where for any $i \in \{1, 2\}$, all symbols of the form $3j-1$ in $x_i$ with $j \in T_i$ are replaced by $a$. Now if $T_1=T_2$ then it doesn't matter where we choose to delete the substrings in $x_1$ and $x_2$, the number of edit operations is always $3n-2+t$ by a direct calculation. On the other hand if $T_1 \neq T_2$ and assume for simplicity that the smallest element they differ is an element in $T_2$, then there is a way to save one substitution, and the the number of edit operations becomes $3n-3+t$.

The key part is now proving our intuition.\ For this, we consider all possible $r \in [3n]$ such that $x_1$ is transformed into $y[1:r]$ and $x_2$ is transformed into $y[r+1:3n]$, and compute the two edit distances respectively. To analyze the edit distance, we first show by a greedy argument that without loss of generality, we can assume that we apply deletions first, followed by insertions, and substitutions at last. This reduces the edit distance problem to the following problem: for a fixed number of deletions and insertions, what is the best way to minimize the Hamming distance (or maximize the number of agreements of symbols at the same indices) in the end. Now we break the analysis of $\ed(x_1, y[1:r])$ into two cases. Case 1 is where the number of deletions (say $d_d$) is large. In this case, the number of insertions (say $d_i$) must also be large, and we argue that the number of agreements is at most $\lcs(x_1, y[1:r])+d_i$. Case 2 is where $d_d$ is small. In this case, $d_i$ must also be small. Now we crucially use the structure of $x_1$ and $y$, and argue that symbols in $x_1$ larger than $3d_i$ (or original index beyond $2d_i$) are guaranteed to be out of agreement. Thus the number of agreements is at most $\lcs(x_1[1: 2d_i], y[1:r])+d_i$. In each case combining the bounds gives us a lower bound on the total number of operations. The situation for $x_2$ and $y[r+1:3n]$ is completely symmetric and this proves our intuition.

In the above construction, $x$ and $y$ have different lengths ($|x|=4n$ while $|y|=3n$). We can fix this by adding a long enough string $z$ with distinct symbols than those in $\{x, y\}$ to the end of both $x$ and $y$, and then add $n$ symbols of $a$ at the end of $z$ for $y$. We argue that the best way to do the transformation is to transform $x$ into $y$, and then insert $n$ symbols of $a$. To show this, we first argue that at least one symbol in $z$ must be kept, for otherwise the number of operations is already larger than the previous transformation. Then, using a greedy argument we show that the entire $z$ must be kept, and thus the natural transformation is the optimal.

To extend the bound to randomized algorithms, we modify the above construction and reduce from \emph{Set Disjointness} ($\dis$), which is known to have randomized communication complexity $\Omega(n)$. Given two strings $\alpha, \beta \in \{0, 1\}^n$ representing the characteristic vectors of two sets $A, B \subseteq [n]$, $\dis(\alpha, \beta)=0$ if and only if $A \cap B \neq \emptyset$, or equivalently, $\exists j \in [n], \alpha_j= \beta_j=1$. For the reduction, we first create two new strings $\alpha' , \beta' \in \{0, 1\}^{2n}$ which are ``balanced" versions of $\alpha, \beta$. Formally, $\forall j \in [n], \alpha'_{2j-1}=\alpha_j$ and $ \alpha'_{2j}=1-\alpha_j$. We create $ \beta' $ slightly differently, i.e., $\forall j \in [n], \beta'_{2j-1}=1-\beta_j$ and $ \beta'_{2j}=\beta_j$. Now both $\alpha'$ and $\beta'$ have $n$ $1$'s, we can use them as  the characteristic vectors of the two sets $T_1, T_2$ in the previous construction. A similar argument now leads to the bound for randomized algorithms.

\subparagraph*{Longest common subsequence.} Our lower bounds for randomized algorithms computing $\lcs$ exactly are obtained by a similar and simpler reduction from $\dis$: we still fix $y$ to be an increasing sequence of length $8n$ and divide $y$ evenly into $4n$ blocks of constant size. Now $x_1$ consists of the blocks with an odd index, while $x_2$ consists of the blocks with an even index. Thus $x$ is a permutation of $y$. Next, from $\alpha, \beta \in \{0, 1\}^n$ we create $\alpha' , \beta' \in \{0, 1\}^{2n}$ in a slightly different way and use $\alpha' , \beta'$ to modify the $2n$ blocks in $x_1$ and $x_2$ respectively. If a bit is $1$ then we arrange the corresponding block in the increasing order, otherwise we arrange the corresponding block in the decreasing order. A similar argument as before now gives the desired $\Omega(n)$ bound. We note that \cite{sun2007communication} has similar results for LIS by reducing from $\dis$. However, our reduction and analysis are different from theirs. Thus we can handle $\lcs$, and even the harder case where $x$ is a permutation of $y$.

We now turn to $\lcs$ over a small alphabet. To illustrate our ideas, let's first consider $\Sigma=\{0, 1\}$ and choose $y=0^{n/2}1^{n/2}$. It is easy to see that $\lcs(x, y)=\lnst(x, n/2)$. We now represent each string $x \in \{0, 1\}^n$ as  follows: at any index $i \in [n] \cup \{0\}$, we record a pair $(p, q)$ where $p=\mathsf{min}(\text{the number of 0's in } x[1:i], n/2)$ and $q=\mathsf{min}(\text{the number of 1's in } x[i+1: n], n/2)$. Thus, if we read $x$ from left to right, then upon reading a $0$, $p$ may increase by $1$ and $q$ does not change; while upon reading a $1$, $p$ does not change and $q$ may decrease by $1$. Hence if we use the horizontal axis to stand for $p$ and the vertical axis to stand for $q$, then these points $(p, q)$ form a polygonal chain. We call $p+q$ the \emph{value} at point $(p, q)$ and it is easy to see that $\lcs(x, y)$ must be the value of an endpoint of some chain segment.

Using the above representation, we now fix $\Sigma=\{0, 1, 2\}$ and choose $y=0^{n/3}1^{n/3}2^{n/3}$, so $\lcs(x, y)=\lnst(x, n/3)$. We let $x=(x_1, x_2)$ such that $x_1 \in \{0, 1\}^{n/2}$ and $x_2 \in \{1, 2\}^{n/2}$. Since any common subsequence between $x$ and $y$ must be of the form $0^a 1^b 2^c$ it suffices to consider common subsequence between $x_1$ and $0^{n/3}1^{n/3}$, and that between $x_2$ and $1^{n/3}2^{n/3}$, and combine them together. Towards that, we impose the following properties on $x_1, x_2$: (1) The number of $0$'s,  $1$'s, and $2$'s in each string is at most $n/3$; (2) In the polygonal chain representation of each string, the values of the endpoints strictly increase when the number of $1$'s increases; and (3) For any endpoint in $x_1$ where the number of $1$'s is some $r$, there is a corresponding endpoint in $x_2$ where the number of $1$'s is $n/3-r$, and the values of these two endpoints sum up to a fixed number $t=\Omega(n)$. Note that property (2) implies that $\lcs(x, y)$ must be the sum of the values of an endpoint in $x_1$ where the number of $1$'s is some $r$, and an endpoint in $x_2$ where the number of $1$'s is $n/3-r$, while property (3) implies that for any string $x_1$, there is a unique corresponding string $x_2$, and $\lcs(x, y)=t$ (regardless of the choice of $r$). 

We show that under these properties, all possible strings $x=(x_1, x_2)$ form a set $S$ with $|S|=2^{\Omega(n)}$, and this set gives a \emph{fooling set} for the two party communication problem of computing $\lcs(x, y)$. Indeed, for any $x=(x_1, x_2) \in S$, we have $\lcs(x, y)=t$. On the other hand, for any $(x_1, x_2) \neq (x'_1, x'_2) \in S$, the values must differ at some point for $x_1$ and $x'_1$. Hence by switching, either $(x_1, x'_2)$ or $(x'_1, x_2)$ will have a $\lcs$ with $y$ that has length at least $t+1$.\ Standard arguments now imply an $\Omega(n)$ communication complexity lower bound.\ A more careful analysis shows that we can even replace the symbol $2$ by $0$, thus resulting in a binary alphabet.

The above argument can be easily modified to give a $\Omega(1/\eps)$ bound for $1+\eps$ approximation of $\lcs$ when $\eps<1$, by taking the string length to be some $n'=\Theta(1/\eps)$. To get a better bound, we combine our technique with the technique in \cite{ergun2008distance} and consider the following direct sum problem: we create $r$ copies of strings $\{x^i, i \in [r]\}$ and $\{y^i, i \in [r]\}$ where each copy uses distinct alphabets with size $2$. Assume for $x^i$ and $y^i$ the alphabet is $\{a_i, b_i\}$, now $x^i$ again consists of $r$ copies of $(x^i_{j1}, x^i_{j2}), j  \in [r]$, where each $x^i_{j\ell} \in \{a_i, b_i\}^{n'/2}$ for $\ell \in [2]$; while $y^i$ consists of $r$ copies $y^i_j=a_i^{n'/3}b_i^{n'/3}a_i^{n'/3}, j \in [r]$. The direct sum problem is to decide between the following two cases for some $t=\Omega(n')$: (1) $\exists i$ such that there are $\Omega(r)$ copies $(x^i_{j1}, x^i_{j2})$ in  $x^i$ with $\lcs((x^i_{j1} \circ x^i_{j2}), y^i_j) \geq t+1$, and (2) $\forall i$ and $\forall j$,  $\lcs((x^i_{j1} \circ x^i_{j2}), y^i_j) \leq t$. We do this by arranging the $x^i$'s row by row into an $r \times 2r$ matrix (each entry is a length $n'/2$ string) and letting $x$ be the concatenation of the \emph{columns}. We call these strings the \emph{contents} of the matrix, and let $y$ be the concatenation of the $y^i$'s. Now intuitively, case (1) and case (2) correspond to deciding whether $\lcs(x, y) \geq 2rt+\Omega(r) $ or $\lcs(x, y) \leq 2rt $, which implies a $1+\Omega(1/t)=1+\eps$ approximation.\ The lower bound follows by analyzing the $2r$-party communication complexity of this problem, where each party holds a column  of the matrix.

However, unlike the constructions in \cite{gal2010lower, ergun2008distance} which are relatively easy to analyze because all symbols in $x$ (respectively $y$) are \emph{distinct}, the repeated symbols in our construction make the analysis of $\lcs$ much more complicated (we can also use distinct symbols but that will only give us a bound of $\frac{\sqrt{|\Sigma|}}{\eps}$ instead of $\frac{|\Sigma|}{\eps}$).\ To ensure that the $\lcs$ is to match each $(x^i_{j1}, x^i_{j2})$ to the corresponding $y^i_j$, we use another $r$ symbols $\{c_i, i \in [r]\}$ and add \emph{buffers} of large size (e.g., size $n'$) between adjacent copies of $(x^i_{j1}, x^i_{j2})$. We do the same thing for $y^i_j$ correspondingly. Moreover, it turns out we need to arrange the buffers carefully to avoid unwanted issues: in each row $x^i$, between each copy of $(x^i_{j1}, x^i_{j2})$ we use a buffer of new symbol. Thus the buffers added to each row  $x^i$ are $c_1^{n'}, c_2^{n'}, \cdots, c_r^{n'}$ sequentially and this is the same for every row. That is, in each row the contents use the same alphabet $\{a_i, b_i\}$ but the buffers use different alphabets $\{c_i, i \in [r]\}$. Now we have a $r \times 3r$ matrix and we again let $x$ be the concatenation of the \emph{columns} while let $y$ be the concatenation of the $y^i$'s. Note that we are using an alphabet of size $|\Sigma|=3r$.\ We use a careful analysis to argue that case (1) and case (2) now correspond to deciding whether $\lcs(x, y) \geq 2rn'+rt+\Omega(r) $ or $\lcs(x, y) \leq 2rn'+rt $, which implies a $1+\eps$ approximation.\ The lower bound follows by analyzing the $3r$-party communication complexity of this problem, and we show a lower bound of $\Omega(r/\eps)=\Omega(|\Sigma|/\eps)$ by generalizing our previous fooling set construction to the multi-party case, where we use a good error correcting code to create the $\Omega(r)$ gap. 

The above technique works for $\eps<1$. For the case of $\eps \geq 1$ our bound for $\lcs$ can be derived directly from our bound for $\lis$, which we describe next.


\subparagraph*{Longest increasing/non-decreasing subsequence.} Our $\Omega(|\Sigma|)$ lower bound over small alphabet is achieved by modifying the construction in \cite{ergun2008distance} and providing a better analysis. Similar as before, we consider a matrix $B\in \{0,1\}^{\frac{r}{c}\times r}$ where $c$ is a large constant and $r=|\Sigma|$. We now consider the $r$-party communication problem where each party holds one column of $ B$, and the problem is to decide between the following two cases for a large enough constant $l$: (1) for each row in $B$, there are at least $l$ $0$'s between any two $1$'s, and (2) there exists a row in $B$ which has more than $\alpha r$ $1$'s, where $\alpha \in (1/2,1)$ is a constant. We can use a similar argument as in  \cite{ergun2008distance} to show that the total communication complexity of this problem is $\Omega(r^2)$ and hence at least one party needs $\Omega(r)$. The difference is that \cite{ergun2008distance} sets $l=1$ while we need to pick $l$ to be a larger constant to handle the case $\eps \geq 1$. For this we use the Lov\'asz Local Lemma with a probabilistic argument to show the existence of a large fooling set. To reduce to $\lis$, we define another matrix $\tB$ such that $\tB_{i,j} = (i-1)\frac{r}{c}+j$  if $B_{i,j} = 1$ and $\tB_{i,j} = 0$ otherwise. Now let $x$ be the concatenation of all columns of $\tB$. We show that case (2) implies $\lis(x)\geq \alpha r$ and case (1) implies $\lis(x)\leq (1/c+1/l)r$. This implies a $1+\eps$ approximation for any constant $\eps > 0$ by setting $c$ and $l$ appropriately.

The construction is slightly different for $\lns$. This is because if we keep the $0$'s in $\tB$, they will already form a very long non-decreasing subsequence and we will not get any gap. Thus, we now let the matrix $B$ have size $r\times cr$ where $c$ can be any constant. We replace all $0$'s in column $i$ with a symbol $b_i$ for $i \in [cr]$, such that $b_1>b_2>\cdots >b_{cr}$. Similarly we replace all $1$'s in row $j$ with a symbol $a_{j}$ for $j \in [r]$, such that $a_1<a_2<\cdots<a_r$. Also, we let $a_1>b_1$. We can show that the two cases now correspond to $\lns(x)> \alpha c r$ and $\lns(x)\leq (2+c/l)r$.

We further prove an $\Omega(|\Sigma|\log(1/\eps))$ lower bound for $1+\eps$ approximation of $\lns$ when $\eps< 1$. This is similar to our previous construction for $\lcs$, except we don't need buffers here, and we only need to record the number of some symbols. More specifically, let $l = \Theta(1/\eps)$ and $S$ be the set of all strings $x = (x_1,x_2)$ over alphabet $\{a,b\}$ with length $2l$ such that $x_1 = a^{\frac{3}{4}l+t}b^{\frac{1}{4}l-t}$ and $x_2 = a^{\frac{3}{4}l-t}b^{\frac{1}{4}l+t}$ for any $t \in [\frac{l}{4}]$. Thus $S$ has size $\frac{l}{4}=\Omega(1/\eps)$ and $\forall x \in S$, the number of $a$'s in $x$ is exactly $\frac{3}{2}l$. Further, for any $(x_1,x_2)\neq (x'_1, x'_2)\in S$, either $(x_1,x'_2)$ or $(x'_1, x_2)$ has more than $\frac{3}{2}l$ $a$'s. We now consider the $r \times 2r$ matrix where each row $i$ consists of $\{(x^i_{j1}, x^i_{j2}), j \in [r]\}$ such that each $x^i_{j{\ell}}$ has length $l$ for $\ell \in [2]$, and for the same row $i$ all  $\{(x^i_{j1}, x^i_{j2})\}$ use the same alphabet $\{a_i, b_i\}$ while for different rows the alphabets are disjoint. To make sure the $\lns$ of the concatenation of the columns is roughly the sum of the number of $a_i$'s, we require that $b_r<b_{r-1}<\cdots<b_1<a_1<a_2<\cdots<a_r$. Now we analyze the $2r$ party communication problem of deciding whether the concatenation of the columns has $\lns \geq c rl +\Omega(r)$ or $\lns \leq c rl $ for some constant $c$, which implies a $1+\eps$ approximation. The lower bound is again achieved by generalizing the set $S$ to a fooling set for the $2r$ party communication problem using an error correcting code based approach. 




In Theorem~\ref{thm:LNSTmain}, we give three lower bounds for $\lnst$. The first two lower bounds are adapted from our lower bounds for $\lis$ and $\lns$, while the last lower bound is adapted from our lower bound for $\lcs$ by ensuring all symbols in different rows or columns of the matrix there are different.  

\subsubsection{Improved Algorithms}

We now turn to our improved algorithms for $\ed$ and $\lcs$ in the asymmetric streaming model. 



\subparagraph*{Edit distance.} Our algorithm for edit distance builds on and improves the algorithm in \cite{farhadi2020streaming, cheng2020space}. The key idea of that algorithm is to use triangle inequality. Given a constant $\delta$, the algorithm first divides $x$ evenly into $b = n^\delta$ blocks. Then for each block $x^i$ of $x$, the algorithm recursively finds an $\alpha$-approximation of the closest substring to $x^i$ in $y$. That is, the algorithm finds a substring $y[l_i:r_i]$ and a value $d_i$ such that for any substring $y[l:r]$ of $y$, $\ed(x^i, y{[l_i:r_i]})\leq d_i \leq \alpha\ed(x^i, y{[l:r]})$. Let $\tilde{y}$ be the concatenation of $y[l_i:r_i]$ from $i = 1$ to $b$. Then using triangle inequality, \cite{farhadi2020streaming} showed that $\ed(y,\tilde{y})+\sum_{i = 1}^b d_i$ is a $2\alpha+1$ approximation of $\ed(x,y)$. The $\tilde{O}(n^\delta)$ space is achieved by recursively applying this idea, which results in a $O(2^{1/\delta})$ approximation.

To further reduce the space complexity, our key observation is that, instead of dividing $x$ into blocks of equal length, we can divide it according to the positions of the edit operations that transform $x$ to $y$. More specifically, assume we are given a value $k$ with $\ed(x,y)\leq k \leq c\ed(x,y)$ for some constant $c$, we show how to design an approximation algorithm using space $\tilde{O}(\sqrt{k})$. Towards this, we can divide $x$ and $y$ each into $\sqrt{k}$ blocks $x = x^1\circ \cdots\circ x^{\sqrt{k}}$ and $y = y^1\circ \cdots \circ y^{\sqrt{k}}$ such that $\ed(x^i, y^i)\leq \frac{\ed(x,y)}{\sqrt{k}} \leq \sqrt{k}$ for any $i\in[\sqrt{k}]$. However, such a partition of $x$ and $y$ is not known to us. Instead, we start from the first position of $x$ and find the largest index $l_1$ such that $\ed(x[1:l_1],y[p_1,q_1])\leq \sqrt{k}$ for some substring $y[p_1:q_1]$ of $y$. To do this, we start with $l = \sqrt{k}$ and try all substrings of $y$ with length in $[l-\sqrt{k},l+\sqrt{k}]$. If there is some substring of $y$ within edit distance $\sqrt{k}$ to $x[1:l]$, we set $l_1 = l$ and store all the edit operations that transform $y[p_1:q_1]$ to $x[1:l_1]$ where $y[p_1:q_1]$ is the substring closest to $x[1:l_1]$ in edit distance. We continue doing this with $l = l+1$ until we can not find a substring of $y$ within edit distance $\sqrt{k}$ to $x[1:l]$. 

One problem here is that $l$ can be much larger than $\sqrt{k}$ and we cannot store $x[1:l]$ with $\tilde{O}(\sqrt{k})$ space.\ However, since we have stored some substring $y[p_1:q_1]$ (we only need to store the two indices $p_1, q_1$) and the  at most $\sqrt{k}$ edit operations that transform $y[p_1:q_1]$ to $x[1:l-1]$, we can still query every bit of $x[1:l]$ using $\tilde{O}(\sqrt{k})$ space. 

After we find the largest possible index $l_1$, we store $l_1$, $(p_1, q_1)$ and $d_1 = \ed(x[1:l_1], y[p_1:q_1])$. We then start from the $(l_1+1)$-th position of $x$ and do the same thing again to find the largest $l_2$ such that there is a substring of $y$ within edit distance $\sqrt{k}$ to $x[l_1+1:l_1+l_2]$.\ We continue doing this until we have processed the entire string $x$. Assume this gives us $T$  pairs of indices $(p_i,q_i)$ and integers $l_i$, $d_i$ from $i=1$ to $T$, we can use $O(T\log n)$ space to store them. We show by induction that  $x^1\circ\cdots\circ x^i$ is a substring of $x[1:\sum_{j = 1}^i l_j]$ for $i\in[T-1]$. Recall that $x = x^1\circ \cdots\circ x^{\sqrt{k}}$ and each $l_i>0, i \in [T-1]$. Thus, the process must end within $\sqrt{k}$ steps and we have $T\leq \sqrt{k}$. Then, let $\tilde{y}$ be the concatenation of $y[p_i:q_i]$ from $i=1$ to $T$. Using techniques developed in \cite{farhadi2020streaming}, we can show $\ed(y,\tilde{y})+\sum_{i = 1}^T d_i$ is a $3$ approximation of $\ed(x,y)$. For any small constant $\eps>0$, we can compute a $1+\eps$ approximation of $\ed(y,\tilde{y})$ with $\polylog(n)$ space using the algorithm in \cite{cheng2020space}. This gives us a $3+\eps$ approximation algorithm with $O(\sqrt{\ed(x,y)}\;\polylog(n))$ space. 


Similar to \cite{farhadi2020streaming}, we can use recursion to further reduce the space. Let $\delta$ be a small constant and a value $k= \Theta(\ed(x,y))$ be given as before. There is a way to partition $x$ and $y$ each into $k^\delta$ blocks such that $\ed(x^i, y^i)\leq \frac{\ed(x,y)}{k^\delta}\leq k^{1-\delta}$.\ Now similarly, we want to find the largest index $l^0$ such that there is a substring of $y$ within edit distance $k^{1-\delta}$ to $x[1:l^0]$. However naively this would require $\Theta(k^{1-\delta})$ space to compute the edit distance.\ Thus again we turn to approximation. 

We introduce a recursive algorithm called $\mathsf{FindLongestSubstring}$. It takes two additional parameters as inputs: an integer $u$ and a parameter $s$ for the amount of space we can use. It outputs a three tuple: an index $l$, a pair of indices $(p,q)$ and an integer $d$. Let $l^0$ be the largest index such that there is a substring of $y$ within edit distance $u$ to $x[1:l^0]$. 

We show the following two properties of $\mathsf{FindLongestSubstring}$: (1) $l\geq l^0$, and (2) for any substring $y[p^*:q^*]$, $\ed(x[1:l], y[p:q])\leq d \leq c(u,s) \ed(x[1:l], y[p^*:q^*])$.\ Here, $c(u,s) $ is a function of $(u, s)$ that measures the approximation factor. If $u\leq s$, $\mathsf{FindLongestSubstring}$ outputs $l=l^0$ and the substring of $y$ that is closest to $x[1:l]$ using $O(s\log n)$ space by doing exact computation. In this case we set $c(u,s) = 1$. Otherwise, it calls $\mathsf{FindLongestSubstring}$ itself up to $s$ times with parameters $ u/s$ and $s$. This gives us $T\leq s$ outputs $\{l_i, (p_i,q_i), d_i\}$ for $i\in[T]$. Let $\tilde{y}$ be the concatenation of $y[p_i:q_i]$ for $i = 1$ to $T$. We find the pair of indices $(p,q)$ such that $y[p:q]$ is the substring that minimizes $\ed(\tilde{y}, y[p:q])$. We output $l = \sum_{j = 1}^{T}l_j$, $(p,q)$, and $d = \ed(\tilde{y}, y[p:q])+\sum_{i = 1}^T d_i$. We then use induction to show property (1) and (2) hold for these outputs, where $c(u,s) = 2(c(u/s,s)+1)$ if $u>s$ and $c(u,s) = 1$ if $u\leq s$. Thus we have $c(u,s) = 2^{O(\log_s{u})}$.

This gives an $O(k^{\delta}/\delta \; \polylog(n))$ space algorithm as follows. We run algorithm $\mathsf{FindLongestSubstring}$ with $u = k^{1-\delta}$ and $s = k^\delta$ to find $T$ tuples: $\{l_i, (p_i,q_i), d_i\}$. Again, let $\tilde{y}$ be the concatenation of $y[p_i:q_i]$ from $i = 1$ to $T$. Similar to the $O(\sqrt{k}\; \polylog(n))$ space algorithm, we can show $T\leq k^\delta$ and $\ed(y,\tilde{y})+\sum_{i = 1}^{T}d_i$ is a $2c(k^{1-\delta},k^{\delta}) + 1 = 2^{O(1/\delta)}$ approximation of $\ed(x,y)$. Since the depth of recursion is at most $1/\delta$ and each level of recursion needs $O(k^{\delta}\; \polylog(n))$ space, $\mathsf{FindLongestSubstring}$ uses $O(k^{\delta}/\delta \;\polylog(n))$ space. 

The two algorithms above both require a given value $k$. To remove this constraint, our observation is that the two previous algorithms actually only need the number $k$ to satisfy the following relaxed condition: there is a partition of $x$ into $k^\delta$ blocks such that for each block $x^i$, there is a substring of $y$ within edit distance $k^{1-\delta}$ to $x^i$. Thus, when such a $k$ is not given, we can do the following. We first set $k$ to be a large constant $k_0$. While the algorithm reads $x$ from left to right, let $T'$ be the number of $\{l_i, (p_i,q_i),d_i\}$ we have stored so far. Each time we run $\mathsf{FindLongestSubstring}$ at this level, we increase $T'$ by $1$. If the current $k$ satisfies the relaxed condition, then by a similar argument as before $T'$ should never exceed $ k^\delta$.\ Thus whenever $T'  = k^\delta$, we increase $k$ by a $2^{1/\delta}$ factor. Assume that $k$ is updated $m$ times in total and after the $i$-th update, $k$ becomes $k_i$. We show that $k_m  = O(\ed(x,y))$ (but $k_m$ may be much smaller than $\ed(x, y)$).\ To see this, suppose $k_j > 2^{1/\delta}\ed(x,y)$ for some $j\leq m$. Let $t_j$ be  the position of $x$ where $k_{j-1}$ is updated to $k_j$. We know it is possible to divide $x[t_j:n]$ into $\ed(x,y)^\delta$ blocks such that for each part, there is a substring of $y$ within edit distance $ \ed(x,y)^{1-\delta} \leq k_j^{1-\delta}$ to it. By property (1) and a similar argument as before,  we will run $\mathsf{FindLongestSubstring}$ at most $\ed(x,y)^\delta$ times until we reach the end of $x$. Since $k_j^\delta-k^\delta_{j-1} > \ed(x,y)^\delta$, $T'$ must be always smaller than $k^\delta_j$ and hence $k_j$ will not be updated.\ Therefore we must have $j = m$. This shows $k_{m-1}\leq 2^{1/\delta}\ed(x,y)$ and $k_m \leq  2^{2/\delta}\ed(x,y)$. Running $\mathsf{FindLongestSubstring}$ with $k \leq k_m$ takes $O(k_m^{\delta}/\delta\; \polylog(n))= O(\ed(x,y)^\delta/\delta\;\polylog(n))$ space and the number of intermediate results ($(p_i,q_i)$ and $d_i$'s) is $O(k_m^\delta) = O(\ed(x,y)^\delta)$. This gives us a $2^{O(1/\delta)}$ approximation algorithm with space complexity $ O(\ed(x,y)^\delta/\delta \; \polylog(n))$. 

\subparagraph*{$\lcs$ and $\lnst$.} We show that the reduction from $\lcs$ to $\ed$ discovered in $\cite{rubinstein2020reducing}$ can work in the asymmetric streaming model with a slight modification. Combined with our algorithm for $\ed$, this gives a $n^{\delta}$ space algorithm for $\lcs$ that achieves a $1/2+\eps$ approximation for binary strings. We stress that the original reduction in $\cite{rubinstein2020reducing}$ is not in the asymmetric streaming model and hence our result does not follow directly from previous works.

We also provide a simple algorithm to approximate $\lnst(x,t)$. Assume the alphabet is $[r]$, the idea is to create a set $\mathcal{D}\subset [t]^r$ and for each $d = (d_1,\dots, d_r)\in \mathcal{D}$, we check whether $\sigma(d) = 1^{d_1}2^{d_2}\cdots r^{d_r}$ is a subsequence of $x$. We show that in a properly chosen $\mathcal{D}$ with size $O\big((\min(t,r/\eps))^r\big)$ obtained by sparsifying $[t]^r$, there is some $d\in \mathcal{D}$  such that $\sigma(d)$ is a subsequence of $x$ and $|\sigma(d)|$ is a $1-\eps$ approximation of $\lnst(x,t).$  

\nopagebreak[4]
\subsection{Open Problems} Our work leaves a plethora of intriguing open problems. The main one is to close the gap between our lower bounds and the upper bounds of known algorithms, especially for the case of small alphabets and large (say constant) approximation. We believe that in this case it is possible to improve both the lower bounds and the upper bounds. Another interesting problem is to  completely characterize the space complexity of $\lnst$.

\section{Preliminaries}
\label{prelim}

We use the following conventional notations. Let $x \in \Sigma^n$ be a string of length $n$ over alphabet $\Sigma$. By $|x|$, we mean the length of $x$. We denote the $i$-th character of $x$ by $x_i$ and the substring from the $i$-th character to the $j$-th character by $x{[i:j]}$. We denote the concatenation of two strings $x$ and $y$ by $x\circ y$. By $[n]$, we mean the set of positive integers no larger than $n$. 

\textbf{Edit Distance} The \emph{edit distance} (or \emph{Levenshtein distance}) between two strings $x,y\in \Sigma^*$ , denoted by $\ed (x,y)$, is the smallest number of edit operations (insertion, deletion, and substitution) needed to transform one into another. The insertion (deletion) operation adds (removes) a character at some position. The substitution operation replace a character with another character from the alphabet set $\Sigma$.

\textbf{Longest Common Subsequence} We say the string $s\in \Sigma^t$ is a \emph{subsequence} of $x\in \Sigma^n$ if there exists indices $1\leq i_1<i_2<\cdots < i_t\leq n$ such that $s = x_{i_1}x_{i_2}\cdots x_{i_t}$.  A string $s$ is called a \emph{common subsequence} of strings $x$ and $y$ if $s$ is a subsequence of both $x$ and $y$. Given two strings $x$ and $y$, we denote the length of the longest common subsequence (LCS) of $x$ and $y$ by $\lcs(x,y)$.

In the proofs, we sometime consider the matching  between $x,y\in \Sigma^n$. By a matching, we mean a function $m:[n]\rightarrow [n]\cup \emptyset$ such that if $m(i)\neq \emptyset$, we have $x_i = y_{m(i)}$. We require the matching to be non-crossing. That is, for $i< j$, if $m(i)$ and $m(j)$ are both not $\emptyset$, we have $m(i)< m(j)$. The size of a matching is the number of $i\in [n]$ such that $m(i) \neq \emptyset$. We say a matching is a best matching if it achieves the maximum size. Each matching between $x$ and $y$ corresponds to a common subsequence. Thus, the size of a best matching between $x$ and $y$  is equal to $\lcs(x,y)$.

\textbf{Longest Increasing Subsequence} In the longest increasing subsequence problem, we assume there is a given total order on the alphabet set $\Sigma$. We say the string $s\in \Sigma^t$ is an \emph{increasing subsequence} of $x\in \Sigma^n$ if there exists indices $1\leq i_1<i_2<\cdots < i_t\leq n$ such that $s = x_{i_1}x_{i_2}\cdots x_{i_t}$ and $x_{i_1} < x_{i_2} < \cdots < x_{i_t}$. We denote the length of the longest increasing subsequence (LIS) of string $x$ by $\lis(x)$. In our discussion, we let $\infty$ and $-\infty$ to be two imaginary characters such that $-\infty<\alpha<\infty$ for all $\alpha\in \Sigma$.

\textbf{Longest Non-decreasing Subsequence} The longest non-decreasing subsequence is a variant of the longest increasing problem. The difference is that in a non-decreasing subsequence $s= x_{i_1}x_{i_2}\cdots x_{i_t}$, we only require $x_{i_1} \leq x_{i_2} \leq \cdots \leq  x_{i_t}$.

\begin{lemma}[Lov\'asz Local Lemma]
	\label{lem:LLL}
	Let $\mathcal{A} = \{A_1, A_2, \dots, A_n\}$ be a finite set of events. For $A\in \mathcal{A}$, let $\Gamma(A)$ denote the neighbours of $A$ in the dependency graph (In the dependency graph, mutually independent events are not adjacent). If there exist an assignment of reals $x: \mathcal{A}\rightarrow [0,1)$ to the events such that
	\begin{equation*}
		\forall \; A\in \mathcal{A}, \; \Pr(A)\leq x(A)\prod_{A'\in \Gamma(A)}(1-x(A')).
	\end{equation*}
	Then, for the probability that none of the events in $\mathcal{A}$ happens, we have 
	\begin{equation*}
		\Pr(\overline{A_1}\wedge \overline{A_2} \wedge \cdots \wedge \overline{A_{t-l}})\geq \prod_{i = 1}^n(1-x(A_i)).
	\end{equation*}
\end{lemma}

In the  \textbf{Set Disjointness} problem, we consider a two party game. Each party holds a binary string of length $n$, say $x$ and $y$. And the goal is to compute the function $\dis(x,y)$ as defined below

\begin{equation*}
	\dis(x,y) = \begin{cases}
		0, \text{ if } \exists \; i, \text{ s.t. } x_i = y_i \\
		1, \text{ otherwise}
	\end{cases}
\end{equation*}

We define $R^{1/3}(f)$ as the minimum number of bits required to be sent between two parties in any randomized multi-round communication protocol with 2-sided error at most $1/3$. The following is a well-known result.

\begin{lemma}[\cite{kalyanasundaram1992probabilistic}, \cite{razborov1990distributional}]
	\label{lem:dis}
	$R^{1/3}(\dis) = \Omega(n)$.
\end{lemma}

We will consider the one-way $t$-party communication model where $t$ players $P_1, P_2, \dots, P_t$ each holds input $x_1, x_2, \dots, x_t$ respectively. The goal is to compute the function $f(x_1, x_2, \dots, x_t)$. In the one-way communication model, each player speaks in turn and player $P_i$ can only send message to player $P_{i+1}$. We sometimes consider multiple round of communication. In an $R$ round protocol, during round $r\leq R$, each player speaks in turn $P_i$ sends message to $P_{i+1}$. At the end of round $r<R$, player $P_{t}$ sends a message to $P_1$. At the end of round $R$, player $P_{t}$ must output the answer of the protocol. We note that our lower bound also hold for a stronger \emph{blackboard} model. In this model, the players can write messages (in the order of $1,\dots, t$) on a blackboard that is visible to all other players. 

We define the \emph{total communication complexity} of $f$ in the $t$-party one-way communication model, denoted by $CC^{tot}_t(f)$, as the minimum number of bits required to be sent by the players in every deterministic communication protocol that always outputs a correct answer. The total communication complexity in the blackboard model is the total length of the messages written on the blackboard by the players.  

For deterministic protocol $P$ that always outputs the correct answer, we let $M(P)$ be the maximum number of bits required to be sent by some player in protocol $P$. We define $CC^{max}_t(f)$, the maximum communication complexity of $f$ as $\min_P{M(P)}$ where $P$ ranges over all deterministic protocol that outputs a correct answer.  We have $CC^{max}_t(f)\geq \frac{1}{tR}CC^{tot}_t(f)$ where $R$ is the number of rounds.

Let $X  $ be a subset of $U^t$ where $U$ is some finite universe and $t$ is an integer. Define the \textbf{span} of $X$ by $\mathsf{Span}(X) = \{y\in U^t| \forall\  i \in [t], \ \exists \ x \in X  \text{ s. t. } y_i = x_i\}$.  The notion $k$-fooling set introduced in \cite{ergun2008distance} is defined as following.

\begin{definition}[$k$-fooling set]
	Let $f: U^t\rightarrow \{0,1\}$ where $U$ is some finite universe. Let $S\subseteq U^t$. For some integer $k$, we say $S$ is a $k$-fooling set for $f$ iff $f(x) = 0$ for each $x\in S$ and for each subset $S'$ of $S$ with cardinality $k$, the span of $S'$ contains a member $y$ such that $f(y) = 1$.
	
\end{definition}

We have the following.

\begin{lemma}
	[Fact 4.1 from \cite{ergun2008distance}]
	\label{lem:k-fooling}
	Let $S$ be a $k$-fooling set for $f$, we have $CC^{tot}_t(f) \geq \log(\frac{|S|}{k-1})$.
\end{lemma}

\section{Lower Bounds for Edit Distance}
\label{sec:ed}

We show a reduction from the Set Disjointness problem ($\dis$) to computing $\ed$ between two strings in the asymmetric streaming model. For this, we define the following two party communication problem between Alice and Bob.

Given an alphabet $\Sigma$ and three integers $n_1, n_2, n_3$. Suppose Alice has a string $x_1 \in \Sigma^{n_1}$ and Bob has a string $x_2 \in \Sigma^{n_1}$. There is another fixed reference string $y \in \Sigma^{n_3}$ that is known to both Alice and Bob. Alice and Bob now tries to compute $\ed((x_1 \circ x_2), y)$.We call this problem $\ed_{cc}(y)$.
We prove the following theorem.

\begin{theorem}\label{thm:EDmain}
Suppose each input string to $\dis$ has length $n$ and let $\Sigma=[6n] \cup \{a\}$. Fix $y=(1, 2, \cdots, 6n)$. Then $R^{1/3}(\ed_{cc}(y)) \geq R^{1/3}(\dis)$.
\end{theorem}

To prove this theorem, we first construct the strings $x_1, x_2$ based on the inputs $\alpha, \beta \in \{0, 1\}^n$ to $\dis$. From $\alpha$, Alice constructs the string $\alpha' \in \{0, 1\}^{2n}$ such that $\forall j \in [n], \alpha'_{2j-1}=\alpha_j$ and $ \alpha'_{2j}=1-\alpha_j$. Similarly, from $\beta$, Bob constructs the string $\beta' \in \{0, 1\}^{2n}$ such that $\forall j \in [n], \beta'_{2j-1}=1-\beta_j$ and $ \beta'_{2j}=\beta_j$. Now Alice lets $x_1$ be a modification from the string $(1, 2, 4, 5, \cdots, 3i-2, 3i-1, \cdots, 6n-2, 6n-1)$ such that $\forall j \in [2n]$, if $\alpha'_j=0$ then the symbol $3j-1$ (at index $2j$) is replaced by $a$. Similarly, Bob lets $x_2$ be a modification from the string $(2, 3, 5, 6, \cdots, 3i-1, 3i, \cdots, 6n-1, 6n)$ such that $\forall j \in [2n]$, if $\beta'_j=0$ then the symbol $3j-1$ (at index $2j-1$) is replaced by $a$. 

We have the following Lemma.

\begin{lemma}\label{lem:EDmain}
If $\dis(\alpha, \beta)=1$ then $\ed((x_1 \circ x_2), y) \geq 7n-2$.
\end{lemma}

To prove the lemma we observe that in a series of edit operations that transforms $(x_1, x_2)$ to $y$, there exists an index $r \in [6n]$ s.t. $x_1$ is transformed into $[1:r]$ and $x_2$ is transformed into $[r+1:n]$. We analyze the edit distance in each part. We first have the following claim:

\begin{claim}\label{claim:substitution}
For any two strings $u$ and $v$, there is a sequence of optimal edit operations (insertion/deletion/substitution)  that transforms $u$ to $v$, where all deletions happen first, followed by all insertions, and all substitutions happen at the end of the operations.
\end{claim}

\begin{proof}
Note that a substitution does not change the indices of the symbols. Thus, in any sequence of such edit operations, consider the last substitution which happens at index l. If there are no insertion/deletion after it, then we are good. Otherwise, consider what happens if we switch this substitution and the insertion/deletions before this substitution and after the second to last substitution. The only symbol that may be affected is the symbol where index l is changed into. Thus, depending on this position, we may or may not need a substitution, which results in a sequence of edit operations where the number of operations is at most the original number. In this way, we can change all substitutions to the end.

Further notice that we can assume without loss of generality that any deletion only deletes the original symbols in $u$, because otherwise we are deleting an inserted symbol, and these two operations cancel each other. Therefore, in a sequence of optimal edit operations,  all the deletions can happen before any insertion.
\end{proof}

For any $i$, let $\Gamma_1(i)$ denote the number of $a$ symbols up to index $2i$ in $x_1$. Note that $\Gamma_1(i)$ is equal to the number of $0$'s in $\alpha'[1:i]$. We have the following lemma.

\begin{lemma}\label{lem:ED1}
For any $p \in [n]$, let $r=3p-q$ where $0 \leq q \leq 2$, then $\ed(x_1, [1:r]) = 4n-p-q+\Gamma_1(p)$ if $q=0, 1$ and $\ed(x_1, [1:r]) = 4n-p+\Gamma_1(p-1)$ if $q=2$.
\end{lemma}

\begin{proof}
By Claim~\ref{claim:substitution} we can first consider deletions and insertions, and then compute the Hamming distance after these operations (for substitutions). 

We consider the three different cases of $q$. Let the number of insertions be $d_i$ and the number of deletions be $d_d$. Note that $d_i-d_d=r-4n$. We define the number of agreements between two strings to be the number of positions where the two corresponding symbols are equal.

\paragraph{The case of $q=0$ and $q=1$.} Here again we have two cases.

\begin{description}
\item[Case (a):]
$d_d \geq 4n-2p.$ In this case, notice that the LCS after the operations between $x_1$ and $y$ is at most the original $\lcs(x_1, y)=2p-\Gamma_1(p)$. With $d_i$ insertions, the number of agreements can be at most $\lcs(x_1, y)+d_i=2p-\Gamma_1(p)+d_i$, thus the Hamming distance at the end is at least $r-2p+\Gamma_1(p)-d_i$. Therefore, in this case the number of edit operations is at least $d_i+d_d+r-2p+\Gamma_1(p)-d_i \geq 4n-p-q+\Gamma_1(p)$, and the equality is achieved when $d_d = 4n-2p$.

\item[Case (b):] $d_d < 4n-2p.$ In this case, notice that all original symbols in $x_1$ larger than $3d_i$ (or beyond index $2d_i$ before the insertions) are guaranteed to be out of agreement. Thus the only possible original symbols in $x_1$ that are in agreement with $y$ after the operations are the symbols with original index at most $2d_i$. Note that the LCS between $x_1[1: 2d_i]$ and $y$ is $2d_i-\Gamma_1(d_i)$. Thus with $d_i$ insertions the number of agreements is at most $3d_i-\Gamma_1(d_i)$, and the Hamming distance at the end is at least $r-3d_i+\Gamma_1(d_i)$.

Therefore the number of edit operations is at least $d_i+d_d+r-3d_i+\Gamma_1(d_i)=r-d_i+(d_d-d_i)+\Gamma_1(d_i)=4n-d_i+\Gamma_1(d_i)$. Now notice that $d_i=d_d+r-4n < p$ and the quantity $d_i-\Gamma_1(d_i)$ is non-decreasing as $d_i$ increases. Thus the number of edit operations is at least $4n-p+\Gamma_1(p) \geq 4n-p-q+\Gamma_1(p)$.
\end{description}

The other case of $q$ is similar, as follows.

\paragraph{The case of $q=2$.} Here again we have two cases.

\begin{description}
\item[Case (a):]
$d_d \geq 4n-2p+1.$ In this case, notice that the LCS after the operations between $x_1$ and $y$ is at most the original $\lcs(x_1, y)=2(p-1)-\Gamma_1(p-1)+1=2p-1-\Gamma_1(p-1)$. With $d_i$ insertions, the number of agreements can be at most $\lcs(x_1, y)+d_i=2p-1-\Gamma_1(p-1)+d_i$, thus the Hamming distance at the end is at least $r-2p+1+\Gamma_1(p-1)-d_i$. Therefore, in this case the number of edit operations is at least $d_i+d_d+r-2p+1+\Gamma_1(p-1)-d_i \geq 4n-p+\Gamma_1(p-1)$, and the equality is achieved when $d_d = 4n-2p+1$.
\item[Case (b):] $d_d \leq 4n-2p.$ In this case, notice that all original symbols in $x_1$ larger than $3d_i$ (or beyond index $2d_i$ before the insertions) are guaranteed to be out of agreement. Thus the only possible original symbols in $x_1$ that are in agreement with $y$ after the operations are the symbols with original index at most $2d_i$. Note that the LCS between $x[1: 2d_i]$ and $y$ is $2d_i-\Gamma_1(d_i)$. Thus with $d_i$ insertions the number of agreements is at most $3d_i-\Gamma_1(d_i)$, and the Hamming distance at the end is at least $r-3d_i+\Gamma_1(d_i)$.

Therefore the number of edit operations is at least $d_i+d_d+r-3d_i+\Gamma_1(d_i)=r-d_i+(d_d-d_i)+\Gamma_1(d_i)=4n-d_i+\Gamma_1(d_i)$. Now notice that $d_i=d_d+r-4n < p-1$ and the quantity $d_i-\Gamma_1(d_i)$ is non-decreasing as $d_i$ increases. Thus the number of edit operations is at least $4n-(p-1)+\Gamma_1(p-1) > 4n-p+\Gamma_1(p-1)$.
\end{description}

\end{proof}

We can now prove a similar lemma for $x_2$. For any $i$, let $\Gamma_2(i)$ denote the number of $a$ symbols from index $2i+1$ to $4n$ in $x_2$. Note that $\Gamma_2(i)$ is equal to the number of $0$'s in $\beta'[i+1:2n]$.

\begin{lemma} \label{lem:ED2}
Let $r=3p+q$ where $0 \leq q \leq 2$, then $\ed(x_2, [r+1:6n]) = 2n+p-q+\Gamma_2(p)$ if $q=0, 1$ and $\ed(x_2, [r+1:6n]) = 2n+p+\Gamma_2(p+1)$ if $q=2$.
\end{lemma}

\begin{proof}
We can reduce to Lemma~\ref{lem:ED1}. To do this, use $6n+1$ to minus every symbol in $x_2$ and in $[r+1:6n]$, while keeping all the $a$ symbols unchanged. Now, reading both strings from right to left, $x_2$ becomes the string $\overline{x_2}=1, 2, \cdots, 3i-2, 3i-1, \cdots, 6n-2, 6n-1$ with some symbols of the form $3j-1$ replaced by $a$'s. Similarly $[r+1:6n]$ becomes $[1:6n-r]$ where $6n-r=3(2n-p)-q$.

If we regard $\overline{x_2}$ as $x_1$ as in Lemma~\ref{lem:ED1} and define $\Gamma_1(i)$ as in that lemma, we can see that $\Gamma_1(i)=\Gamma_2(2n-i)$.

Now the lemma basically follows from Lemma~\ref{lem:ED1}. In the case of $q=0, 1$, we have

\[\ed(x_2, [r+1:6n]) = \ed(\overline{x'}, [1:6n-r]) = 4n-(2n-p)-q+\Gamma_1(2n-p)=2n+p-q+\Gamma_2(p).\]

In the case of $q=2$, we have

\[\ed(x_2, [r+1:6n]) = \ed(\overline{x'}, [1:6n-r]) = 4n-(2n-p)+\Gamma_1(2n-p-1)=2n+p+\Gamma_2(p+1).\]
\end{proof}

We can now prove Lemma~\ref{lem:EDmain}.

\begin{proof}[Proof of Lemma~\ref{lem:EDmain}]
We show that for any $r \in [6n]$, $\ed(x_1, [1:r])+\ed(x_2, [r+1:6n]) \geq 7n-2$. First we have the following claim.

\begin{claim}
If $\dis(\alpha, \beta)=1$, then for any $i \in [2n]$, we have $\Gamma_1(i)+\Gamma_2(i) \geq n$.
\end{claim}

To see this, note that when $i$ is even, we have $\Gamma_1(i)=i/2$ and $\Gamma_1(i)=n-i/2$ so $\Gamma_1(i)+\Gamma_2(i)=n$. Now consider the case of $i$ being odd and let $i=2j-1$ for some $j \in [2n]$. We know $\Gamma_1(i-1)=(i-1)/2=j-1$ and $\Gamma_2(i+1)=n-(i+1)/2=n-j$, so we only need to look at $x_1[2i-1, 2i]$ and $x_2[2i+1, 2i+2]$ and count the number of symbols $a$'s in them. If the number of $a$'s is at least $1$, then we are done.

The only possible situation where the number of $a$'s is $0$ is that $\alpha'_i=\beta'_{i+1}=1$ which means $\alpha_j=\beta_j=1$ and this contradicts the fact that $\dis(\alpha, \beta)=1$.

We now have the following cases.

\begin{description}
\item[Case (a):] $r=3p$. In this case, by Lemma~\ref{lem:ED1} and Lemma~\ref{lem:ED2} we have $\ed(x_1, [1:r]) = 4n-p+\Gamma_1(p)$ and $\ed(x_2, [r+1:6n]) = 2n+p+\Gamma_2(p)$. Thus we have  $\ed(x_1, [1:r])+\ed(x_2, [r+1:6n]) = 6n+n=7n$.

\item[Case (b):] $r=3p-1=3(p-1)+2$. In this case, by Lemma~\ref{lem:ED1} and Lemma~\ref{lem:ED2} we have $\ed(x_1, [1:r]) = 4n-p-1+\Gamma_1(p)$ and $\ed(x_2, [r+1:6n]) = 2n+(p-1)+\Gamma_2(p)$, thus we have $\ed(x_1, [1:r])+\ed(x_2, [r+1:6n]) = 6n-2+n=7n-2$.

\item[Case (c):] $r=3p-2=3(p-1)+1$. In this case, by Lemma~\ref{lem:ED1} and Lemma~\ref{lem:ED2} we have $\ed(x_1, [1:r]) = 4n-p+\Gamma_1(p-1)$ and $\ed(x_2, [r+1:6n]) = 2n+(p-1)-1+\Gamma_2(p-1)$, thus we have $\ed(x_1, [1:r])+\ed(x_2, [r+1:6n]) = 6n-2+n=7n-2$.

\end{description}

\end{proof}

We now prove Theorem~\ref{thm:EDmain}. 

\begin{proof}[Proof of Theorem~\ref{thm:EDmain}]
We begin by  upper bounding $\ed((x_1 \circ x_2), y)$ when $\dis(\alpha, \beta)=0$. 

\begin{claim} \label{clm:swith}
If $\dis(\alpha, \beta)=0$ then $\ed((x_1 \circ x_2), y) \leq 7n-3$.
\end{claim}

To see this, note that if $\dis(\alpha, \beta)=0$ then there exists a $j \in [n]$ such that $\alpha_j=\beta_j=1$. Thus $\alpha'_{2j-1}=1$, $\beta'_{2j-1}=0$ and  $\alpha'_{2j}=0$, $\beta'_{2j}=1$. Note that the number of $0$'s in $\alpha'[1:2j-1]$ is $j-1$ and thus $\Gamma_1(2j-1)=j-1$. Similarly the number of $0$'s in $\beta'[2j:2n]$ is $n-j$ and thus $\Gamma_2(2j-1)=n-j$.  To transform $(x_1, x_2)$ to $y$, we choose $r=6j-2$,  transform $x_1$ to $y[1:r]$, and transform $x_2$ to $y[r+1:6n]$. 

By Lemma~\ref{lem:ED1} and Lemma~\ref{lem:ED2} we have $\ed(x_1, [1:r]) = 4n-2j+\Gamma_1(2j-1)$ and $\ed(x_2, [r+1:6n]) = 2n+(2j-1)-1+\Gamma_2(2j-1)$. .Thus we have $\ed(x_1, [1:r])+\ed(x_2, [r+1:6n]) = 6n-2+\Gamma_1(2j-1)+\Gamma_2(2j-1)=6n-2+n-1=7n-3$. Therefore $\ed((x_1, x_2), y) \leq 7n-3$.

Therefore, in the case of $\dis(\alpha, \beta)=1$, we have $\ed((x_1 \circ x_2), y) \geq 7n-2$ while in the case of $\dis(\alpha, \beta)=0$, we have $\ed((x_1 \circ x_2), y) \leq 7n-3$. Thus any protocol that solves $\ed_{cc}(y)$ can also solve $\dis$, hence the theorem follows.
\end{proof}

In the proof of Theorem~\ref{thm:EDmain}, the two strings $x=(x_1 \circ x_2)$ and $y$ have different lengths, however we can extend it to the case where the two strings have the same length and prove the following theorem.

\begin{theorem}\label{thm:EDmain2}
Suppose each input string to $\dis$ has length $n$ and let $\Sigma=[16n] \cup \{a\}$. Fix $\tilde{y}=(1, 2, \cdots, 16n, a^{2n})$, let $\tilde{x}_1 \in \Sigma^{4n}$ and $\tilde{x}_2 \in \Sigma^{14n}$. Define $\ed_{cc}(\tilde{y})$ as the two party communication problem of computing $\ed((\tilde{x}_1 \circ\tilde{x}_2), \tilde{y})$. Then $R^{1/3}(\ed_{cc}(\tilde{y})) \geq R^{1/3}(\dis)$.
\end{theorem}

\begin{proof}
We extend the construction of Theorem~\ref{thm:EDmain} as follows. From input $(\alpha, \beta)$ to $\dis$, first construct $(x_1, x_2)$ as before. Then, let $z=(6n+1, 6n+2, \cdots, 16n)$, $\tilde{x}_1=x_1$ and $\tilde{x}_2=x_2 \circ z$. Note that we also have $\tilde{y}=y \circ z \circ a^{2n}$, and $|\tilde{x}_1 \circ \tilde{x}_2|=18n=|\tilde{y}|$.

We finish the proof by establishing the following two lemmas.

\begin{lemma}
If $\dis(\alpha, \beta)=1$ then $\ed((\tilde{x}_1 \circ \tilde{x}_2), \tilde{y}) \geq 9n-2$.
\end{lemma}

\begin{proof}
First we can see that $\ed((\tilde{x}_1 \circ \tilde{x}_2), \tilde{y}) < 10n$ since we can first use at most $8n-1$ edit operations to change $(x_1, x_2)$ into $y$ (note that the first symbols are the same), and then add $2n$ symbols of $a$ at the end.

Now we have the following claim:
\begin{claim}\label{clm:one}
In an optimal sequence of edit operations that transforms $(\tilde{x}_1 \circ \tilde{x}_2)$ to $\tilde{y}$, at the end some symbol in $z$ must be kept and thus matched to $\tilde{y}$ at the same position.
\end{claim}

To see this, assume for the sake of contradiction that none of the symbols in $z$ is kept, then this already incurs at least $10n$ edit operations, contradicting the fact that $\ed((\tilde{x}, \tilde{x}'), \tilde{y}) < 10n$. 

We now have a second claim:
\begin{claim}\label{clm:two}
In an optimal sequence of edit operations that transforms $(\tilde{x}_1 \circ \tilde{x}_2)$ to $\tilde{y}$, at the end all symbols in $z$ must be kept and thus matched to $\tilde{y}$ at the same positions.
\end{claim}

To see this, we use Claim~\ref{clm:one} and first argue that some symbol of $z$ is kept and matched to $\tilde{y}$. Assume this is the symbol $r$. Then we can grow this symbol both to the left and to the right and argue that all the other symbols of $z$ must be kept. For example, consider the symbol $r+1$ if $r < 16n$. There is a symbol $r+1$ that is matched to $\tilde{y}$ in the end. If this symbol is not the original $r+1$, then the original one must be removed either by deletion or substitution, since there cannot be two symbols of $r+1$ in the end. Thus instead, we can keep the original  $r+1$ symbol and reduce the number of edit operations.

More precisely, if this $r+1$ symbol is an insertion, then we can keep the original$r+1$ symbol and get rid of this insertion and the deletion of the original $r+1$, which saves $2$ operations. If this $r+1$ symbol is a substitution, then we can keep the original$r+1$ symbol, delete the symbol being substituted, and get rid of the deletion of the original $r+1$, which saves $1$ operation. The case of $r-1$ is completely symmetric. Continue doing this, we see that all symbols of $z$ must be kept and thus matched to $\tilde{y}$.

Now, we can see the optimal sequence of edit operations must transform $(x_1, x_2)$ into $y$, and transform the empty string into $a^{2n}$. Thus by Lemma~\ref{lem:EDmain} we have 
\[\ed((\tilde{x}_1 \circ \tilde{x}_2), \tilde{y}) = \ed((x_1, x_2), y) +2n \geq 9n-2.\]
\end{proof}

We now have the next lemma.

\begin{lemma}
If $\dis(\alpha, \beta)=0$ then $\ed((\tilde{x}_1 \circ \tilde{x}_2), \tilde{y}) \leq 9n-3$.
\end{lemma}

\begin{proof}
Again, to transform $(\tilde{x}_1 \circ \tilde{x}_2)$  to $\tilde{y}$, we can first transform $(x_1, x_2)$ into $y$, and insert $a^{2n}$ at the end. If $\dis(\alpha, \beta)=0$ then by Claim~\ref{clm:swith} $\ed((x_1 \circ x_2), y) \leq 7n-3$. Therefore we have

\[\ed((\tilde{x}_1 \circ \tilde{x}_2), \tilde{y}) \leq \ed((x_1, x_2), y) +2n \leq 9n-3.\]

Thus any protocol that solves $\ed_{cc}(\tilde{y})$ can also solve $\dis$, hence the theorem follows.
\end{proof}

\end{proof}

From Theorem~\ref{thm:EDmain2} we immediately have the following theorem.

\begin{theorem}\label{thm:EDmain3}
Any $R$-pass randomized algorithm in the asymmetric streaming model that computes $\ed(x, y)$ exactly between two strings $x, y$ of length $n$ with success probability at least $2/3$ must use space at least $\Omega(n/R)$.
\end{theorem}

We can generalize the theorem to the case of deciding if $\ed(x, y)$ is a given number $k$. First we prove the following lemmas.

\begin{lemma}\label{lem:pad1}
Let $\Sigma$ be an alphabet. For any $n, \ell \in \N$ let $x, y \in \Sigma^n$ and $z \in \Sigma^{\ell}$ be three strings. Then $\ed(x \circ z, y \circ z)=\ed(x, y)$.
\end{lemma}

\begin{proof}
First it is clear that $\ed(x \circ z, y \circ z) \leq \ed(x, y)$, since we can just transform $x$ to $y$. Next we show that $\ed(x \circ z, y \circ z) \geq \ed(x, y)$.

To see this, suppose a series of edit operations transforms $x$ to $y'=y[1:n-r]$ or $y'=y[1:n] \circ z[1:r]$ for some $r \geq 0$ and transforms $z$ to the other part of $y \circ z$ (called $z'$). Then by triangle inequality we have $\ed(x, y') \geq \ed(x, y)-r$. Also note that $\ed(z, z') \geq \left| |z|-|z'| \right |=r$. Thus the number of edit operations is at least  $\ed(x, y')+\ed(z, z') \geq \ed(x, y)$.
\end{proof}

\begin{lemma}\label{lem:pad2}
Let $\Sigma$ be an alphabet. For any $n, \ell \in \N$ let $x, y \in \Sigma^n$ and $u, v \in \Sigma^{\ell}$ be four strings. If there is no common symbol between any of the three pairs of strings $(u, v)$, $(u, y)$ and $(v, x)$, then $\ed(x \circ u, y \circ v)=\ed(x, y)+\ell$.
\end{lemma}

\begin{proof}
First it is clear that $\ed(x \circ u, y \circ v) \leq \ed(x, y)+\ell$, since we can just transform $x$ to $y$ and then replace $u$ by $v$. Next we show that $\ed(x \circ u, y \circ v) \geq \ed(x, y)+\ell$.

To see this, suppose a series of edit operations transforms $x$ to $y'=y[1:n-r]$ for some $r \geq 0$ and transforms $u$ to the other part of $v'=y[n-r+1:n] \circ v$. Then by triangle inequality we have $\ed(x, y') \geq \ed(x, y)-r$. Since there is no common symbol between $(u, v)$ and $(u, y)$ , we have $\ed(u, v') \geq r+\ell$. Thus the number of edit operations is at least  $\ed(x, y')+\ed(u, v') \geq \ed(x, y)+\ell$. The case of transforming $x$ to $y'=y[1:n] \circ v[1:r]$ for some $r \geq 0$ is completely symmetric since equivalently it is transforming $y$ to $x'=[1:n-r']$ for some $r' \geq 0$.
\end{proof}

We have the following two theorems.

\begin{theorem}\label{thm:EDmain4}
There is a constant $c>1$ such that for any $k, n \in \N$ with $n \geq ck$, and alphabet $\Sigma$ with $|\Sigma| \geq c k$, any $R$-pass randomized algorithm in the asymmetric streaming model that decides if $\ed(x, y) \geq k$ between two strings $x, y \in \Sigma^n$ with success probability at least $2/3$ must use space at least $\Omega(k/R)$.
\end{theorem}

\begin{proof}
Theorem~\ref{thm:EDmain2} and Theorem~\ref{thm:EDmain3} can be viewed as deciding if $\ed(x, y) \geq 9n-2$ for two strings of length $18n$ over an alphabet with size $16n+1$. Thus we can first use the constructions there to reduce $\dis$ to the problem of deciding if $\ed(x, y) \geq k$ with a fixed string $y$ of length $O(k)$. The number of symbols used is $O(k)$ as well. Now to increase the length of the strings to $n$, we pad a sequence of the symbol $1$ at the end of both $x$ and $y$ until the length reaches $n$.\ By Lemma~\ref{lem:pad1} the edit distance stays the same and thus the problem is still deciding if $\ed(x, y) \geq k$. By Theorem~\ref{thm:EDmain2} the communication complexity is $\Omega(k)$ and thus the theorem follows.
\end{proof}

\begin{theorem}\label{thm:EDmain5}
There is a constant $c>1$ such that for any $k, n \in \N$ and alphabet $\Sigma$ with $n \geq ck \geq  |\Sigma|$, any $R$-pass randomized algorithm in the asymmetric streaming model that decides if $\ed(x, y) \geq k$ between two strings $x, y \in \Sigma^n$ with success probability at least $2/3$ must use space at least $\Omega(|\Sigma|/R)$.
\end{theorem}

\begin{proof}
Theorem~\ref{thm:EDmain2} and Theorem~\ref{thm:EDmain3} can be viewed as deciding if $\ed(x, y) \geq 9n-2$ for two strings of length $n$ over an alphabet with size $18n+1$. Thus we can first use the constructions there to reduce $\dis$ to the problem of deciding if $\ed(x, y) \geq k'=\Omega(|\Sigma|)$ with a fixed string $y$ of length $\Theta(|\Sigma|)$, and the number of symbols used is $|\Sigma|-2$. Now we take the 2 unused symbols and pad a sequence of these two symbols with length $k-k'$ at the end of $x$ and $y$. By Lemma~\ref{lem:pad2} the edit distance increases by $k-k'$ and thus the problem becomes deciding if $\ed(x, y) \geq k$. Next, to increase the length of the strings to $n$, we pad a sequence of the symbol $1$ at the end of both strings until the length reaches $n$. By Lemma~\ref{lem:pad1} the edit distance stays the same and thus the problem is still deciding if $\ed(x, y) \geq k$. By Theorem~\ref{thm:EDmain2} the communication complexity is $\Omega(|\Sigma|)$ and thus the theorem follows.
\end{proof}

Combining the previous two theorems we have the following theorem, which is a restatement of Theorem~\ref{thm:ED1}.

\begin{theorem}[Restatement of Theorem~\ref{thm:ED1}]\label{thm:EDmain6}
There is a constant $c>1$ such that for any $k, n \in \N$ with $n \geq ck$, given an alphabet $\Sigma$, any $R$-pass randomized algorithm in the asymmetric streaming model that decides if $\ed(x, y) \geq k$ between two strings $x, y \in \Sigma^n$ with success probability at least $2/3$ must use space at least $\Omega(\mathsf{min}(k, |\Sigma|)/R)$.
\end{theorem}

For $0< \eps < 1$, by taking $k=1/\eps$ we also get the following corollary:

\begin{corollary}\label{cor:EDmain6}
Given an alphabet $\Sigma$, for any $0< \eps < 1$, any $R$-pass randomized algorithm in the asymmetric streaming model that achieves a $1+\eps$ approximation of $\ed(x, y)$ between two strings $x, y \in \Sigma^n$ with success probability at least $2/3$ must use space at least $\Omega(\mathsf{min}(1/\eps, |\Sigma|)/R)$.
\end{corollary}


\section{Lower Bounds for LCS}
\label{sec:lcs}

In this section, we study the space lower bounds for asymmetric streaming LCS. 

\subsection{Exact computation}

\subsubsection{Binary alphabet, deterministic algorithm}

\label{sec:lcs_exact_det}

In this section, we assume $n$ can be diveded by $60$ and let $l = \frac{n}{30} -1$. We assume the alphabet is $\Sigma = \{a,b\}$. Consider strings $x$ of the form 

\begin{equation}
	\label{lcs:eq_a}
	x = b^{10}a^{s_1}b^{10}a^{s_2}b^{10} \cdots b^{10}a^{s_l}b^{10} .
\end{equation}

That is, $x$ contains $l$ blocks of consecutive $a$ symbols. Between each block of $a$ symbols, we insert $10$ $b$'s and we also add $10$ $b$'s to the front, and the end of $x$. $s_1,\dots, s_l$ are $l$ integers such that

\begin{align}
	&\sum_{i = 1}^{l} s_i = \frac{n}{6}+5  \label{lcs:eq_b}, \\
	 &1\leq s_i\leq 9, \; \forall i\in [l] \label{lcs:eq_c}. 
\end{align}

Thus, the length of $x$ is $\sum_{i = 1}^{l} \frac{n}{6}+5+ 10 (l+1)  = \frac{n}{2}+ 5$ and it contains exactly $\frac{n}{3}$ $b$'s. 

Let $S$ be the set of all $x\in \{a,b\}^{\frac{n}{2}+5}$ of form \ref{lcs:eq_a} that satisfying equations \ref{lcs:eq_b}, \ref{lcs:eq_c}. For each string $x\in S$, we can define a string $f(x) \in \{a,b\}^{\frac{n}{2}-5}$ as following. Assume $x = b^{10}a^{s_1}b^{10}a^{s_2}b^{10} \cdots b^{10} a^{s_l}b^{10} $, we set $f(x) = a^{s_1}b^{10}a^{s_2}b^{10} \cdots b^{10} a^{s_l}b^{10} $. That is, $f(x)$ simply removed the first $10$ $b$'s of $x$. We denote $\bar{S} = \{f(x)| x\in S\}$ .
 
\begin{claim}
	\label{claim:|S|}
	$|S|  = |\bar{S}| = 2^{\Omega(n)}$.
\end{claim}

\begin{proof}
	Notice that for $x^1, x^2 \in S$, if $x^1 \neq x^2$, then $f(x^1)\neq f(x^2)$. We have $|S| =| \bar{S}|$.
	
	The size of $S$ equals to the number of choices of $l$ integers $s_1, s_2, \dots, s_l$ that satisfies \ref{lcs:eq_b} and \ref{lcs:eq_c}. For an lower bound of $|S|$, we can pick $\frac{n}{60}$ of the integers to be $9$, and set the remaining to be $1$ or $2$. Thus the number of such choices is at least $\binom{l}{\frac{n}{60}} = \binom{\frac{n}{30}-1}{\frac{n}{60}} = 2^{\Omega(n)}.$ 
	
	
\end{proof}

We first show the following lemma.

\begin{lemma}
	\label{lem:S}
	Let $y = a^{n/3}b^{n/3}a^{n/3}$. For every $x\in S$, $$\lcs(x\circ f(x), y) = \frac{n}{2}+5.$$  For any two distinct $x^1, x^2\in S$, \[\max\{\lcs(x^1\circ f(x^2), y),\lcs(x^2\circ f(x^1), y)\} > \frac{n}{2}+5.\]
\end{lemma}

\begin{proof}[Proof of Lemma~\ref{lem:S}]
	We first show $\lcs(x\circ f(x), y) = \frac{n}{2}+5$. Notice that $x\circ f(x)$ is of the form 
	$$b^{10}a^{s_1}b^{10} \cdots b^{10} a^{s_l }b^{10} a^{s_1} b^{10} \cdots b^{10}a^{s_l }b^{10} .$$
	
	It cantains $2l+1$ block of $b$'s, each consists $ 10$ consecutive $b$'s. These blocks of $b$'s are seperated by some $a$'s. Also, $x\circ f(x)$ has $2\sum_{i = 1}^{l}s_i = \frac{n}{3}+10$ $a$'s and $ \frac{2n}{3}-10$ $b$'s. Let $p_i$ be the first position of the $i$-th block of $b$'s. 
	
	Let us consider a matching between $x\circ f(x)$ and $y$. 
	
	If the matching does not match any $b$'s in $y$ to $x\circ f(x)$, the size of such a matching is at most $\frac{n}{3}+10$ since it is the number of $a$'s in $x\circ f(x)$. 
	
	Now we assume some 1's in $y$ are matched to $x\circ f(x)$. Without loss of generality, we can assume the first $b$ symbol in $y$ is matched. This is because all $b$'s in $y$ are consecutive and if the first $b$ in $y$ (i.e. $y_{\frac{n}{3}+1}$) is not matched, we can find another matching of the same size that matches $y_{\frac{n}{3}+1}$. For the same reason, we can assume the first $b$ in $y$ is matched to position $p_i$ for some $i\in [2l+1]$. Assume $n_b$ is the number of $b$'s matched. Again, without loss of generality, we can assume the first $n_b$ $b$'s starting from position $p_i$ in $x\circ f(x)$ are matched since all $b$'s in $y$ are consecutive and there are no matched $a$'s between two matched $b$'s. We have two cases.
	
	\begin{description}

		\item{Case 1}: $y_{\frac{n}{3}+1}$ is matched to $p_i$ for some $1\leq i\leq l+1$. Let $n_b$ be the number of matched $b$'s. We know $n_b\leq \frac{n}{3}$ since there are $\frac{n}{3}$ $b$'s in $y$. 
	
		If $n_b = \frac{n}{3}$, we match first $\frac{n}{3}$ $b$'s in $x\circ f(x)$ starting from position $p_i$. Consider the number of $a$'s that are still free to match in $x\circ f(x)$. The number of $a$'s before $p_i$ is $\sum^{i-1}_{j = 1}s_i$. Since $i\leq l+1$, $\sum^{i-1}_{j = 1}s_i$ is at most $\frac{n}{6} + 5$, we can match all of them to first third of $y$. Also, we need $l+1$ blocks of $b$'s to match all $b$'s in $y$. The number of $a$'s after last matched $b$ in $x\circ f(x)$ is $\sum_{j = i}^{l} s_j$ (which is zero when $i = l+1$). Again, we can match all these $a$'s since $\sum_{j = i}^{l} s_j\leq \frac{n}{3}$. In total, we can match $\sum_{j = 1}^{l} s_j = \frac{n}{6}+ 5$ $a$'s.  This gives us a matching of size $\frac{n}{3} + \frac{n}{6}+ 5 = \frac{n}{2}+ 5$. We argue that the best strategy is to always match $\frac{n}{3}$ $b$'s. To see this, if we removed $10$ matched $b$'s, this will let us match $s_j$ additional  $0$'s for some $j\in [l]$. By our construction of $x$, $s_j$ is strictly smaller than $10$ for all $j\in [l]$. If we keep doing this, the size of matching will decrease. Thus, the largest matching we can find in this case is $\frac{n}{2}+ 5$.

		\item{Case 2}: $y_{\frac{n}{3}+1}$ is matched to $p_i$ for some $l+1 < i \leq 2l+1$. By the same argument in Case 1, the best strategy is to match as many $b$'s as possible. The number of $b$'s in $x\circ f(x)$ starting from position $p_i$ is $ (2l+1-i+1)t = (2l-i+2 )t $. The number of $a$'s that are free to match is $\sum_{j = 1}^{l} s_j + \sum_{j' = 1}^{i-l+1}s_{j'} = \frac{n}{6} + 5 +  \sum_{j = 1}^{i-l-1}s_j$. Since  $s_j < t $ for all $j\in [l]$, the number of $a$'s  can be matched is strictly smaller than $\frac{n}{6} + 5 +  (i-l+1)t$. The largest matching we can find in this case is smaller than $\frac{n}{6} + 5 +  (l+1)t = \frac{n}{2}+ 5$.

	\end{description}
	 
	This proves the size of the largest matching we can find is exactly $\frac{n}{2}+ 5$. We have $\lcs(x\circ f(x), y) = \frac{n}{2}+5$.

	For the second statement in the lemma, say $x^1 $ and $x^2$ are two distinct strings in $S$. For convenience, we assume $x^1 = b^{10}a^{s^1_1}b^{10}a^{s^1_2}b^{10} \cdots b^{10} a^{s^1_l}b^{10}$, and $x^2 = b^{10}a^{s^2_1}b^{10}a^{s^2_2}b^{10} \cdots b^{10} a^{s^2_l} b^{10}$. Let $i$ be the smallest integer such that $s^1_i\neq s^i_2$. We have $i\leq l-1$ since $\sum_{j = 1}^{l} s^1_j =  \sum_{j = 1}^{l} s^2_j $. Without loss of generality, we assume $s^1_i > s^2_i$. We show that $\lcs(x^1\circ f(x^2), y)> \frac{n}{2}+5$. Notice that $ x^1\circ f(x^2)$ is of the form 
	
	$$b^{10}a^{s^1_1}b^{10} \cdots b^{10} a^{s^1_l }b^{10} a^{s^2_1} b^{10} \cdots b^{10} a^{s^2_l } b^{10} .$$
	
	By the same notation, let $p_{j}$  be the first position of the $j$-th block of $b$'s in $ x^1\circ f(x^2)$
	
	Consider the match that matches the first $b$ in $y$ to position $p_{i+1}$. The number of $a$'s before $p_{i+1}$ is $\sum_{j = 1}^{i}s^1_j$. We matches all $\frac{n}{3}$ $b$'s in $y$. The number of $a$'s after the last matched $b$ in $x^1\circ f(x^2)$ is $\sum_{j = i+1}^{l} s^2_j$. This gives us a match of size $\frac{n}{3} + \sum_{j = 1}^{i}s^1_j+ \sum_{j = i+1}^{l} s^2_j $. By our choice of $i$, we have  $s^1_j = s^2_j$ for $j\in [i-1]$. The size of the matching equals to $ \frac{n}{3} + \sum_{j = 1}^{l}s^2_j + s^1_i - s^2_i = \frac{n}{2} + 5 + s^1_i - s^2_i$ which is larger than $ \frac{n}{2}+ 5$. Thus, the length of LCS between $x^1\circ f(x^2)$ and $y$ is larger than $ \frac{n}{2}+ 5$. This finishes our proof.

\end{proof}

\begin{lemma}
	\label{lem:lcs_exact_det}
	In the asymmetric streaming model, any deterministic protocol that computes $\lcs(x,y)$ for any $x, y \in \{0,1\}^n$, in $R$ passes of $x$ needs $\Omega(n/R)$ space.
\end{lemma}

\begin{proof}
	Consider a two party game where player 1 holds a string $x^1\in S$ and player 2 holds a string $x^2 \in S$. The goal is to verify whether $x^1 = x^2$.  It is known that the total communication complexity of testing the equality of two elements from set $S$ is $\Omega(\log |S|) $, see \cite{communication_complexity_book} for example. We can reduce this to computing the length of LCS. To see this, we first compute $\lcs(x^1\circ f(x^2), y)$ and $\lcs(x^2\circ f(x^1), y)$ with $y= a^{n/3}b^{n/3}a^{n/3}$. By lemma~\ref{lem:S}, if both $\lcs(x^1\circ f(x^2), y) = \lcs(x^2\circ f(x), y) =\frac{n}{2}+5$, we know $x^1 = x^2$, otherwise, $x^1 \neq x^2$. Here, $y$ is known to both parties. 
	
	The above reduction shows the total communication complexity of this game is $\Omega(n)$ since $|S| = 2^{\Omega(n)}$. If we only allow $R$ rounds of communication, the size of the longest message sent by the players is $\Omega(n/R)$. Thus, in the asymmetric model, any protocol that computes $\lcs(x,y)$ in $R$ passes of $x$ needs $\Omega(n/R)$ space.

\end{proof}

\subsubsection{$\Omega(n)$  size alphabet, randomized algorithm}

\begin{lemma}
	\label{lem:lcs_exact_rand}
	Assume $|\Sigma|  = \Omega(n)$. In the asymmetric streaming model, any randomized protocol that computes $\lcs(x,y)$ correctly with probability at least $2/3$ for any $x,y \in \Sigma^n$, in $R$ of passes of $x$ needs $\Omega(n/R)$ space. The lower bound also holds when $x$ is a permutation of $y$. 
\end{lemma}

\begin{proof}[Proof of Lemma~\ref{lem:lcs_exact_rand}]
	The proof is by giving a reduction from set-disjointness. 
	 
	In the following, we assume alphabet set $\Sigma = [n]$ which is the set of integers from 1 to $n$. We let the online string $x$ be a permutation of $n$ and the offline string $y = 12\cdots n$ be the concatenation from 1 to $n$.  Then computing $\lcs(x,y)$ is equivalent to compute $LIS(x)$ since any common subsequence of $x$ and $y$ must be increasing.

	We now describe our approach.  For convenience, let $n' = \frac{n}{2}$.  Without loss of generality, we assume $n'$ can be divided by $4$. Consider a string $z \in \{0,1\}^{n'}$ with the following property.  
	
	\begin{equation}
		 \forall \; \; i \in [\frac{n'}{2}], \qquad z_{2i} = 1-z_{2i-1}.
	\end{equation}

	For each $i\in [n']$, we consider subsets $\sigma^i\subset [n]$ for $i\in [n']$ as defined below
	\begin{equation*}
		\sigma^i = \begin{cases}
			\{ 4i-1,  4i\} , &\text{ if $i$ is odd and $i\leq \frac{n'}{2}$}\\
			\{4(i-\frac{n'}{2})-3, 4(i-\frac{n'}{2})-2\} &\text{ if $i$ is odd and $ i> \frac{n'}{2}$}\\
			\{4i-3,  4i-2\} &\text{ if $i$ is even and $ i\leq \frac{n'}{2}$}\\
			\{4(i-\frac{n'}{2})-1, 4(i-\frac{n'}{2})\} &\text{ if $i$ is even and $i > \frac{n'}{2}$}
		\end{cases}
	\end{equation*} 
	Notice that $\sigma^i\cap \sigma^j = \emptyset$ if $i\neq j$ and $\cup^{n'}_{i = 1}\sigma^i = [n]$. For an odd $i\in [n'/2]$, $\min \sigma^i > \max \sigma^{i+n'/2}$. Oppositely, for an even $i\in [n'/2]$, $\max \sigma^i< \min \sigma^{i+n'/2}$. Also notice that for distinct $i, j \in [n'/2]$ with $i<j$, $\max \sigma^i< \min \sigma^j$ and $\max \sigma^{i+n'/2}< \min \sigma^{j+n'/2}$.
	
	We abuse the notation a bit and let $\sigma^i(z_i)$ be a string such that if $z_i = 1$, the string consists of elements in set $\sigma^i$ arranged in an increasing order, and if $z_i = 0$, the string is arranged in decreasing order. 
	
	Let $x$ be the concatenation of $\sigma^i(z_i)$ for $i \in [n']$ such that $x = \sigma^1(z_1)\circ \sigma^2(z_2) \circ \cdots \circ \sigma^{n'}(z_{n'})$. By the definition of  $\sigma^i $'s, we know $x$ is a permutation of $[n]$.

	For convenience, let $z^1$ and $z^2$ be two subsequences of $z$ such that
	\begin{align*}
		&	z^1 = z_{2}\circ z_{4}\circ  \cdots z_{\frac{n'}{2}} \\
		&	z^2 = z_{\frac{n'}{2}+2}\circ z_{\frac{n'}{2}+4 }\circ \cdots \circ z_{n'}
	\end{align*}
	
	If $\dis(z^1, z^2) = 0$. Then there exist some $i\in [n'/4]$ such that $z_{2i} = z_{n'/2 + 2i} = 1$. Notice that in $z$, we have $\forall j\in [\frac{n'}{2}], \; z_{2j} = 1-z_{2j-1}.$ Thus, $\lis(\sigma^{2j-1}({z_{2j-1}})\circ \sigma^{2j}({z_{2j}})) = 3$ since only one of $\sigma^{2j-1}({z_{2j-1}})$ and $\sigma^{2j}({z_{2j}})$ is increasing and the other is decreasing. For any $j\in [n'/4]$, we have 
	
	\begin{align}
		\label{eq:lcs1}
		&\lis(\sigma^1(z_1)\circ \sigma^2(z_2) \circ \cdots \circ  \sigma^{2j}(z_{2j})) = 3j\\
		&\lis(\sigma^{2j+n'/2+1}(z_{2j+n'/2+1})\circ \sigma^{2j+n'/2+2}(z_{2j+n'/2+2}) \circ \cdots \circ \sigma^{n'}(z_{n'})) = 3\frac{n'-2j}{2}. \label{eq:lcs2}
	\end{align}
	
	Since $z_{2i + n'/2 } = 1$, $\lis(\sigma^{2i + n'/2}(z_{2i+n'/2}))  = 2$. Since $\max\sigma^{2i} < \min \sigma^{2i+n'/2} < \max \sigma^{2i+n'/2} < \min \sigma^{2i+n'/2+1}$,  combining with equations~\ref{eq:lcs1} and \ref{eq:lcs2}, we know $\lis(x)\geq 3\frac{n'}{2}+2$.
	
	If $\dis(z^1, z^2) = 1$, we prove that $\lis(x)  = 3\frac{n'}{2}  + 1$. We only need to consider in an longest increasing subsequence, when do we first pick some elements from the second half of $x$. Say the first element picked in the second half of $x$ is in $\sigma^{2i+n'/2-1}$ or $\sigma^{2i+n'/2}$ for some $i\in [n'/2]$. We have 
	$$ \lis(\sigma^{2i-1}(z_{2i-1}) \circ  \sigma^{2i}(z_{2i}) \circ \sigma^{2i+n'/2-1}(z_{2i+n'/2-1}) \circ \sigma^{2i+n/2}(z_{2i+n'/2})   ) = 4.$$
	
	This is because  $z_{2i}$ and $z_{2i+n'/2}$ can not both be 1,  $z_{2i}  = 1- z_{2i-1}$ , and $z_{2i+n'/2} = 1-z_{2i+n'/2-1}$ . The length of $\lis$ of the substring of $x$ before $\sigma^{2i-1}(z_{2i-1})$ is  $3(i-1)$ and the length of LIS after $\sigma^{2i+n/2}(z_{2i+n'/2})  $ is $3(n'/2-i)$. Thus, we have
	$\lis(x)  = 3\frac{n'}{2}+ 1$.
	
	This gives a reduction from computing $\dis(z^1, z^2)$ to computing $\lis(x) = \lcs(x,y)$. Now assume player 1 holds the first half of $x$ and player 2 holds the second half. Both players have access to $y$. Since $|z^1|  = |z^2| = n'/4 = \frac{n}{8}$. Any randomized protocol that computes $\lcs(x,y)$ with success probability at least $2/3$ has an total communication complexity $\Omega(n)$. Thus, any randomized asymmetric streaming algorithm with $R$ passes of $x$ needs $\Omega(n/R)$ space. 
\end{proof}

We can generalize the above lemma to the following.

\begin{theorem}[Restatement of Theorem~\ref{thm:LCS1}]
	There is a constant $c>1$ such that for any $k, n\in \N$ with $n>ck$, given an alphabet $\Sigma$, any $R$-pass randomized algorithm in the asymmetric streaming model that decides if $\lcs(x,y)\geq k$ between two strings $x, y \in \Sigma^n$ with success probability at least $2/3$ must use at least $\Omega\big(\mathsf{min}(k,|\Sigma|)/R\big) $ space.
\end{theorem}

\begin{proof}

	Without loss of generality, assume $\Sigma = [r]$ so $|\Sigma| = r$. Since we assume $k< n$, we have $\mathsf{min}(k, r )< n$.

	Let $d = \mathsf{min}(k,r)$.  we let the offline string $y$ be the concatenation of two parts $y^1$ and $y^2$ where $y^1$ is the concatenation of symbols in $[d-2]\subseteq \Sigma$ in ascending order. $y^2$ is the symbol $d-1$ repeated $n-d+2$ times. Thus, $y\in \Sigma^n$. $x$ also consists of two parts $x^1$ and $x^2$ such that $x^1$ is over alphabet $[d-2]$ with length $d-2$ and $x^2$ is the symbol $d$ repeated $n-d+2$ times. Since the symbol $d-1$ does not appear in $x$ and the symbol $d$ does not appear in $y$. Thus, $\lcs(x,y) = \lcs(x^1,y^1)$. By Lemma~\ref{lem:lcs_exact_rand}, any randomized algorithm that computes $\lcs(x^1,y^1)$ with probability at least 2/3 using $R$ passes of  $x^1$ requires $\Omega(d/R)$ space.
	
	
	
\end{proof}

\subsection{Approximation}

We now show a lower bound for deterministic $1+\eps$ approximation of $\lcs$ in the asymmetric streaming model. 

\begin{theorem}[Restatement of Theorem~\ref{thm:LCS3}]
	\label{lem:lcs_approx_r_eps}
	Assume $\eps>0$, and $\frac{|\Sigma|^2}{\eps}  \leq n$ . In the asymmetric streaming model, any deterministic protocol that computes an $1+\eps$ approximation of $\lcs(x,y)$ for any $x, y\in \Sigma^n$, with constant number of passes of $x$ needs $\Omega(\frac{|\Sigma|}{\eps})$ space.
\end{theorem}

\begin{proof}[Proof of Theorem~\ref{lem:lcs_approx_r_eps}]
	In the following, we assume the size of alphabet set $\Sigma$ is $3r$ such that 
	$$\Sigma  = \{a_1, b_1, c_1, a_2, b_2, c_2, \dots a_r, b_r, c_r\}. $$

	We let $n_\eps = \Theta(1/\eps)$ such that $n_\eps$ can be divided by 60. 
	
	For any distinct $a,b\in\Sigma$, we can build a set $S_{a,b}$ in the same way as we did in section~\ref{sec:lcs_exact_det} except that we replace $n$ wtih $n_\eps$ (we used notation $S$ instead $S_{a,b}$). Thus, $S_{a,b}\subseteq \{a,b\}^{n_\eps/2+5}$. Similarly, we can define function $f$ and $\overline{S_{a,b}} = \{f(\sigma) | \sigma\in S_{a,b}\}\subset \{a,b\}^{n_\eps/2-5}$. Let $y^{a,b} = a^{n_\eps/3}b^{n_\eps/3} a^{n_\eps/3}$. By Claim~\ref{claim:|S|} and  Lemma~\ref{lem:S}, we know $|S_{a,b}| = |\overline{S_{a,b}}|  = 2^{\Omega(n_\eps)}$. For any $\sigma\in S_{a,b}$, we have $ \lcs(\sigma\circ f(\sigma), y^{a,b}) = n_\eps/2 + 5$. For any two distinct $\sigma^1, \sigma^2\in S_{a,b}$, we know at least one of $\lcs(\sigma^1\circ f(\sigma^2), y^{a,b})$ and $\lcs(\sigma^2\circ f(\sigma^1), y^{a,b})$ is at larger than $n_\eps/2+5$.

	Consider $w = (w_1, w_2, \dots, w_{2r})$ such that $w_{2i-1} \in S_{a,b}$ and $w_{2i} \in \overline{S_{a,b}}$ for $i\in [r]$. Thus, $w$ can be viewed as an element in 
	\[U  =  \underbrace{\big(S_{a,b}\times \overline{S_{a, b}}\big)\times \cdots \times \big(S_{a,b}\times \overline{S_{a, b}}\big)}_{r\text{ times}}. \]

	For alphabet $\{a_i, b_i\}$, we can similarly define $S_{a_i, b_i}$, $\overline{S_{a_i, b_i}}$ and  $y^{a_i, b_i}$. We let $U_i $ be similarly defined as $U$ but over alphabet $\{a_i, b_i\}$.

	Let $\beta = 1/3$. We can define function $h:(S_{a,b}\times \overline{S_{a, b}})^r \rightarrow \{0,1\}$ such that
	\begin{equation}
		h(w) = \begin{cases}
			0, &\text{ if } \forall i\in [r], \;  \lcs(w_{2i-1}\circ w_{2i}, y^{a_i,b_i}) = n_\eps/2 +5 \\
			1, &\text{ if for at least $\beta r$ indices $i\in [r]$, } \lcs(w_{2i-1}\circ w_{2i}, y^{a,b}) > n_\eps/2 +5 \\
			\text{undefined}, &\text{ otherwise}. 
		\end{cases}
	\end{equation}

	Consider an error-correcting code $T_{a,b} \subseteq S_{a,b}^r$ over alphabet $S_{a,b}$ with constant rate $\alpha$ and constant distance $\beta$. We can pick $\alpha = 1/2$ and $\beta = 1/3$, for example. Then the size of the code $T_{a,b}$ is $|T_{a,b} |= |S_{a,b}|^{\alpha r} = 2^{\Omega(n_\eps r)}$.  For any code word $\chi  = \chi_1 \chi_2 \cdots \chi_r \in T_{a,b}$ where $\chi_i\in S_{a,b}$, we can write 
	
	$$\nu (\chi) = (\chi_1, f(\chi_1), \chi_2, f(\chi_2), \dots, \chi_r, f(\chi_r)).$$

	Let $W = \{\nu(\chi)| \chi\in T_{a,b}\}\in (S_{a,b}\times \overline{S_{a, b}})^r $. Then

	\begin{equation}
		\label{eq:W_size}
		|W| = |T_{a,b}| = 2^{\Omega(n_\eps r)}
	\end{equation}

	We consider a $2r$-player one-way game where the goal is to compute the function $h$ on input $w = (w_1, w_2, \dots, w_{2r})$. In this game, player $i$ holds $w_i$ for $i\in[2r]$ and can only send message to player $i+1$ (Here, $2r+1 = 1$).  We now show the following claim.
	
	\begin{claim}
		\label{claim:W}
		$W$ is a fooling set for the function $h$. 
	\end{claim}

	\begin{proof}
		Consider any two different codewords $\chi, \chi' \in T_{a,b}$. Denote $w = \nu(\chi)$ and $w' = \nu(\chi')$. By our defintion of $w$, we know $w_{2i} = f(w_{2i-1})$ and $w_{2i-1}\in S_{a_i, b_i}$ for $i\in [r]$. We have $\lcs(w_{2i-1}\circ w_{2i}, y^{a_i,b_i}) = n_\eps/2 + 5$. Thus, $h(w) = h(w') = 0$. 
	
		The span of $\{w,w'\} $ is the set $\{(v_1, v_2, \dots, v_{2r})| v_i \in \{w_i, w'_i\} \text{ for } i \in [2r]\}$. We need to show that there exists some $v$ in the span of $\{w,w'\} $ such that $h(v) = 1$. Since $\chi, \chi'$ are two different codewords of $T_{a,b}$. We know there are at least $\beta r$ indices $i$ such that $w_{2i-1}\neq w'_{2i-1}$. Let $I\subseteq [r]$ be the set of indices that $\chi_i\neq \chi'_i$. Then, for $i\in I$, we have 
		\begin{equation*}
			\max\big(\lcs(w_{2i-1}\circ w'_{2i}, y^{a_i,b_i}), \lcs(w'_{2i-1}\circ w_{2i}, y^{a_i,b_i}) )\big) \geq n_\eps + 6. 
		\end{equation*}
		We can build a $v$ as following. For $i\in I$,  if $\lcs(w_{2i-1}\circ w'_{2i}, y^{a,b}\geq n_\eps + 6$. We then set $v_{2i-1} = w_{2i-1}$ and $v_{2i} = w'_{2i}$. Otherwise, we set $v_{2i-1} = w'_{2i-1}$ and $v_{2i} = w_{2i}$. For $i\in [r]\setminus I$, we set $v_{2i-1} = w_{2i-1}$ and $v_{2i} = w_{2i}$. Thus, for at least $\beta r$ indices $i\in [r]$, we have $\lcs(v_{2i-1}\circ v_{2i}, y^{a_i,b_i}) \geq n_\eps/2 +6$. We must have $h(v) = 1$. 
	\end{proof}

	Consider a matrix $B$ of size $r\times 2r$ of the following form
	
	\[\begin{pmatrix}
		B_{1,1}	& B_{1,2} & \cdots & B_{1,2r} \\
		B_{2,1}	& B_{2,2} & \cdots & B_{2,2r}  \\
		\vdots	& \vdots & \ddots & \vdots  \\
		B_{r,1}	& B_{r,2} & \cdots & B_{r,2r}
	\end{pmatrix}\]
	
	where  $B_{i,2j-1}\in S_{a_i, b_i}$ and $B_{i, 2j}\in \overline{S_{a_i, b_i}}$ for $j\in [r]$ (the elements of matrix $B$ are strings over alphabet $\Sigma$). Thus, the $i$-th row of $B$ is an element in $U_i$. We define the following function $g$. 
	
	\begin{equation}
		g(B) = h_1(R_1(B))\vee h_2(R_2(B))\vee \cdots \vee h_r(R_r(B))
	\end{equation}  

	where $R_i(B)\in U$ is the $i$-th row of matrix $B$ and $h_i$ is the same as $h$ except the inputs are over alphabet $\{a_i, b_i\}$ instead of $\{a,b\}$ for $i\in [r]$. Also, we define $W_i$ in exactly the same way as $W$ except elements in $W_i$ are over alphabet $\{a_i, b_i\}$ instead of $\{a,b\}$. 
	
	Consider a $2t$ player game where each player holds one column of $B$. The goal is to compute $g(B)$. We first show the following Claim. 
	
	\begin{claim}
		\label{claim:B}
		The set of all $r\times 2r$ matrix $B$ such that $R_i(B)\in W_i \; \forall\; i\in [r]$ is a fooling set for $g$. 
	\end{claim}
	
	\begin{proof}[Proof of Claim~\ref{claim:B}]
		For any two matrix $B_1 \neq B_2$ such that $R_i(B_1), R_i(B_2)\in W_i \; \forall\; i\in [r]$. We know $g(B_1) = g(B_2) = 0$. There is some row $i$ such that $R_i(B_1)\neq R_i(B_2)$. We know there is some elements $v$ in the span of $R_i(B_1)$ and $R_i(B_2)$, such that $h_i(v) = 1$ by Claim~\ref{claim:W}. Thus, there is some element $B' $ in the span of $B_1$ and $B_2$ such that $g(B') = 1$. Here, by $B' $ in the span of $B_1$ and $B_2$, we mean the $i$-th column is either $C_i(B_1)$ or $C_i(B_2)$ .
	\end{proof}

	Since $\lvert W\rvert = 2^{\Omega(n_\eps r)}$. By the above claim, we have a fooling set for $g$ size $\lvert W\rvert^r = 2^{\Omega(n_\eps r^2)}$. Thus, $CC^{tot}_{2r}(g) \geq \log(2^{\Omega(n_\eps r^2)}) = \Omega(n_\eps r^2)$. Since $CC^{\max}_{2r}\geq CC^{tot})_{2r} (g) = \Omega(n_\eps r)$.

	We now show how to reduce computing $g(B)$ to approximating the length of LCS in the asymmetric streaming model.
	
%

	We consider a matrix $\tilde{B}$ of size $r\times 3r$ such that
	
	\begin{equation}
		\label{eq:tB}
		\tB = \begin{pmatrix}
			B_{1,1}	& B_{1,2} & c_{1}^{n_\eps}  &B_{1,3}	& B_{1,4} & c_{2}^{n_\eps} & \cdots & B_{1,2r-1}  & B_{1,2r} &  c_{r}^{n_\eps}\\
			B_{2,1}	& B_{2,2}& c_{1}^{n_\eps}  &B_{2,3}	& B_{2,4} & c_{2}^{n_\eps} & \cdots & B_{2,2r-1} & B_{2,2r} &  c_{r}^{n_\eps}\\
			\vdots	& \vdots & \vdots  			  &\vdots	& \vdots& \vdots & \ddots & \vdots       &\vdots    & \vdots \\
			B_{r,1}	& B_{r,2} &  c_{1}^{n_\eps} &B_{r,3}	& B_{r,4} & c_{2}^{n_\eps} & \cdots & B_{r,2r-1}  & B_{r,2r}  &  c_{r}^{n_\eps}
		\end{pmatrix}_{(r\times 3r)}
	\end{equation}

	In other words, $\tB$ is obtained by inserting a column of $c$ symbols to $B$ at every third position. For $j\in [3r]$, let $C_j(\tilde{B}) = \tilde{B}_{1,j}\circ  \tilde{B}_{2,j} \circ \dots \circ \tilde{B}_{r,j}$. That is, $C_j(\tilde{B})$ is the concatenation of elements in the $j$-th column of $\tilde{B}$. We can set $x$ to be the concatenation of the columns of $\tB$. Thus, since $B_{i,2j-1}\circ B_{i, 2j}\in \Sigma^{n_\eps}$ for any $i,j\in [r]$, we have 
	
	\begin{equation*}
		 x  = C_1(\tilde{B})\circ C_2(\tilde{B}) \circ \cdots \circ C_{3r}(\tilde{B})\in \Sigma^{2r^2n_\eps }. 
	\end{equation*}
	
	For $i\in r$, we have defined $y^{a_i,b_i} = a_i^{n_\eps/3}b_i^{n_\eps/3}a_i^{n_\eps/3}\in \Sigma^{n_\eps}$. We let $y^{i,1}$ and $y^{i,2}$ be two non-empty strings such that $y^{a_i, b_i} = y^{i,1}\circ y^{i,2}$. We consider another matrix $\bar{B}$ of size $r\times 3r $ such that
	
	\begin{equation}
		\label{eq:bB}
		\bB = \begin{pmatrix}
			y^{1,1}	& y^{1,2} & c_{1}^{n_\eps}  &	y^{1,1}	& y^{1,2}& c_{2}^{n_\eps} & \cdots & 	y^{1,1}	& y^{1,2} &  c_{r}^{n_\eps}\\
			y^{2,1}	& y^{2,2}& c_{1}^{n_\eps}  &  y^{2,1}  & y^{2,2}& c_{2}^{n_\eps} & \cdots & y^{2,1}  & y^{2,2} &  c_{r}^{n_\eps}\\
			\vdots	& \vdots & \vdots  			  &\vdots	& \vdots& \vdots & \ddots & \vdots       &\vdots    & \vdots \\
			y^{r,1}  & y^{r,2} &  c_{1}^{n_\eps} &y^{r,1}  & y^{r,2} & c_{2}^{n_\eps} & \cdots &y^{r,1}  & y^{r,2}  &  c_{r}^{n_\eps}
		\end{pmatrix}_{(r\times 3r)}
	\end{equation}
	
	For $i\in [r]$, let $R_i(\bB)$ be the concatenation of elements in the $i$-th row of $\bB$. We can set $y$ to be the concatenation of rows of $\bB$. Thus, 
	
	\begin{equation*}
		y  = R_1(\bB)\circ R_2(\bB) \circ \cdots \circ C_{r}(\bB)\in \Sigma^{2r^2n_\eps }. 
	\end{equation*}

	We now show the following Claim.
	
	\begin{claim}
		\label{claim:reduction}
		If $g(B) = 0$, $\lcs(x,y)\leq (\frac{5}{2}r-1)n_\eps+5r$. If $g(B) = 1$, $\lcs(x,y)\geq (\frac{5}{2}r-1)n_\eps+5r + \beta r-1$. 
	\end{claim}

	\begin{proof}
		We first divide $x$ into $r$ blocks such that 
		
		\[x = x^1\circ x^2\circ \cdots x^r\]
		where $x^i = C_{3i-2}(\tB)\circ C_{3i-1}(\tB)\circ C_{3i}(\tB)$. We know $C_{3i-2}(\tB) = C_{2i-1}(B)$ is the $(2i-1)$-th column of $B$, $C_{3i-1}(\tB) = C_{2i}(B)$ is the $2i$-th column of $B$ and $C_{3i}(\tB) = c_i^{rn_\eps}$ is symbol $c_i$ repeated $rn_\eps$ times.
		
		If $g(B) = 0$, we show $\lcs(x,y)\leq (\frac{5}{2}r-1)n_\eps+5r$.  We consider a matching between $x$ and $y$. For our analysis, we let $t_i$ be the largest integer such that some symbols in $x^i$ is matched to symbols in the $t_i$-th row of $\bB$. Since $y$ is the concatenation of rows of $\bB$. If no symbols in $x^i$ is matched, we let $t_i = t_{i-1}$ and we set $t_0 = 1$, we have 
		
		\[1 = t_0 \leq t_1\leq t_2\leq \cdots \leq t_r\leq r\]
		
		We now show that, there is an optimal matching such that $x^i$ is matched to at most $n_\eps/2 + 5 +(t_i-t_{i-1}+1)n_\eps$ symbols in $y$. There are two cases:
		
		\begin{description}
			\item{Case (a):} $t_i = t_{i-1}$. In this case, $x^i$ can only be matched to $t_i$-th row of $\bB$, $R_{t_i}(\bB)$. $R_{t_i}(\bB)$ consists of symbols $a_{t_i}, b_{t_i}$ and $c_1,\dots, c_r$ and is of the form
			
			\[R_{t_i}(\bB)  = y^{a_{t_i},b_{t_i}}\circ c_1^{n_\eps}\circ y^{a_{t_i},b_{t_i}}\circ c_2^{n_\eps}\circ \cdots \circ y^{a_{t_i},b_{t_i}}\circ c_r^{n_\eps}.\]
			
			We first show that we can assume $a_{t_i}$ and $b_{t_i}$ symbols in $x^i$ are only matched to the $y^{a_{t_i}, b_{t_i}}$ block between the block of $c_{i-1}$ and the block of $c_{i}$.
			
			If in $R_{t_i}(\bB)$, there are some $ a_{t_i}$, $b_{t_i}$ symbols before the block of $c_{i-1}$ matched to $x^i$. Say the number of such matches is at most $n'$. We know $n'\leq n_\eps$ there are at most $n_\eps$ $ a_{t_i}$, $b_{t_i}$ symbols in $x^i$. In this case, notice that, there is no $c_{i-1}$ symbol in $R_{t_i}(\bB)$ can be matched to $x$. We can build another matching between $x$ and $y$ by removing these $n'$ matches and add $n_\eps$ matches between $c_{i-1}$ symbols in $x^{i-1}$ and $R_{t_i}(\bB)$. We can do this since there are $rn_\eps\;\;$ $c_{i-1}$ symbols in $x^{i-1}$ and before the $t_i$-th row, we can match at most $(t_i-1)n_\eps\;\;$ $c_{i-1}$ symbols. So there are at least $(r-t_i+1)n_\eps$ unmatched $c_{i-1}$ symbols in $x^{i-1}$. The size of the new matching is at least the size of the old matching.
%
			
			Similarly, if there are some $ a_{t_i}$, $b_{t_i}$ symbols after the block of $c_{i}$ in $R_{t_i}(y)$ matched to $x^i$. Then, no $c_i$ symbol in $x^i$ can be matched to $R_{t_i}(\bB)$. We can remove these matches and add $n_\eps$ matched $c_i$ symbols. This gives us a matching with size at least the size fo the old matching. 
			
			Thus, we only need to consider the case where $B_{t_i, 2i-1}\circ B_{t_i, 2i} $ is matched to the part of $R_{t_i}(\bB)$ after the block of $c_{i-1}$ symbols and before the block of $c_i$ symbols, which is $y^{a_{t_i},b_{t_i}}$. Since $g(B) = 0$, we know $\lcs(B_{t_i, 2i-1}\circ B_{t_i, 2i}, y^{a_{t_i}, b_{t_i}})$ is exactly $n_\eps/2+5$. Also, we can match at most $n_\eps$ $c_i$ symbols. Thus, $x^i$ is matched to at most $n_\eps/2 + 5 +n_\eps$ symbols in $y$.
			
			
			\item{Case (b):} $t_i > t_{i-1}$. We can assume in $x^i$, except $c_i$ symbols, only symbols $a_{t_{i-1}}, b_{t_{i-1}}$ are matched to $y$. To see this, assume $t'$ with $t_{i-1}<t'\leq t_i$ is the largest integer that some symbol $a_{t'} $ or $b_{t'}$ in $x^i$ are matched to $y$. By $a, b$ symbols, we mean symbols $a_1, \dots, a_r$ and $b_1, \dots, b_r$. We now consider how many $a, b$ symbols in $x^i$ can be matched to $y$. We only need to consider the substring
			
			\[B_{t_{i-1}, 2i-1}\circ B_{t_{i-1}+1, 2i-1}\circ \cdots \circ B_{t', 2i-1}\circ B_{t_{i-1}, 2i}\circ B_{t_{i-1}+1, 2i}\circ \cdots \circ B_{t', 2i}\]
			
			Let $t_{i-1}\leq k \leq t'$ be the largest integer such that some symbol $a_k$ or $b_k$ from $B_{k, 2i-1}$ is matched. Notice that for any $t\in[r]$, symbol $a_t$, $b_t$ only appears in the $t$-th row of $\bB$ and $y$ is the concatenation of row of $\bB$. For $a_t, b_t$ with $t< k$, only those in $B_{t, 2i-1}$ can be matched since block $B_{t, 2i}$ is after $B_{k, 2i-1}$ in $x^i$.  For $a_t, b_t$, with $t>k$, only those in $B_{t, 2i}$ can be matched by assumption on $k$. 
			
			Notice that for any $t$, we have $B_{t, 2i-1}$ has length $n_\eps/2+5$ and  $B_{t, 2i}$ has length $n_\eps/2-5$. Thus, the number of matched $a, b$ symbols is at most $(t'-t_{i-1})(n_\eps/2+5)+n_\eps$. We can build another matching by first remove the matches of $a, b$ symbols in $x^i$. Then, we match another $(t'-t_{i-1})n_\eps\;$ $c_{i-1}$ symbols in $x^{i-1}$ to the $c_{i-1}$ symbols in $\bB$ from $(t_{i-1}+1)$-th row to $t'$-th row. These $c_{i-1}$ symbols are not matched since we assume $x^{i-1}$ is only matched to first $t_{i-1}$ rows of $\bB$ in the original matching. Further, we can match  $B_{t', 2i-1}\circ B_{t', 2i}$ to the $y^{a_{t'}, b_{t'}}$ block in $R_{t'}(\bB)$ between the block of $c_{i-1}$ and the block of $c_i$. This gives us $(t'-t_{i-1})n_\eps + n_\eps/2+5$ additional matches. Since we $t'>t_{i-1}$, we know the number of added matches is at least the number of removed matches. In the new matching, $t_{i-1} = t'$. Thus, we can assume in $x^i$, except $c_i$ symbols, only symbols $a_{t_{i-1}}, b_{t_{i-1}}$ are matched to $y$.
			
			By the same argument in Case (a), we can assume $a_{t_i}$ and $b_{t_i}$ symbols in $x^i$ are only matched to the $y^{a_{t_i}, b_{t_i}}$ block between the block of $c_{i-1}$ and the block of $c_{i}$. Thus, we can match $n_\eps/2+5\;$ $a_{t_i}$ and $b_{t_i}$ symbols. Also, we can match at most $(t_i-t_{i-1}+1)n_\eps\;$ $c_i$ symbols since there are this many $c_i$ symbols in $\bB$ from $t_{i-1}$-th row to $t_i$-th row. Thus, $x^i$ is matched to at most $n_\eps/2 + 5 +(t_i-t_{i-1}+1)n_\eps$ symbols in $y$.


		\end{description}
		
		In total, the size of matching is at most 
		\[\sum_{i = 1}^r \big(n_\eps/2 + 5 +(t_i-t_{i-1}+1)n_\eps\big)  \leq (2r-1)n_\eps + (n_\eps/2 + 5)r  = (\frac{5}{2}r-1)n_\eps+5r\]
		
		Thus, if $g(B) = 0$, we know $ \lcs(x,y)\leq (\frac{5}{2}r-1)n_\eps+5r$
		
		If $g(B) =1$, that means there is some row $i$ of $B$ such that for at least $\beta r $ positions $j\in [r]$, we have $\lcs(B_{i,2j-1}\circ B_{i, 2j}, y^{a_i, b_i}) \geq n_\eps/2+6$. 
		
		We now build a mathing between $x$ and $y$ with size at least $(\frac{5}{2}r-1)n_\eps+5r + \beta r-1$. We first match $B_{1,1}\circ B_{1,2}$, which is a subsequence of $C_1(\tB)\circ C_2(\tB)$, to the first $y^{a_1, b_1}$ block in $R_1(\bB)$, this gives us at least $n_\eps/2+5$ matches. Then, we match all the $c_1$ symbols in the first $i$ rows of $\bB$ to $C_3{\tB}$. This gives us $in_\eps$ matches.  
		
		We consider the string 
		
		\[\tilde{x} = B_{i,3}\circ B_{i,4}\circ \circ c_2^{n_\eps} \circ \cdots \circ B_{i,2r-1}\circ B_{i,2r} \circ c_r^{n_\eps}.\] 
		
		It is a subsequence of $x^2\circ \cdots x^r$. Also, for at least $\beta r -1$ positions $j\in \{2,3,\dots, r\}$, we know $\lcs(B_{i,2j-1}\circ B_{i, 2j}, y^{a_i, b_i}) \geq n_\eps/2+6$. For the rest of the positions, we know $\lcs(B_{i,2j-1}\circ B_{i, 2j}, y^{a_i, b_i}) = n_\eps/2+5$. Thus, $\lcs(\tilde{x}, R_i(\bB)) \geq (r-1)\big(n_\eps + (n_\eps/2+5)\big) + \beta r -1$. 
		
		After the $i$-th row of $\bB$, we can match another $(r-i)n_\eps\;$ $c_r$ symbols to $x^r$. This gives us a matching of size at least $(\frac{5}{2}r-1)n_\eps+5r + \beta r-1$. Thus, if $g(B) = 1$, $\lcs(x,y)\geq (\frac{5}{2}r-1)n_\eps+5r + \beta r-1$.

	\end{proof}

	Assume $n_\eps = \lambda /\eps$ where $\lambda $ is some constant.  Let $\eps' = \frac{\beta}{10\lambda}\eps = \Theta(\eps)$. If we can give a $1+\eps'$ approximation of $\lcs(x,y)$, we can distinguish $g(B) = 0$ and $g(B) = 1$. 
	
	Thus, we can reduce computing $g(B)$ in the $2r$ player setting to computing $\lcs(x,y)$. The string $y$ (the offline string) is known to all players and it contains no information about $B$. For the $i$-th player, it holds $i$-th column of $B$. If $i$ is odd, the player $i$ knows the $3\frac{i-1}{2}$-th column of $\tB$. If $i$ is even, the player $i$ knows the $(\frac{3i}{2}-1)$-th row of $\tB$ which is the $i$-th row of $B$ and $(\frac{3i}{2})$-th row of $\tB$ which consist of only $c_{\frac{3i}{2}}$.

\end{proof}

\section{Lower Bounds for LIS and LNS}
\label{sec:lis}

In this section, we introduce our space lower bound for $\lis$ and $\lns$. 

%
%
%
%
%
%
%

Let $a \in \{0,1\}^t$, $a$ can be seen as a binary string of length $t$. For each integer $l\geq 1$, we can define a function $h^{(l)}$ whose domain is a subset of $\{0,1\}^t$. Let $\alpha \in (1/2,1)$ be some constant. We have following definition


\begin{equation}
	h^{(l)}(a) = \begin{cases}
		1, \text{ if there are at least $l$ zeros between any two nonzero positions in $a$.}\\
		0, \text{ if $a$ contains at least $\alpha t$ nonzeros.}
	\end{cases}
\end{equation}

We leave $h^{(l)}$ undefined otherwise. Let $B\in \{0,1\}^{s\times t}$ be a matrix and denote the $i$-th row of $B$ by $R_i(B)$. We can define $g^{(l)}$ as the direct sum of $s$ copies of $h^{(l)}$. Let

\begin{equation}
	\label{eq1}
	g^{(l)}(B) = h^{(l)}(R_1(B))\vee h^{(l)}(R_2(B)) \vee \cdots \vee h^{(l)}(R_s(B)).
\end{equation}

That is, $g^{(l)}(B) = 1$ if there is some $i\in [s]$ such that  $ h^{(l)}(R_i(B)) = 1$ and $g^{(l)}(B) = 0$ if for all $i\in[s]$, $ h^{(l)}(R_i(B)) = 0$.

In the following, we consider computing $h^{(l)}$ and $g^{(l)}$ in the $t$-party one-way communication model. When computing $h^{(l)}(a)$, player $P_i$ holds the $i$-th element of $a\in \{0,1\}^t$ for $i\in [t]$. When computing $g^{(l)}(B)$, player $P_i$ holds the $i$-th column of matrix $B$ for $i\in [t]$. In the following, we use $CC^{tot}_{t}(h^{(l)})$  to denote the total communication complexity of $h^{(l)}$ and respectively use $CC^{tot}_{t}(g^{(l)})$ to denote the total communication complexity of $g^{(l)}$. We also consider multiple rounds of communication and we denote the number of rounds by $R$.

\begin{lemma}
	\label{lem:k-fooling-h}
	For any constant $l\geq 1$, there exists a constant $k$ (depending on $l$), such that there is a $k$-fooling set for function $h^{(l)}$  of size $c^t$ for some constant $c>1$. 
\end{lemma}

We note that Lemma 4.2 of \cite{ergun2008distance} proved a same result for the case $l=1$. 

\begin{proof}[Proof of Lemma~\ref{lem:k-fooling-h}]
	We consider sampling randomly from $\{0,1\}^t$ as follows. For $i\in [t]$, we independently pick $a_i  = 1$ with probability $p$ and $a_i = 0$ with probability $1-p$. We set $p = \frac{1}{k}$ for some large constant $k$. For $i\in [t-l]$, we let $A_i$ to be event that there are no two 1's in the substring $a[i:i+l]$. Let $\Pr(A_i)$ be the probability that event $A_i$ happens. By a union bound, we have 
	\begin{equation}
		\Pr(A_i) \leq \binom{l}{2} p^2, \; \forall  \; i\in [t-l]
	\end{equation}
	
	Let $\mathcal{A} = \{A_i, \; i\in [t-l]\}$. Notice that since we are sampling each position of $a$ independently, the event $A_i$ is dependent to at most $2l$ other events in $\mathcal{A}$. We set $v = 2l\binom{l}{2} p^2$. For large enough $k$, we have 
	\begin{equation}
		\forall \; i \in [t-l], \; \Pr(A_i)\leq \binom{l}{2} p^2 \leq v (1-v)^{2l}
	\end{equation} 
	Here, the second inequality follows from the fact that $l$ is a constant and $\binom{l}{2}p^2 = v/(2l)$, so we can pick $k$ to be large enough, say $k \ge \sqrt{\frac{l^3}{1-log(2l)/(2l)}}$ (or $p = 1/k$ to be small enough) to guarantee $(1-v)^{2l} \ge 1/{2l}$.  
	
	Thus, we can use Lov\'as Local Lemma here. By Lemma~\ref{lem:LLL}, we have 
	\begin{equation}
		\Pr\big(\overline{A_1}\wedge \overline{A_2} \wedge \cdots \wedge \overline{A_{t-l}}\big) \geq (1-v)^t \geq (1-l^3 p^2)^t.
	\end{equation}
	Notice that ``there are at least $l$ $0$'s between any two $1$'s in $a$"  is equivalent to none of events $A_i$ happens. We say a sampled string $a$ is good if none of $A_i$ happens. Thus, for any good string $a$, we have $h^{(l)}(a) = 0$. The probability that a sampled  string $a$ is good is at least $(1-l^3 p^2)^t$. For convenience, we let $q = 1-l^3p^2$.
	
	Assume we independently sample $M$ strings in this way, the expected number good string is $ q^tM $. Let $a^1, a^2, \dots, a^k$ be $k$ independent random samples. We consider a string $b$ in the span of these $k$ strings, such that, for $i\in [t]$, let $b_i = 1$ if there is some $j\in[k]$ such that $a^j_i = 1$. $b$ is in the span of these $k$ strings. We now consider the probability that $b$ has at least $\alpha t $ 1's, i.e. $h^{(l)}(b) = 1$. Notice that $a^j_i = 1$ with probability $p$, thus $Pr(b_i = 1) = 1-(1-p)^t$.  Let $\gamma = 1-(1-p)^t$. The expected number of 1's in $b$ is $\gamma t$. Let $\eps = \gamma-\alpha$ and $\delta$ be the probability that $b$ has less than $\alpha t$ 1's. Using Chernoff bound, we have, 
	\begin{equation}
		\Pr(\text{$b$ has less than $\alpha t $ 1's}) \leq  e^{\frac{-\eps^2 \gamma }{2}t}.
	\end{equation}
	Let $\delta = e^{\frac{-\eps^2 \gamma }{2}t}$ and $M = \frac{e}{k}(\frac{q^t}{2\delta})^{1/k} $. We consider the probability that these $M$ sample is not a $k$-fooling set for $h^{(l)}$. Since for any $k$ samples, it is not a $k$-fooling set with probability at most $\delta$. Let $ E$ denote the event that these $M$ samples form a $k$-fooling set. Using a union bound, the probability that $E$ does not happen is
	\begin{align*}
		\Pr\big(\bar{E}\big)\leq \binom{M}{k}\delta \leq (\frac{eM}{k})^k\delta  \leq \frac{1}{2}q^t.
	\end{align*}
	
	Let $Z$ be a random variable equals to the number of good string among the $M$ samples. As we have shown, the expectation of $Z$, $ \mathbb{E}(Z) = q^tM$. Also notice that $\mathbb{E}(Z) = \mathbb{E}(Z|E)\Pr(E) + \mathbb{E}(Z|\overline{E})\Pr(\overline{E})$. Thus, with a positive probability, there are $\frac{1}{2}q^t M$ good samples and they form a $k$-fooling set. 
	\begin{equation}
		\frac{1}{2}q^t M =( \frac{1}{2})^{1+1/k}\frac{e}{k}(q^{1+1/k}(\frac{1}{\delta})^{1/k})^t.
	\end{equation}
	Notice that 
	\begin{equation}
		q^{1+1/k}(\frac{1}{\delta})^{1/k} =(1-\frac{l^3}{k^2})^{1+1/k} e^{\frac{\eps^2\gamma}{2k}}.
	\end{equation}
	Since we assume $l$ is a constant, it is larger than 1 when $k$ is a constant large enough (depends on $l$). This finishes the proof. 
	
\end{proof}

%

The following lemma is essentially the same as Lemma 4.3 in \cite{ergun2008distance}.
\begin{lemma}
	
	\label{lem:ks-fooling-g}
	Let $F\subseteq \{0,1\}^{t}$ be a $k$-fooling set for $h^{(l)}$. Then the set of all matrix $B\in \{0,1\}^{s\times t}$ such that $R_i(B)\in F$ is a $k^s$-fooling set for $g^{(l)}$. 
\end{lemma}

\begin{lemma}	
	\label{lem:g}
	$CC^{max}_t(g^{(l)}) = \Omega(s/R).$  
\end{lemma}

\begin{proof}
	By Lemma~\ref{lem:k-fooling-h} and Lemma~\ref{lem:ks-fooling-g}, there is a $k^s$-fooling set for function $g^{(l)}$ of size $c^{ts}$ for some large enough constant $k$ and  some constant $c>1$. By Lemma~\ref{lem:k-fooling}, in the $t$-party one-way communication model, $ CC^{tot}_t(g^{(l)}) = \Omega(\log \frac{c^{ts}}{k^s-1})= \Omega(ts).$  Thus, we have $CC^{max}_t(g^{(l)}) \geq \frac{1}{tR}CC^{tot}_t(g^{(l)}) =  \Omega(s/R).$  
\end{proof}


\subsection{Lower bound for streaming LIS over small alphabet}

We now present our space lower bound for approximating LIS in the streaming model. 



\begin{lemma}
	\label{lem:lis_length}
	For $x\in\Sigma^n$ with $|\Sigma| = O(\sqrt{n})$ and any constant $\eps>0$,  any deterministic algorithm that makes $R$ passes of $x$ and outputs a $(1+\eps)$-approximation of $\lis(x)$ requires $\Omega(|\Sigma|/R)$ space. 
\end{lemma}
\begin{proof}[Proof of Lemma~\ref{lem:lis_length}]
	
	We assume the alphabet set $\Sigma =  \{0,1,\dots, 2r\}$ which has size $|\Sigma| = 2r+1$. Let $c$ be a large constant and assume $r$ can be divided by $c$ for similicity.  We set $s = \frac{r}{c}$ and $t = r$. Consider a matrix $B$ of size $s\times t$. We denote the element on $i$-th row and $j$-th column by $B_{i,j}$. ALso, we require that $B_{i,j} $ is either $ (i-1)\frac{r}{c} + j$ or $0$. For each row of $B$, say $R_i(B)$, either there are at least $l$ 0's between any two nonzeros or it has more than $\alpha r$ nonzeros. We let $\tilde{B}\in \{0,1\}^{s\times r}$ be a binary matrix such that $\tilde{B}_{i,j} = 1$ if $B_{i,j}\neq 0$ and $\tilde{B}_{i,j} = 0$ if $B_{i,j} = 0$ for $(i,j)\in [s]\times[r]$. 
	
	Without loss of generality, we can view $R_i(B)$ for $i\in [s]$, or $C_i(B)$ for $i\in [r]$ as a string. More specifically, $R_i(B) = B_{i,1}B_{i,2}\dots B_{i,r}$ for $i\in [s]$, and $C_i(B) = B_{1,i}B_{2,i}\dots B_{s,i}$ for $i\in[r]$.
	
	
	We let $\sigma(B) = C_1(B)\circ C_2(B) \circ \cdots \circ C_r(B)$. Thus, $\sigma(B)$ is a string of length $sr$. For convenience, we denote $\sigma = \sigma(B)$. Here, we require the length of $\sigma = r^2/c \leq n$. If $|\sigma|<n$, we can pad $\sigma$ with 0 symbols to make it has length $n$. This will not affect the length of the longest increasing subsequence of $\sigma$.
	
	We first show that if there is some row of $B$ that contains more than $\alpha t$ nonzeros, then $\lis(\sigma)\geq \alpha r $. Say $R_i(B)$ contains more than $\alpha t$ nonzeros. By our definition of $B$, $R_i(B)$ is strictly increasing when restricted to the nonzero positions. Thus, $\lis(R_i(B))\geq \alpha r$. Also notice that $R_i(B)$ is a subsequence of $\sigma$. This is because $C_j(B)$ contains element $B_{i,j}$ for $j\in [r]$ and $\sigma$ is the concatenation of $C_j(B)$'s for $j$ from 1 to $r$. Thus, $\lis(\sigma)\geq \alpha r $.

	Otherwise, for any row $R_i(B)$, there are at least $l$ zeros between any two nonzero positions. We show that $\lis(\sigma(B))\leq (\frac{1}{r}  + \frac{1}{c})r$. Assume $\lis(\sigma) = m$ and let $\sigma' = B_{p_1, q_1} B_{p_2, q_2} \dots B_{p_m, q_m}$ be a longest increasing subsquence of $\sigma$.
	
	We can think of $\sigma'$ as a path on the matrix. By go down $k$ steps from $(i,j)$, we mean walk from $(i,j) $ to $(i+k, j)$. By go right $k$ steps form $(i,j)$, we mean walk from $(i,j)$ to $(i,j+k)$. Then $\sigma'$ corresponds to the path $P(\sigma') = (p_1, q_1)\rightarrow (p_2, q_2)\rightarrow \cdots \rightarrow (p_m, q_m)$. Notice that for each step $(p_i, q_i)\rightarrow (p_{i+1}, q_{i+1})$, we know $1\leq B_{p_i, q_i}< B_{p_{i+1}, q_{i+1}}$. Thus, by our construction of matrix $B$, the step can only be one of the three types:
	
	\begin{description}
	\item{Type 1:} $p_i< p_{i+1}$ and $q_{i+1}\geq  q_{i}$.  If  $p_i< p_{i+1}$, then for any $q_{i+1}\geq  q_{i}$, we have $B_{p_i, q_i}< B_{p_{i+1}, q_{i+1}}$ if $B_{p_i, q_i}$ and $ B_{p_{i+1}, q_{i+1}}$ are both non zero,
	
	
	\item{Type 2:} $p_i = p_{i+1} $ and $q_{i+1}-q_i\geq l+1$.This correspons to the case where $B_{p_i, q_i}$ and $B_{p_{i+1}, q_{i+1}}$ are picked from the same row of $B$. However, since we assume in each row of $B$, the number of 0's between any two nonzeros is at least $l$. Since $B_{p_i, q_i}$ and $B_{p_{i+1}, q_{i+1}}$ are both nonzeros and $B_{p_i, q_i}< B_{p_{i+1}, q_{i+1}}$, we must have $q_{i+1}-q_i\geq l+1$.
	
	\item{Type 3:} $p_i>p_{i+1}$ and $q_{i+1}-q_{i}\geq (p_{i}-p_{i+1})\frac{r}{c} + 1$. When $p_i> p_{i+1}$,  $B_{p_{i},q_i}-B_{p_{i+1},q_i}> (p_i-p_{i+1})\frac{r}{c}$. Since we require $B_{p_i, q_i}< B_{p_{i+1}, q_{i+1}}$, we must have $q_{i+1}-q_{i}\geq (p_{i}-p_{i+1})\frac{r}{c} + 1$. 
	\end{description}

	For $i \in [3]$, let $a_i$ be the number of step of Type $i$ in the path $P(\sigma')$. Then, $a_1+a_2 + a_3 = m-1$. We say $(p_i, q_i)\rightarrow (p_{i+1}, q_{i+1})$ is step $i$ for $i\in [m-1]$. And let 
	\[U_i = \{j\in [m-1]| \text{step $j$ is of Type $i$}\} \text{ for } i\in[3]. \] 
	$$a'_1=  \sum_{i\in U_1}(p_{i+1}-p_i)$$
	$$a'_3 = \sum_{i\in U_3}(p_{i}-p_{i+1}).$$
	Or equivalently, $a'_1$ is the distance we go downward with steps of Type 1 and $a'_2$ is the distance we go upward with steps of Type 3. Since only steps of Type 1 and Type 3 can go up and down, we know 
	\begin{equation}
		\label{eq:lis_lb_1}
		a'_1 -a'_3 = p_m-p_1 \leq \frac{r}{c}
	\end{equation} 
	For the number of distance we go right, for each step of Type 2, we go at least ${l+1}$ positions right. For the step of Type 3, we go right at least $\sum_{i \in T_3 }((p_{i}-p_{i+1})\frac{r}{c} + 1)  = \frac{r}{c}a'_3 + |T_3| = \frac{r}{c}a'_3 + a_3$ steps. Since the total distance we go right is $q_m-q_1\leq r$. Thus, we have
	\begin{equation}
		\label{eq:lis_lb_2}
		a_2 l + a'_3 \frac{r}{c} + a_3\leq q_m-q_1\leq r.
	\end{equation}
	
	We assume $l$ and $c$ are both constants and $\frac{r}{cl}\geq 1$. Notice that $a_1\leq a'_1$ and $a_3\leq a'_3$, combining \ref{eq:lis_lb_1} and \ref{eq:lis_lb_2}, we have 
	$$\lis(\sigma(B)) = a_1+a_2+a_3 + 1 \leq \frac{r}{c} + \frac{r}{l}.$$
	
	This show that if $g^{(l)}(\tilde{B}) = 0$, we have $\lis(\sigma(B))\geq \alpha r$. And if $g^{(l)}(\tilde{B}) = 1$, $\lis(\sigma(B)) \leq (\frac{1}{c} + \frac{1}{l})r$. Here, $c$ and $l$ can be any large constant up to our choice and $\alpha\in(1/2, 1)$ is fixed. For any $\eps>0$, we can choose $c$ and $l$ such that $(1+\eps)(\frac{1}{c} + \frac{1}{l})\leq \alpha$. This gives us a reduction from computing $g^{(l)}(\tilde{B})$ to compute a $(1+\eps)$-approximation of $\lis(\sigma(B))$.

	In the $t$-party game for computing $g^{(l)}(\tB)$, each player holds one column of $\tB$. Thus, player $P_i$ also holds $C_i(B)$ since $C_i(B)$ is determined by $C_i(\tB)$. If the $t$ players can compute a $(1+\eps)$ approximation of $\sigma(B)$ in the one-way communication model, we can distinguish the case of $g^{(l)}(\tB)=  0$ and $g^{(l)}(\tB)=1$. Thus, any $R$ passes deterministic streaming algorihtm that approximate $\lis$ within a $1+\eps$ factor requires at least $CC^{max}_t(g^{(l)})$. By Lemma~\ref{lem:g}, $CC^{max}_t(g^{(l)}) = \Omega(s/R) =  \Omega(|\Sigma|/R)$.

\end{proof}

%
%
%
%
%

\subsection{Longest Non-decreasing Subsequence}
We can proof a similar space lower bound for approximating the length of longest non-decreasing subsequence in the streaming model.

\begin{lemma}
	\label{lem:lns_length}
	For $x\in\Sigma^n$ with $|\Sigma| = O(\sqrt{n})$ and any constant $\eps>0$, any deterministic algorithm that makes $R$ passes of $x$ and outputs a $(1+\eps)$-approximation of $\lns(x)$ requires $\Omega(|\Sigma|/R)$ space. 
\end{lemma}

\begin{proof}[Proof of Lemma~\ref{lem:lns_length}]
	In this proof, we let the alphabet set $\Sigma = \{1,2,\dots, (c+1)r\}$. The size of the alphabet is $(c+1)r$. Without loss of generality, we can assume $cr^2\leq n$.  
	
	Let $\tilde{B}\in \{0,1\}^{s\times t}$ be a binary matrix such that for any row of $\tilde{B}$, say $R_i(\tilde{B})$ for $i\in [s]$, either there are at least $l$ 0s between any two 1's, or, $R_i(\tilde{B})$ has at least $\alpha t$ 1's. 
	
	Similar to the proof of Lemma~\ref{lem:lis_length}, we show a reduction from computing $g((\tilde{B}))$ to approximating the length of $\lns$. In the following, we set $s = r$ and $t = cr$ for some constant $c>0$ such that $t$ is an integer. 
	
	Let us consider a matrix $B\in \Sigma^{s\times t}$ such that for any $(i,j)\in [s]\times[t]$:
	\begin{equation}
		B_{i,j} = \begin{cases}
			i, &\text{ if $\tilde{B}_{i,j} = 1$,} \\
			cr+r + 1-j, &\text{ if $\tilde{B}_{i,j} = 0$.} 
		\end{cases}
	\end{equation}

	Thus, for all positions $(i,j)$ such that  $\tilde{B}_{i,j} = 0$, we know $B_{i,j} > r$. Also, for $1\leq j< j'\leq cr$, assume $\tilde{B}_{i,j} = \tilde{B}_{i',j'}$ for some $i, i'\in [r]$, we have $B_{i,j}> B_{i',j'}$. For positions $(i,j)$ such that $\tilde{B}_{i,j} = 1$, we have $B_{i,j} = i$.

	Consider the sequence $\sigma(B) = C_1(B)\circ C_2(B) \circ \cdots \circ C_r(B)$ where $C_i(B) = B_{1,i}B_{2,i}\dots B_{r,i}$ is the concatenation of symbols in the $i$-th column of matrix $B$. 
	
	We now show that if $g^{(l)}(\tilde{B}) = 0$, $\lns(\sigma)\leq 2r + \frac{cr}{l}$ and if $g(\tilde{B}) = 1$, $\lns(\sigma) \geq \alpha cr$. \footnote{Here, we assume $n  = |\sigma| = cr^2$. If $n>cr^2$, we can repeat each symbol in $\sigma$  $\frac{n}{cr^2}$ times and show $g^{(l)}(\tilde{B}) = 0$, $\lns(\sigma)\leq \Big(2r + \frac{cr}{l}\big)\frac{n}{cr^2}$ and if $g^{(l)}(\tilde{B}) = 1$, $\lns(\sigma) \geq \big(\alpha cr\big)\frac{n}{cr^2}$. The proof is the same.}
	
	If $g^{(l)}(\tilde{B}) = 0$, consider a longest non-decreasing subsequence of $\sigma$ denoted by $\sigma'$. Than $\sigma'$ can be divided into two parts $\sigma'^1$ and $\sigma'^2$ such that $\sigma'^1$ consists of symbols from $[r]$ and $\sigma'^2$ consists of symbols from $\{r+1, r+2, \dots, cr\}$. Similar to the proof of Lemma~\ref{lem:lis_length}, $\sigma'^1$ corresponds to a path on matrix $\tilde{B}$. Since we are concatenating the columns of $B$, the path can never go left. Each step is either go right at least $l+1$ positions since there are at least $l$ 0's between any two 1's in the same row of $\tilde{B}$, or, go downward to another row. Thus, the total number of steps is at most $\frac{t}{l}+ r$ since $\tilde{B}$ has $r$ rows and $t$ columns. For $\sigma'^2$, if we restricted $B$ to positions that are in $\{r+1, r+2, \dots, (c+1)r\}$, symbols in column $j$ of $B$ must be smaller than symbols in column $j'$ if $j<j'$. Thus, the length of $\sigma'^2$ must be at most the length of $C_j(B)$ for any $j\in [cr]$, which is at most $r$. Thus the length of $\sigma$ is at most $2r + \frac{cr}{l}$.  
	
	If $g^{(l)}(\tilde{B}) = 1$, then we know there is some $i \in [r]$ such that row $i$ of $\tilde{B}$ constains at least $\alpha cr$ 1's. We know $B_{i,j} = i$ if $\tilde{B}_{i,j} = 1$. Thus, $R_i(B)$ contains a non-decreasing subseqeunce of length at least $\alpha cr$. Since $R_i$ is a subsequence of $\sigma$. We know $\lns(\sigma)\geq \alpha c r$.
	
	For any constant $\eps>0$, we can pick constants $c,l >1$ and $\alpha\in (1/2 , 1)$ such that $ (1+\eps)(2+c/l) \leq \alpha c$. Thus, if we can approximate $\lns(\sigma)$ to within a $1+\eps$ factor, we can distinguish the case of $g^{(l)}(\tilde{B}) = 0$ and $g^{(l)}(\tilde{B}) = 1$. The lower bound then follows from Lemma~\ref{lem:g}. 
	

\end{proof}

\begin{lemma}
	\label{lem:lns_length_eps}
	Let $x\in \Sigma^n$ and $\eps>0$ such that $|\Sigma|^2/\eps = O(n)$. Then any deterministic algorithm that makes constant pass of $x$ and outputs a $(1+\eps)$ approximation of $ \lns(x)$ takes $\Omega(r\log \frac{1}{\eps})$ space. 
\end{lemma}

\begin{proof}
	Let the alphabet set $\Sigma = \{a_1, a_2, \dots, a_r\} \cup \{b_1, b_2, \dots, b_r\}$ and assume that 
	\[b_r<b_{r-1}<\cdots<b_1 <a_1<a_2<\cdots<a_r\].
	
	We assume $t\geq 2r/\eps$. Since we assume $t = \Omega(r/\eps)$, if $t< 2r/\eps$, we can use less symbols in $\Sigma$ for our construction and this will not affect the result. Let $l = \Theta(1/\eps)$ that can be divided by $4$. For any two symbols $a, b$, we consider the set $A_{a,b}\in \{a,b\}^l$ such that 
	\[A_{a,b} = \{  a^{\frac{3}{4}l+t} b^{\frac{1}{4}l - t}\vert 1\leq t \leq l/4 \} \]
	
	We define a function $f$ such that for any $\sigma = a^{\frac{3}{4}l+t} b^{\frac{1}{4}l - t}\in A_{a,b}$, $f(\sigma) = a^{\frac{3}{4}l - t} b^{\frac{1}{4}l +t}$. Thus, for any $\sigma\in A_{a,b}$, the string $\sigma\circ f(\sigma)$ has exactly $ \frac{3}{2}l$ $a$ symbols and $\frac{1}{2}l$ $b$ symbols. 
	
	We let $\overline{A_{a, b}} = \{f(\sigma)\vert \sigma\in A_{a,b}\}$. We know $\lvert A_{a,b}\rvert = \lvert \overline{A_{a, b}} \rvert  = l/2 = \Theta(1/\eps)$.
	
	Consider an error-correcting code $T_{a,b} \subset S_{a,b}^r$ over alphabet set $S_{a,b}$ with constant rate $\alpha$ and constant distance $\beta$. We can pick $\alpha = 1/2$ and $\beta = 1/3$. Then the size of the code $T_{a,b}$ is $|T_{a,b} |= |S_{a,b}|^{\alpha r} = 2^{\Omega(lr)}$.  For any code word $\chi  = \chi_1 \chi_2 \cdots \chi_r \in T_{a,b}$ where $\chi_i\in S_{a,b}$, we can write 
	
	$$\nu (\chi) = (\chi_1, f(\chi_1), \chi_2, f(\chi_2), \dots, \chi_r, f(\chi_r)).$$
	
	Let $W = \{\nu(\chi)| \chi\in T_{a,b}\}\in (A_{a,b}\times \overline{A_{a, b}})^r $. Since the code has constant rate $\alpha$, the size of $W$ is $(l/4)^{\alpha r}$.

	Let $\beta = 1/3$. We can define function $h:(A_{a,b}\times \overline{A_{a, b}})^r \rightarrow \{0,1\}$ such that
	\begin{equation}
		h(w) = \begin{cases}
			0, &\text{ if } \forall i\in [r], \;  f(w_{2i-1}) = w_{2i},  \\
			1, &\text{ if for at least $\beta r$ indices $i\in [r]$, } w_{2i-1}\circ w_{2i} \text{ contains more than $\frac{3}{2}l$ $a$ symbols.} \\
			\text{undefined}, &\text{ otherwise}. 
		\end{cases}
	\end{equation}
	
	\begin{claim}
		\label{claim:W_lnst}
		$W$ is a fooling set for $h$. 
	\end{claim}
	
	\begin{proof}
		Let $w$ and $w'$ be two distinct elements in $W$. Let $w = (w_1, f(w_1), w_2, f(w_2), \dots, w_r, f(w_r))$ and $w' = (w'_1, f(w'_1), w'_2, f(w'_2), \dots, w'_r, f(w'_r))$. By the definition, we know for at least $\beta r$ positions $i\in [r]$,  we have $w_i\neq w'_i$. Also, by the construction of set $S_{a,b}$, if $w_i\neq w'_i$, then one of $w_i\circ f(w'_i)$ and $w'_i\circ f(w_i)$ has more than $\frac{3}{2}l$ $a$ symbols.
		
		$v$ in the span of $w$ and $w'$ if $v = (v_1, \bar{v}_1, \dots, v_r, \bar{v}_r)$ such that $v_i\in \{w_1, w'_i\}$ and $\bar{v}_{i}\in \{f(w_i), f(w'_i)\}$.  We can find a $v$ in the span of $w$ and $w'$ such that $h(v) = 1$. 
		
	\end{proof}
	
	For $i\in [r]$, we can define $A_{a_i, b_i}$, $\overline{A_{a_i, b_i}}$, $W_i$, $h_i$ similarly except the alphabet is $\{a_i, b_i\}$ instead of $\{a,b\}$.

	Consider a matrix $B$ of size $r\times 2r$ such that $B_{i,2j-1}\in A_{a_i,b_i}$ and $B_{i,j}\in \overline{A_{a_i, b_i}}$. We define a function $g$ such that 
	
	\begin{equation}
		g(B) = h_1(R_1(B))\vee h_2(R_2(B))\vee \cdots \vee h_r(R_r(B)).
	\end{equation}  
	
	In the following, we consider a $2r$-party one way game. In the game, player $i$ holds $C_i(B)$. The goal is to compute $g(B)$. 
	
	\begin{claim}
		The set of all matrix $B$ such that $R_i(B)\in W_i$ for $i\in [r]$ is a fooling set for $g$. 
	\end{claim}
	\begin{proof}
		For any two matrix $B_1 \neq B_2$ such that $R_i(B_2), R_i(B_)\in W_i \; \forall\; i\in [r]$. We know $g(B_1) = g(B_2) = 0$. There is some row $i$ such that $R_i(B_1)\neq R_i(B_2)$. We know there is some elements $v$ in the span of $R_i(B_1)$ and $R_i(B_2)$, such that $h_i(v) = 1$ by Claim~\ref{claim:W_lnst}. Thus, there is some element $B' $ in the span of $B_1$ and $B_2$ such that $g(B') = 1$. 
	\end{proof}
	
	Thus, we get a $2$-fooling set for function $g$ in the $2r$-party setting. The size of the fooling set is $|W|^r =(l/4)^{\alpha r^2})$. Thus, $CC^{tot}_{2r}(g) = \log (|W|^r) = \Omega(r^2\log\frac{1}{\eps})$ and $CC^{max}_{2r}(g)  = CC^{tot}_{2r}(g) /2r = \Omega(r\log\frac{1}{\eps})$.
	
	Consider a matrix $\tB$ of size $r\times 3r$ such that $\tB$ is obtained by inserting a column of $a$ symbols to $B$ at every third position. Thus, 
	
	\begin{equation*}
		\tB = \begin{pmatrix}
			B_{1,1}	& B_{1,2} & a_{1}^{2l}  &B_{1,3}	& B_{1,4} & a_{1}^{2l} & \cdots & B_{1,2r-1}  & B_{1,2r} &  a_{1}^{2l}\\
			B_{2,1}	& B_{2,2}& a_{2}^{2l}  &B_{2,3}	& B_{2,4} & a_{2}^{2l} & \cdots & B_{2,2r-1} & B_{2,2r} &  a_{2}^{2l}\\
			\vdots	& \vdots & \vdots  			  &\vdots	& \vdots& \vdots & \ddots & \vdots       &\vdots    & \vdots \\
			B_{r,1}	& B_{r,2} &  a_{r}^{2l} &B_{r,3}	& B_{r,4} & a_{2}^{2l} & \cdots & B_{r,2r-1}  & B_{r,2r}  &  a_{r}^{2l}
		\end{pmatrix}_{(r\times 3r)}
	\end{equation*}

	Let $x = C_1(\tB)\circ C_2(\tB)\circ \cdots \circ C_{3r}(\tB)$.
	
	We now show how to reduce computing $g(B)$ to approximating the length of $\lns(x)$. We claim the following. 
	
	\begin{claim}
		If $g(B) = 0$, $\lns(x)\leq \frac{11}{2}rl-2l$. If $g(B)= 1$, $\lns(x)\geq \frac{11}{2}rl-2l+ \beta r -1$.
	\end{claim}
	
	\begin{proof}
		
		We first divide $x$ into $r$ parts such that
		
		\[x = x^1\circ x^2\circ \cdots x^r\]
		
		where $x^i = C_{3i-2}(\tB)\circ C_{3i-1}(\tB)\circ C_{3i}(\tB)$. We know $C_{3i-2}(\tB) = C_{2i-1}(B)$ is the $(2i-1)$-th column of $B$, $C_{3i-1}(\tB) = C_{2i}(B)$ is the $2i$-th column of $B$ and $C_{3i}(\tB) = a_1^{2l}a_2^{2l}\cdots a_r^{2l}$.
		
		If $g(B) = 0$, we show $\lns(x) = \frac{11}{2}rl-2l$. Let $\Sigma$ be a longest non-decreasing subsequence of $x$. We can divide $\sigma$ into $r$ parts such that $\sigma^i$ is a subsequence of $x^i$. For our analysis, we  set $t_0 = 1$. If $\sigma^i$ is empty or contains no $a$ symbols, let $t_i = t_{i-1}$. Otherwise, we let $t_i$ be the largest integer such that $a_{t_i}$ appeared in $\sigma^i$.We have
		
		\[1 = t_0 \leq t_1\leq t_2\leq \cdots \leq t_r\leq r\]
		
		We now show that, for any $i$, the length of $\sigma^i$ is at most  $\frac{3}{2}l  +2(t_i-t_{i-1}+1)l$. To see this, if $t_i = t_{i-1}$, $x^i$ to $\sigma$. Since $g(B) = 0$, there are exactly $\frac{3}{2}l+2l$ $a_{t_i}$ symbols in $x^i$. Thus $|\sigma^i|\leq \frac{3}{2}l+2l$. 
		
		If $t_i>t_{i-1}$, if there is some $a_t$ symbol in the   $C_{3i-2}(\tB)\circ C_{3i-1}(\tB)$ included in $\sigma^i$ for some $t_{i-1}<t\leq t_i$. Then $\sigma^i$ can not include $a_{t'}$ symbols in $C_{3i}(\tB)$ for all $t_{i-1}\leq t'< t$. Also for any $t_{i-1}\leq t'< t$, the number of $a_{t'}$ symbols in $C_{3i-2}(\tB)\circ C_{3i-1}(\tB)$ is $\frac{3}{2} l$ but the number in $C_{3i}(\tB)$ is $2l$. Thus, the optimal strategy it to pick the $\frac{3}{2}l$ $a_{t_{i-1}}$ symbols in $C_{3i-2}(\tB)\circ C_{3i-1}(\tB)$ and then add $2(t_i-t_{i-1}+1)l$ symbols ($2l$ $a_t$'s for $t_{i-1}\leq t\leq t_i$) in $C_{3i}(\tB)$ to $\sigma^i$.

		In total, length of $\sigma$ is at most 
		\[\sum_{i = 1}^r \sigma^i  \leq \frac{11}{2}rl-2l \]
		
		Thus, if $g(B) = 0$, we know $\lns(x)\leq \frac{11}{2}rl-2l$.
		
		If $g(B) =1$, that means there is some row $i$ of $B$ such that for at least $\beta r $ positions $j\in [r]$, the number of $a_i$ symbols in $B_{i,2j-1}\circ B_{i, 2j}$ is at least $\frac{3}{2}l+1$. 
		
		We now build a non-decreasing subsequence $\sigma$ of $x$ with length at least $\frac{11}{2}rl-2l+\beta r-1$. We set $\sigma$ to be empty initially. There are at least $\frac{3}{2}l$ $a_1$ symbols in $B_{1,1}\circ B_{1,2}$. We add all of them to $\sigma$.   Then, we add all the $a_t$ symbols for $1\leq t\leq i$ in $C_3{\tB}$ to $\sigma$. This adds  $2il$ symbols to $\sigma$
		
		We consider the string 
		
		\[\tilde{x} = B_{i,3}\circ B_{i,4}\circ \circ a_i^{2l} \circ \cdots \circ B_{i,2r-1}\circ B_{i,2r} \circ a_i^{2l}.\] 
		
		It is a subsequence of $x^2\circ \cdots x^r$. There are $(r-1)(\frac{3}{2}l+ 2l +\beta r-1) $ $a_i$ symbols in $\tilde{x}$. This is because for at least $\beta r -1$ positions $j\in \{2,3,\dots, r\}$, we know $B_{i,2j-1}\circ B_{i, 2j}$ contains more than $\frac{3}{2}l$ $a_i$ symbols. For the rest of the positions, we know $B_{i,2j-1}\circ B_{i, 2j}$ contains $\frac{3}{2} l$ $a_i$ symbols. 
		
		Finally, we add all the $a_t$ symbols for $i< t\leq n$ in $C_{3r}{\tB}$ to $\sigma$. This adds another $2(r-i)l$ symbols to $\sigma$. The sequence $\sigma$ has length at least $\frac{11}{2}rl-2l+ \beta r -1 $.
	\end{proof}
	
	Assume $l = \lambda /\eps$ where $\lambda $ is some constant.  Let $\eps' = \frac{\beta}{10\lambda}\eps = \Theta(\eps)$. If we can give a $1+\eps'$ approximation of $\lns(x)$, we can distinguish $g(B) = 0$ and $g(B) = 1$. 
	
	By the fact that $CC^{max}_{2r}(g) = \Omega(r\log \frac{1}{\eps})$, any deterministic streaming algorithm with constant passes of $x$ that approximate $\lns(x)$ within a $(1+\eps')$ factor requires $\Omega(r\log \frac{1}{\eps})$ space. 
\end{proof}

\subsection{Longest Non-decreasing Subsequence with Threshold}

We now consider a variant of $\lns$ problem we call longest non-decreasing subsequence with threshold ($\lnst$). In this section, we assume the alphabet is $\Sigma = \{0,1,2,\dots, r\}$. In this problem, we are given a sequence $x\in \Sigma^n$ and a threshold $t\in [n]$, the longest non-decreasing subsequence with threshold $t$ is the longest non-decreasing subsequence of $x$ such that each symbol appeared in it is repeated at most $t$ times. We denote the length of such a subsequence by $\lnst(x,t)$. 

We first give our lower bounds for $\lnst$.

\begin{theorem}[Restatement of Theorem~\ref{thm:LNSTmain}]
	\label{thm:lnst_lb}
	Given an alphabet $\Sigma$, for deterministic $(1+\eps)$ approximation of $\lnst(x, t)$ for a string $x\in \Sigma^n$ in the streaming model with $R$ passes, we have the following space lower bounds:
	\begin{enumerate}
		
		\item $\Omega(\min(\sqrt{n}, |\Sigma|)/R)$ for any constant $t$ (this includes $\lis$), when $\eps$ is any constant.
		
		\item $\Omega(|\Sigma|\log(1/\eps)/R)$ for $t\geq n/|\Sigma|$ (this includes $\lns$), when $|\Sigma|^2/\eps = O(n)$.
		
		\item $\Omega\left (\frac{\sqrt{|\Sigma|}}{\eps}/R \right ) $ for $t = \Theta(1/\eps)$, when $|\Sigma|/\eps = O(n)$.
	\end{enumerate}
\end{theorem}

Theorem~\ref{thm:lnst_lb} is a result of the following four Lemmas.

\begin{lemma}
	\label{lem:lnst1}
	Let $x\in \Sigma^n$ and $|\Sigma| =  r \leq \sqrt{n})$ and $\eps>0$ be any constant. Let $t \geq n/r$. Then any $R$-pass deterministic algorithm a $1+\eps$ approximation of $\lnst(x,t)$ takes $\Omega(r/R)$ space. 
	
\end{lemma}

\begin{proof}
	
	 In the proof of Lemma~\ref{lem:lns_length}, each symbol in sequence $\sigma$ appeared no more than $\max(r,n/r) $ times. If $r\leq \sqrt{n}$ and $t\geq n/r$, we have $\lns(\sigma) = \lnst(\sigma, t)$. The lower bound follows from our lower bound for $\lns$ in Lemma~\ref{lem:lns_length}. 

\end{proof}

\begin{lemma}
	\label{lem:lnst4}
	Let $x\in \Sigma^n$ and $|\Sigma| =  r $ and $\eps>0$ be any constant. Let $t $ be some constant. Then any $R$ pass deterministic algorithm outputs a $1+\eps$ approximation of $\lnst(x,t)$ takes $\Omega(\min(\sqrt{n}, r))$ space. 
	
\end{lemma}

\begin{proof}
	Let $\sigma\in \Sigma^{n/t}$. If we repeat every symbol in $\sigma$ $t$ times, we get a string $\sigma'\in \Sigma$. Then, $\lis(x) = \frac{1}{t}\lnst(x,t)$. When $t$ is a constant, the lower bound follows from lower bounds for $\lis$ in Lemma~\ref{lem:lis_length}.
\end{proof}

\begin{lemma}
	\label{lem:lnst2}
	
	Let $x\in \Sigma^n$. Assume $|\Sigma| =  r $ and $\eps>0$ such that $r^2/\eps = O(n)$. Let $t = \Theta(r/\eps)$. Then any $R$-pass deterministic algorithm that outputs a $1+\eps$ approximation of $\lnst(x,t)$ takes $\Omega(r\log1/\eps$ space. 
\end{lemma}

\begin{proof}
	In the proof of Lemma~\ref{lem:lns_length_eps}, we considered strings $x$ where each symbol is appeared at most $4rl$ times  where $l = 1/\eps$. We $t = 4rl = \Theta(r/\eps)$. Thus, $\lnst(x, t) = \lns(x)$. The lower bound follows from Lemma~\ref{lem:lns_length_eps}.

\end{proof}

\begin{lemma}
	\label{lem:lnst3}
	Assume $x\in \Sigma^n$, $\eps>0$, and $\frac{|\Sigma|}{\eps}   = O(n) $ . Let $t = \Theta(1/\eps)$, any $R$-pass deterministic algorithm that outputs an $1+\eps$ approximation of $\lnst(x,t)$ requires $\Omega(\frac{\sqrt{|\Sigma|}}{\eps})$ space.
\end{lemma}

\begin{proof}
	The lower bound is achieved using the same construction in the proof of Lemma~\ref{lem:lcs_approx_r_eps} with some modifications. In Section~\ref{sec:lcs_exact_det}, for any $n$, we build a fooling set $S_{a,b}\subset \{a,b\}^{n/2}$ and $\overline{S_{a,b}}\subseteq \{a,b\}^{n/2+5}$ (we used the notation $S$ instead of $S_{a,b}$ in Section~\ref{sec:lcs_exact_det}) such that $\overline{S_{a,b}} = \{f(x)| x\in S_{a,b}\}$ where the function $f$ simply delete the first 10 $b$'s in $x$.  We prove Lemma~\ref{lem:S} and Claim~\ref{claim:|S|}. We modify the construction of $S_{a,b}$ and $\overline{S_{a,b}}$ with three symbols $a,b,c$. The modification is to replace every $a$ symbols in the strings in $\overline{S_{a, b}}$ with $c$ symbols. This gives us a new set $\overline{S_{b, c}}\subset \{b,c\}^{n/2-5}$. Thus, the function $f$ now becomes, on input $x\in S_{a,b}$, first remove 10 $b$'s and then replace all $a$ symbols with $c$. 
	
	Let $y = a^{n/3}b^{n/3}c^{n/3}$. We can show that for every $x\in S_{a,b}$, 
	
	\[\lcs(x\circ f(x), y) = \frac{n}{2}+5.\]  
	
	Also, for any two distinct $x^1, x^2\in S_{a,b}$, 
	
	\[\max\{\lcs(x^1\circ f(x^2), y),\lcs(x^2\circ f(x^1), y)\} > \frac{n}{2}+5.\]
	
	The proof is the same as the proof of Lemma~\ref{lem:S}.
	
	
	We now modify the construction in the proof of Lemma~\ref{lem:lcs_approx_r_eps}. Since $t = \Theta(1/\eps)$, we choose $n_\eps $ such that  $n_\eps/3 =  t$. More specifically, for the matrix $\bB$ (see equation~\ref{eq:bB}), we do the following modification. For any $i, j\in [r]$, $y^{i,2j-1}\circ y^{i,2j} = a_i^{n_\eps/3}b_i^{n_\eps/3}a_i^{n_\eps/3}$. We replace $y^{i,2j-1}\circ y^{i,2j}$ with string $d_{i,j}^{n_\eps/3}e_{i,j}^{n_\eps/3}f_{i,j}^{n_\eps/3}$. Here, $d_{i,j}, e_{i,j},f_{i,j}$ are three symbols in $\Sigma$ such that $d_{i,j}< e_{i,j}< f_{i,j}$. In the matrix $\tB$ (see equation~\ref{eq:tB}), for $B_{i,2j-1}$, we replace every $a_i$ symbols with $d_{i,j}$ and $b_i$ symbols with $e_{i,j}$. For $B_{i,2j}$, we place $a_i$ symbols with $f_{i,j}$ and $b_i$ symbols with $e_{i,j}$.
	
	We also replace the $c_{j}^{n_\eps}$ block in  the $i$-th row of both $\bB$ and $\tB$ with $c_{i,j,1}^{n_\eps/3}c_{i,j,2}^{n_\eps/3}c_{i,j,2}^{n_\eps/3}$. Here, $c_{i,j,1}, c_{i,j,2}.c_{i,j,3}$ are three different symbols in $\Sigma$ such that $c_{i,j,1}< c_{i,j,2}<c_{i,j,3}$.
	
	$y$ is the concatenation of rows of $\bB$. We require that symbols appeared earlier in $y$ are smaller. Since $y$ is the concatenation of all symbols that appeared in $x$ and each symbol in $y$ repeated $t = n_\eps$ times. After the symbol replacement, we have $\lcs(x,y)  = \lnst(x,t)$. Also notice that the alphabet size is now $O(r^2)$ instead of $r$ in the proof of Lemma~\ref{lem:lcs_approx_r_eps}. The $\Omega(\frac{|\Sigma|}{\eps})$ space lower bound then follows from a similar analysis in the proof of Lemma~\ref{lem:lcs_approx_r_eps}.

\end{proof}

We now prove a trivial space upper bound for $\lnst$.

%

\begin{algorithm}
	\SetAlgoLined
	\DontPrintSemicolon
	\KwIn{An online string $x\in \Sigma^n$ where $\Sigma = [r]$.}

	Let $D = \{0, \frac{\eps}{r}t, 2\frac{\eps}{r}t, \dots, t\}$ \Comment*[f]{if $r/\eps \geq t$, $D = [t]\cup \{0\}$.}\\ 
	$S\leftarrow \emptyset$. \\
	\ForEach{$d = (d_1,\dots, d_r)\in\mathcal{D}^r$}{
		\Comment*[f]{try every $d = (d_1,\dots, d_r)\in\mathcal{D}^r$ in parallel}\\
		
		Let $\sigma(d) = 1^{d_1}2^{d_2}\cdots r^{d_r}$. \\
		If $\sigma(d)$ is a subsequence of $x$, add $d$ to $S$.
	}
	in parallel, compute $\lns(x)$ exactly with an additional $\tilde{O}(r)$ bits of space. \\
	\uIf{$\lns(x)\leq t$}{
		\Return{$\lns(x)$.}\\
	}
	\Else{
		\Return{$\max_{d\in S}\big(\sum_{i\in [r]} d_i\big)$}	
	}
	\caption{Approximate the length of the longest non-decreasing subsequence with threshold}
	\label{algo:lnst}
\end{algorithm}

\begin{lemma}
	\label{lem:upper_bound}
	Given an alphabet $\Sigma$ with $\lvert \Sigma \rvert = r$.\ For any $\eps>0$ and $t \geq 1$, there is a one-pass streaming algorithm that computes a $(1+\eps)$ approximation of $\lnst(x, t)$ for any $x\in \Sigma^n$ with $\tilde{O}\Big(\big(\min(t, r/\eps )+1\big)^r\Big)$ space. 
\end{lemma}

\begin{proof}
	We assume $\Sigma = [r]$ and the input is a string $x\in \Sigma^n$. Let $\sigma$ be a longest non-decreasing subsequence of $x$ with threshold $t$ and we can write $\sigma = 1^{n_1}2^{n_2}\cdots r^{n_r}$ where $ n_i$ is the number of times symbol $i$ repeated and $0\leq n_i\leq t$. 
	
	We let $D$ be the set $\{0, \frac{\eps}{r}t, 2\frac{\eps}{r}t, \dots, t\}$. Thus, if $r/\eps \geq t$, $D = [t]\cup \{0\}$. Consider the set $\mathcal{D}$ such that $\mathcal{D} = D^r$. Thus, $\lvert \mathcal{D}\rvert = \big(\min(t+1, r/\eps + 1 )\big)^r$. For convenience, we let $f(d) = \sum_{i = 1}^r d_i$. We initialize the set $S$ to be an empty set. For each $d\in \mathcal{D}$, run in parallel, we check is $\sigma(d) = 1^{d_1}2^{d_2}\cdots r^{d_r}$ is a subsequence of $x$. If $\sigma(d)$ is a subsequence of $x$, add $d$ to $S$. 
	
	Meanwhile, we also compute $\lns(x)$ exactly with an additional $ \tilde{O}(r)$ bits of space. If $\lns<t$, we output $\lns(x)$. Otherwise, we output $\max_{d\in S}\big(\sum_{i\in [r]} d_i\big)$. 
	
	We now show the output is a $(1-\eps)$-approximation of $\lnst(x,t)$. Let $d'_i$ be the largest element in $\mathcal{D}$ that is no larger than $n_i$ for $i\in[r]$ and $d' = (d'_1, d'_2, \dots, d'_r)$. 	
	
	If $t \leq |\sigma|$, we have $0\leq n_i-d'_i\leq \eps/r t$ and 
	
	\[\lvert \sigma \rvert -f(d') \leq \sum_{i = 1}^r(n_i-d'_i)\leq \eps t \leq \eps \lvert \sigma \rvert \] 
	
	since we assume $t \leq |\sigma|$. Thus, $f(d')\geq (1-\eps)\lvert \sigma \rvert.$ Note that $1^{d'_1}2^{d'_2}\cdots r^{d'_r}$ is a subsequence of $\sigma$ and thus also a subsequence of $x$. Thus, we add $d'$ to the set $S$. That means, the final output will be at least $f(d')$. Denote the final output by $l$, we have
	
	\[l \geq f(d') \geq (1-\eps)\lvert \sigma \rvert.\]
	
	On the otherhand, the the output is $f(d)$ for some $d = (d_1, \dots, d_r)$, we find a subsequence $1^{d_1}2^{d_2}\cdots r^{d_r}$ of $x$ and thus $f(d)\leq \lvert \sigma \rvert$. We know  
	
	\[l = \max_{d\in S} f(d) \leq \lvert \sigma \rvert.\]
	
	If $t \geq \lns(x)$, no symbol in $\sigma$ is repeated more than $t-1$ times. Thus, $\lnst(x, t) = \lns(x)$.  Thus, we output $\lns(x)$. Notice that if $ \lns(x)> t$, either some symbol in the longest non-decreasing subsequence is repeated more than $t$ times, or $\lnst(x,t) = \lns(x)$. In either case, we have $t \leq |\sigma|$ and $\max_{d\in S}\big(\sum_{i\in [r]} d_i\big)$ is a $1-\eps$ approximation of $\lns(x)$.

\end{proof}


\section{Algorithms for Edit Distance and LCS}
\label{sec:algorithm}

\subsection{Algorithm for edit distance}

In this section, we presents our space efficient algorithms for approximating edit distance in the asymmetric streaming setting. 

We can compute edit distance exactly with dynamic programming using the following recurrence equation. We initialize $A(0,0) = 0$ and $A(i,0) = A(0,i) = i$ for $i\in [n]$. Then for $0\leq i,j \leq n$, 

\begin{equation}
	\label{eq:ed_dp}
	A(i,j) = \begin{cases}
		A(i-1,j-1) , & \text{if $x_i = y_j$.} \\
		\min\begin{cases}
			A(i-1,j-1)+1, \\
			A(i,j-1)+ 1, \\
			A(i-1,j) + 1, 
		\end{cases}, & \text{if $x_i \neq y_j$.}
	\end{cases}
\end{equation}

Where $ A(i,j)$ is the edit distance between $x[1:i]$ and $y[1:j]$. The three options in the case $x_i\neq y_j$ each corresponds to one of the edit operations (substitution, insertion, and deletion). Thus, we can run a dynamic programming to compute matrix $A$. When the edit distance is bounded by $k$, we only need to compute the diagonal stripe of matrix $A$ with width $O(k)$. Thus, we have the following. 


\begin{lemma}
	\label{lem:O(k)_space_ed}
	Assume we are given streaming access to the online string $x \in \Sigma^*$ and random access to the offline string $y\in \Sigma^n$. We can check whether $\ed(x,y)\leq k $ or not, with $O (k\log n)$ bits of space in $O(nk)$ time. If $\ed(x,y)\leq k$, we can compute $\ed(x,y)$ exactly with $O(k\log n)$ bits of space in $O(nk)$ time. 
\end{lemma}

\begin{claim}
	\label{claim:ed_operation}
	For two strings $x, y\in \Sigma^n$, assume $\ed(x,y) < k$, we can output the edit operations that transforms $x $ to $y$ with $O(k\log n)$ bits of space in $O(nk^2)$ time. 
\end{claim} 

\begin{proof}
	The idea is to run the dynamic programming $O(k)$ times. We initialize $n_1 = n_2 = n$. If $x_{n_1} = y_{n_2}$, we find the largest integer such that $x[n_1-i:n_1] = y[n_2-i:n_2]$ and set $n_1 \leftarrow n_1-i$ and $n_2 \leftarrow n_2-i$ and continue. If $x_{n_1} \neq y_{n_2}$, we compute the the elements $A(n_1,n_2-1)$,  $A(n_1-1,n_2-1)$, and $A(n_1-1,n_2)$. By \ref{eq:ed_dp}, we know there is some $(n'_1, n'_2)\in \{(n_1,n_2-1), (n_1-1,n_2-1), (n_1-1,n_2)\}$ such that  $A(n'_1, n'_2) = A(n_1,n_2)-1$. We set $n_1 \leftarrow n'_1$ and $n_2 \leftarrow n'_2$, output the corresponding edit operation and continue. We need to run the dynamic programming $O(k)$ times. The running time is $O(nk^2)$.
\end{proof}

\begin{lemma}
	\label{lem:ed_divide}
	For two strings $x, y \in \Sigma^n$, assume $d = \ed(x,y)$. For any integer $1\leq t \leq d$, there is a way to divide $x$ and $y$ each into $t$ parts so that $x = x^1\circ x^2 \circ \cdots \circ x^t$ and $y = y^1\circ y^2 \circ \cdots \circ y^t$ (we allow $x^i $or $y^i$ to be empty for some $i$), such that $d = \sum_{i = 1}^t \ed(x^i, y^i)$ and $\ed(x^i, y^i)\leq \lceil\frac{d}{t}\rceil $ for all $i \in [t]$.
\end{lemma}

\begin{proof}
	[Proof of Lemma~\ref{lem:ed_divide}]
	Since $\ed(x,y) = d$, we can find $d$ edit operations on $x$ that transforms $x$ into $y$. We can write these edit operations in the same order as where they occured in $x$. Then, we first find the largest $i_1$ and $j_1$, such that the first $\lceil \frac{d}{t}\rceil$ edit operations transforms $x[1:i_1]$, to $y[1:j_1]$. Notice that $i_1$ (or $j_1$) is 0 if the first $\lceil \frac{d}{t}\rceil$ edit operations insert $\lceil \frac{d}{t}\rceil$ before $x^1$ (or delete first $\lceil \frac{d}{t}\rceil$ symbols in $x$).  We can set $x^1 = x[1:i_1]$ and $y^1 = y[1:j_1]$ and continue doing this until we have seen all $d$ edit operations. This will divide $x$ and $y$ each in to at most $t$ parts. 
\end{proof}

In our construction, we utilize the following result from \cite{cheng2020space}.

\begin{lemma}
	\label{lem:smallspace_ed}
	[Theorem 1.1 of \cite{cheng2020space}]
	For two strings $x, y \in \Sigma^n$, there is a deterministic algorithm that compute a $1+o(1)$-approximation of $\ed(x,y)$ with $O(\frac{\log^4 n}{\log\log n}) $ bits of space in polynomial time. 
\end{lemma}

We now present our algorithm that gives a constant approximation of edit distance with improved space complexity compared to the algorithm in \cite{farhadi2020streaming}, \cite{cheng2020space}. The main ingredient of our algorithm is a recursive procedure called $\mathsf{FindLongestSubstring}$. The pseudocode is given in algorithm~\ref{algo:longestsubstringR}. It takes four inputs: an online string $x$, an offline string $y$, an upper bound of edit distance $u$, and an upper bound of space available $s$. The output of $\mathsf{FindLongestSubstring}$ is a three tuple: a pair of indices $(p,q)$, two integers $l$ and $d$. Througout the analysis, we assume $\eps$ is a small constant up to our choice.

\begin{algorithm}
	\SetAlgoLined
	\DontPrintSemicolon
	\KwIn{An online $x\in \Sigma^*$ with streaming access, an local string $y\in \Sigma^n$, an upper bound of edit distance $u$, and an upper bound of space available $s$.}
	\KwOut{A pair of indices $(p,q)\in [n]\times [n]$, an integer $l\in n$, and an approximation of edit distance $d$}
	\uIf{$u\leq s$}{
		\For{ $l' = u $  \textbf{to} $|x|$}{\label{algo:longestsubstring_for}
			try all $1\leq p' <q' \leq n$ such that $l'-u \leq |q'-p'+1| \leq l'+u$ and check whether $\ed(y[p':q'], x[1:l'])\leq u$. \;	
			\uIf{there exists a pair $p',q'$ such that $\ed(y[p':q'], x[1: l'])\leq u$}{
				set $p, q$ to be the pair of indices that minimizes $\ed(y[p':q'], x[1: l'])$.\;
				$d \leftarrow \ed(y[p:q], x[1: l'])$, and $l \leftarrow l'$. \;
				compute and record the edit operations that transform $y[p:q]$ to $x[1: l']$. \;
				\textbf{continue}
			}
			\Else{
				\textbf{break}
			}
			
		}
	}
	\Else{
		initialize $i = 1$, $a_i = 1$, $l = 0$.  \label{algo:findlongestsubstring_initial}\;
		\While{$a_i\leq |x|$ and $i\leq s$ \label{algo:findlongestsubstring_while}}{
		
			$(p_i, q_i), l_i, d_i \leftarrow \mathsf{FindLongestSubstring}(x[a_i:n], y, \lceil u/s\rceil , s )$. \;
			$ i \leftarrow i+1$.  \;
			$a_i \leftarrow a_{i-1} + l_i $. \;
			$l \leftarrow l+l_i$.\;
		}
		$T \leftarrow i-1$. \;
		for all $1\leq p' < q' \leq n $, use the algorithm guaranteed by Lemma~\ref{lem:smallspace_ed} to compute $\tilde{d}(p',q')$, a $1+\eps$ approximation of $ \ed(y[p':q'], y[p_1:q_1]\circ y[p_2:  q_2]\circ \cdots \circ y[p_T: q_T])$. \;
		$p,q \leftarrow \textbf{argmin}_{p',q'}\tilde{d}(p',q')$.  \Comment*[f]{$p,q$ minimizes $\tilde{d}$}\;
		$d \leftarrow \tilde{d}(p,q)+ \sum_{i = 1}^{T}d_i$. \;

	}
	\Return{$(p,q)$, $l$, $d$}
	\caption{FindLongestSubstring}
	\label{algo:longestsubstringR}
\end{algorithm}

We have the following Lemma.

\begin{lemma}
	\label{lem:findlongestsubstring}
	Let $(p,q),l, d$ be the output of  $\mathsf{FindLongestSubstring}$ with input $x, y, u, s$. Then assume $l^0$ is the largest integer such that there is a substring of $y$, say $y[p^0: q^0]$, with $\ed(x[1:l^0], y[p^0: q^0])\leq u$. Then, we have $ l\geq l^0$. Also, assume $y[p^*:q^*]$ is the substring of $y$ that is closest to $x[1:l]$ in edit distance. We have 
	\begin{equation}
		\label{eq:findlongestsubstring}
		\ed(x[1:l], y[p^*:q^*]) \leq \ed(x[1:l], y[p:q])\leq d \leq c(u,s)\ed(x[1:l], y[p^*:q^*]).
	\end{equation}
	Here $c(u,s) = 2^{O(\log_s u)}$ if $u > s$ and $c = 1$ if $ u\leq s$
\end{lemma}

\begin{proof}[Proof of Lemma~\ref{lem:findlongestsubstring}]
	
	In the following, we let $(p,q),l, d$ be the output of  $\mathsf{FindLongestSubstring}$ with input $x, y, u, s$. We let $l^0$ be the largest integer such that there is a substring of $y$, say $y[p^0:q^0]$, with $\ed(x[1:l^0], y[p^0: q^0])\leq u$ and $y[p^*:q^*]$ is the substring of $y $ that minimizes the edit distance to $x[1:l]$.
	 
	$\mathsf{FindLongestSubstring}$ first compares $u$ and $s$. If $u\leq s$, we first set $l' = u$, then we know for any $p,q\in [n]$ such that $q-p+1  = u$ , $\ed(x[1:l'], y[p,q])\leq u$ since we can transform $x[1:l']$ to $y[p:q]$ with $u$ substitutions. Thus, we are guaranteed to find such a pair $(p,q)$ and we can record the edit operation that transform $y[p:q]$ to $x[1:l]$. We set $l = l'$, $d = \ed(x[1:l'], y[p:q])$ and continue with $l'\leftarrow l'+1$. When $l'>u$ ($l = l'-1$).  we may not be able to store $x[1:l]$ in the memory for random access since our algorithm uses at most $O(s\log n)$ bits of space. However, we have remembered a pair of indices $(p,q)$ and at most $u$ edit operations that can transform $y[p:q]$ to $x[1:l]$. This allows us to query each bit of $x[1:l]$ from left to right once with $O(u+ l)$ time. Thus, for each substring $y[p:q]$ of y, we can compute its edit distance from $x[1:l']$. Once we find such a substring with $\ed(x[1:l'],y[p:q])\leq u$, by Claim~\ref{claim:ed_operation}, we can then compute the edit operations that transfom $y[p,q]$ to $x[1:l']$ with $O(u\log n)$ space. Thus, if $u\leq s$, we can find the largest integer $l$ such that there is a substring of $y$, denoted by $y[p:q]$, with $\ed(x[1:l],y[p:q])=u$ with $O(s\log n)$ bits of space. If there is no substring $y[p:q]$ with $\ed(x[1:l'],y[p:q])\leq u$, we terminate the procedure and return current $(p,q)$, $l$, and $d$.
	
	If $l^0> l$, then $x[1:l+1]$ is a substring of $x[1:l^0]$, we can find a substring of $y[p^0:q^0]$ such that its edit distance to $x[1:l+1]$ is at most $u$. Thus, $\mathsf{FindLongestSubstring}$ will not terminate at $l$. We must have $l = l^0$.
	
	Also notice that when $u\leq s$, we always do exact computation. $y[p:q]$ is the substring in $y$ that minimizes the edit distance to $x[1:l]$. Thus, Lemma~\ref{lem:findlongestsubstring} is correct when $u\leq s$.

	For the case $u> s$, $\mathsf{FindLongestSubstring}$ needs to call itself recursively. Notice that each time the algorithm calls itself and enters the next level, the upper bound of edit distance $u$ is reduced by a factor $s$. The recursion ends when $u\leq s$. Thus, we denote the depth of recursion by $d$, where $d = O(\log_s u)$.  We assume the algoirthm starts from level 1 and at level $i$ for $i\leq d$, the upper bound of edit distance becomes $\frac{u}{s^{i-1}}$.

	We prove the Lemma by induction on the depth of recursion. The base case of the induction is when $u\leq s$, for which we have shown the correctness of Lemma~\ref{lem:findlongestsubstring}. We now assume Lemma~\ref{lem:findlongestsubstring} holds when the input is $x, y , \lceil u/s\rceil , s$ for any strings $x, y$. 
	
	We first show $l\geq l^0$. Notice that the while loop at line~\ref{algo:findlongestsubstring_while} terminates when either $a_i > |x|$ or $i>s$. If $a_i>|x|$, we know $l$ is set to be $a_i-1  = |x|$. Since $l^0\leq |x|$ by definition. We must have $l\geq l^0$. 
	
	If the while loop terminates when $i>s$. By the definition of $l^0$, we know $x[1:l^0]$  and $y[p^0: q^0]$ can be divided into $s$ blocks, say, 
	\begin{align*}
		&x[1:l^0] = x[1:l^0_1]\circ x[l^0_1+1:l^0_1 + l^0_2]\circ \cdots \circ x[\sum_{i = 1}^{s-1}l^0_{i}+1:\sum_{i = 1}^{s}l^0_{i}] \\
		&y[p^0:q^0] = y[p^0_1:q^0_1]\circ y[p^0_2:q^0_2]\circ \cdots \circ y[p^0_{s}:q^0_s]
	\end{align*}
	where $l^0 =  \sum_{i = 1}^{s}l^0_{i}$, such that 
	\begin{equation*}
		\ed(x[\sum_{j = 1}^{i-1}l^0_{j}+1:\sum_{j = 1}^{i}l^0_{j}], y[p^0_{i}:q^0_i])\leq \lceil u/s \rceil , \; \forall i\in [s].
	\end{equation*}
	For convenience, we denote $b^0_i = \sum_{j = 1}^{i}l^0_{j}$ and $b_i = \sum_{j = 1}^{i}l_{j}$. By the defintion, we know $b^0_s = l^0$ and $b_s = l$ and all $l^0_i, l_i$ are non-negative. We have $b_i = a_{i+1}-1$.
	
	We show that $b_i\geq b^0_i$ for all $i\in [s]$ by induction on $i$. For the base case $i = 1$, let $\bar{l}^0_1$ be the largest integer such that there is a substring of $y$ within edit distance $\lceil u/s \rceil $ to $x[1: \bar{l}^0_1]$. We know $\bar{l}^0_1 \geq l^0_1$. Since we assume Lemma~\ref{lem:findlongestsubstring} holds for inputs $x, y, \lceil u/s\rceil  , s$, we know $l_1\geq \bar{l}^0_1$. Thus, $l_1\geq l^0_1$ and $b_1\geq b^0_1$. Now assume $b_i\geq b^0_i$ holds for some $i \in [s-1]$, we show that $b_{i+1}\geq b^0_{i+1}$. If $b_i\geq b^0_{i+1}$, $b_{i+1}\geq b^0_{i+1}$ holds. We can assume $ b^0_{i+1} > b_i \geq b^0_{i}$. We show that $l_{i+1}\geq b^0_{i+1}-b_i$. To see this, let $\bar{l}^0_i$ be the largest integer such that there is a substring of $y$ within edit distance $\lceil u/s \rceil$ to $x[b_i+1: b_i+ \bar{l}^0_{i+1}]$. We know $l_{i+1}\geq \bar{l}^0_{i+1}$ since we assume Lemma~\ref{lem:findlongestsubstring} holds for inputs $x[b_i+1:|x|], y, \lceil u/s\rceil  , s$. Notice that $\ed(x[b^0_{i}+1:b^0_{i+1}], y[p^0_i, q^0_i])\leq \lceil u/s\rceil$, we know $\bar{l}^0_i$ is at least $b^0_{i+1}-b_i$ since $x[b_{i}+1:b^0_{i+1}] $ is a substring of $x[b^0_{i}+1:b^0_{i+1}]$. We know $l_{i+1} \geq \bar{l}^0_{i+1}\geq b^0_{i+1}-b_i$. Thus, $b_{i+1} = b_{i}+ l_{i+1}\geq b^0_{i+1}$. Thus, we must have $l\geq l_0$.

	

	
	We now prove inequality~\ref{eq:findlongestsubstring}. After the while loop, the algorithm then finds a substring of $y$, $y[p:q]$, that minimizes $\tilde{d}(p', q')$ where $\tilde{d}(p',q')$ is a $1+\eps$ approximation of $ \ed(y[p':q'], y[p_1:q_1]\circ y[p_2:  q_2]\circ \cdots \circ y[p_T: q_T])$. For convenience, we denote 
	$$\tilde{y} = y[p_1:q_1]\circ y[p_2:  q_2]\circ \cdots \circ y[p_T: q_T]. $$ 
	Thus,
	\begin{equation}
		\label{ed:eq1}
		\ed(y[p:q], \tilde{y}) \leq \tilde{d}(p, q) \leq  (1+\eps)\ed(y[p^*:q^*], \tilde{y}).
	\end{equation}

	Let $y[\bar{p}^*_j: \bar{q}^*_j]$ be the substring of $y$ that is closest to $x[a_j: a_{j+1}-1]$ in edit distance. By the inductive hypothesis, we assume the output of $\mathsf{FindLongestSubstring}(x, y , \lceil u/s \rceil, s)$ satisfies Lemma~\ref{lem:findlongestsubstring}. We know 
	\begin{align}
		\ed(x[a_j: a_{j+1}-1], y[\bar{p}^*_j: \bar{q}^*_j]) &\leq \ed(x[a_j: a_{j+1}-1], y[p_j: q_j]) \notag\\
		&\leq d_j \label{ed:eq2}\\
		&\leq c(u/s,s)\ed(x[a_j: a_{j+1}-1], y[\bar{p}^*_j: \bar{q}^*_j]). \notag
	\end{align}

	
	By the optimality of $y[p^*:q^*]$ and triangle inequality, we have 

	\begin{align*}
		\ed(x[1:l], y[p^*: q^*]) &\leq \ed(x[1:l], y[p: q]) & \\
		&\leq \ed(x[1:l], \tilde{y}) + \ed( \tilde{y},y[p:q] ) &\text{By triangle inequality}\\
		& \leq \sum_{i = 1}^{T}d_i + \tilde{d}(p,q) & \\
		& = d & \\
	\end{align*}

	Also notice that we can write $y[p^*:q^*] = y[p^*_1: q^*_1]\circ y[p^*_2:q^*_2] \circ \cdots \circ y[p^*_T:q^*_T]$ such that 
	
	\begin{equation}
		\label{ed:eq3}
		 \ed(x[1:l], y[p^*:q^*]) = \sum_{i = 1}^T\ed(x[a_i:a_{i+1}+1], y[p^*_i, q^*_i]).
	\end{equation}

	We have
	\begin{align*}
			d & = \sum_{i= 1}^{T}d_j + \tilde{d}(p,q) &\\
			& \leq \sum_{i = 1}^{T} d_i + (1+\eps)\ed(\tilde{y}, y[p^*: q^*]) & \text{By \ref{ed:eq1}}\\
			& \leq \sum_{i = 1}^{T} d_i + (1+\eps)\Big(\ed(x[1:l], y[p^*:q^*]) + \ed(x[1:l], \tilde{y})\Big) & \text{By triangle inequality}\\
			&\leq \sum_{i = 1}^{T} d_i + (1+\eps)\sum_{i = 1}^T\ed(x[a_i: a_{i+1}-1], y[p^*_i: q^*_i]) \\
			& \qquad + (1+\eps)\sum_{i = 1}^T\ed(x[a_i:a_{i+1}-1], y[p_i: q_i]) & \\
			& \leq (1+\eps) \sum_{i = 1}^{T} (2c(u/s, s) + 1)\ed(x[a_i: a_{i+1}-1], y[p^*_i: q^*_i]) &\text{By \ref{ed:eq2}} \\
			& \leq (1+\eps)(2c(u/s, s) + 1)\ed(x[1:l], y[p^* : q^*]) & \text{By \ref{ed:eq3}}\\
	\end{align*}

	We set $c(u, s) = (1+\eps)(2c(u/s, s) + 1)$. Since we assume $c(u/s, s) = 2^{O(\log_s (u/s))}$, we know $c(u,s) = 2^{O(\log_s u)}$. This proves inequality~\ref{eq:findlongestsubstring}.

	
\end{proof}

\begin{lemma}
	\label{lem:findlongestsubstring_space}
	Given any $x, y\in \Sigma^n$, let $u, s\leq n$, $\mathsf{FindLongestSubstring}(x,y, u, s)$ runs in polynomial time. It queries $x$ from left to right in one pass and uses $O(s \log_s u\;\polylog(n))$ bits of extra space. 
\end{lemma}

\begin{proof}[Proof of Lemma~\ref{lem:findlongestsubstring_space}]
	If $u\leq s$, we need to store at most $u$ edit operations that transforms, $y[p: q]$ to $x[1:l]$ for current $(p,q)$ and $l$. This takes $O(u\log n) = O(s\log n)$ bits of space. Notice when $l \geq  u$, we do not need to remember $x[1:l]$. Instead, we query $x[1:l]$ by looking at $y[p:q]$ and the edit operations. 

	If $u \geq s $, the algorithm is recursive.  Let us consider the recursion tree. We assume the algoirthm starts from level 1 (root of the tree) and the depth of the recursion tree is $d$. At level $i$ for $i\leq d$, the upper bound of edit distance (third input to algorithm $\mathsf{FindLongestSubstring}$) is $u_i = \frac{u}{s^{i-1}}$. The recursion stops when $u_i\leq s $. Thus, the depth of recursion $d$ is $O(\log_s u)$ by our assumption on $u$ and $s$. The order of computation on the recursion tree is the same as depth-first search and we only need to query $x$ at the bottom level. There are at most $s^d = O(u)$ leaf nodes (nodes at the bottom level). For the $i$-th leaf nodes of the recursion tree, we are computing $\mathsf{FindLongestSubstring}(x[a_i:n], y, u_d, s)$ with $u_d\leq s$ where $a_i$ is the last position visited by the previous leaf node. Thus, we only need to access $x$ from left to right in one pass. 
	
	For each inner node, we need to remember $s$ pairs of indices $(p_i, q_i)$ and $s$ integers $d_i$ for $i\in [s]$, which takes $O(s\log n)$ space. For computing an $(1+\eps)$-approximation of $\tilde{d}(p',q')$, we can use the space-efficient approximation algorithm from~\cite{cheng2020space} that uses only $\polylog(n)$ space. Thus, each inner node takes $O(s \;\polylog(n))$ bits of space. For the leaf nodes, we have $u\leq s$. Thus, we can compute it with $O(s\log n)$ bits of extra memory. Since the order of computation is the same as depth-first search, we only need to maintain one node in each recursive level and we can reuse the space for those nodes we have already explored. Since the depth of recursion is $O(\log_s u)$, the total space required is $O(s\log_s u \;\polylog(n))$. 
	
	For the time complexity, notice that the space efficient algorithm for $1+\eps$ approximating $\ed$ takes polynomial time. The depth of recursion is $O(\log_s u)$ and at each recursive level, the number of nodes is polynomial in $n$ , we need to try $O(n^2)$ different $p',q'$ at each node except the leaf nodes. Thus, the running time is still polynomial.

\end{proof}

We now present our algorithm for approximating edit distance in the asymmetric streaming model. The pseudocode of our algorithm is given in algorithm~\ref{algo:ed_1}. It takes three input, an online string $x$, an offline string $y$ and a parametert $\delta\in (0,1/2]$.

\begin{algorithm}
	\SetAlgoLined
	\DontPrintSemicolon
	\KwIn{Two strings: $x\in \Sigma^*$ and $y\in \Sigma^n$, a constant $\delta \in(0,1/2]$}
	\KwOut{An estimation of $\ed(x,y)$ }
	initialize $a\leftarrow 1$, $i\leftarrow 1$. \;
	
	set $k$ to be a large constant such that $s = k^\delta$ is an integer. \;
	\While{$a\leq n$ \label{algo:ed_1_loop} }{
		$(p_i, q_i), l_i, d_i \leftarrow \mathsf{FindLongestSubstring}(x[a:n], y, k/s, s )$. \label{algo:ed2_findlongestsubstring}\;
		$a \leftarrow a + l_i $. \;
		$ i \leftarrow i+1$.  \;
		\If{$i\geq k^\delta$}{
			$k\leftarrow 2^{1/\delta} k.$ \;
			$s\leftarrow k^\delta$. \;
		}
		
	}
	$T \leftarrow i-1$. \;
	compute $\tilde{d}$, an $(1+\eps)$-approximation of $\ed(y, y[p_1,q_1]\circ y[p_2, q_2]\circ \cdots \circ y[p_T, q_T])$. \Comment*[f]{Using the algorithm guaranteed by Lemma~\ref{lem:smallspace_ed}}\;
	
	\Return{$\bar{d} = \tilde{d} + \sum_{i = 1}^{T}d_i$}\;
	\caption{Approximate edit distance in asymmetric streaming model}
	\label{algo:ed_1}	
\end{algorithm}

\begin{lemma}
	\label{lem:ed_time_space}
	Assume $d = \ed(x,y)$, Algorithm~\ref{algo:ed_1} can be run with $O(\frac{d^\delta}{\delta} \polylog(n))$ bits of space in polynomial time. 
\end{lemma}

\begin{proof}
	Notice that $k$ is initially set to be a constant $k_0$ such that $s_0  = k^\delta_0$ is an integer. $k$ is multiplied by $2^{1/\delta}$ whenever $i\geq k$. We assume that in total, $k$ is updated $u$ times and after the $j$-th update, $k = k_j$, where $k_j = 2^{1/\delta}k_{j-1}$. We let $s_j = k^\delta_j$. Thus, $k^\delta_{j+1} = 2k^\delta_j $ and $s_{j+1} = 2s_j$. We denote the $a$ before $i$-th while loop by $a_i$ so that $a_1 = 1$ and $a_i = 1+\sum_{j = 1}^{i-1}l_j$ for $1<i\leq T$.
	
	We first show the following claim.
	
	\begin{claim}
		\label{claim:ed_space}
		$k^\delta_u \leq 8d^\delta.$
	\end{claim}

	\begin{proof}
		Assume the contrary, we have $k^\delta_u > 8d^\delta$, and thus $k^\delta_{u-1} >4 d^\delta$. 
		
		Let $i_0 = k^\delta_{u-2} $. That is, after $i$ is updated to $i_0$, $k$ is updated from $k_{u-2}$ to $k_{u-1}$. For convenience, we denote $\bar{x} = x[a_{i_0}:n]$. Since $\ed(x,y)\leq d$, there is a substring of $y$, say $\bar{y}$, such that $\ed(\bar{x}, \bar{y})\leq d$. Let $\xi = \frac{k_{u-1}}{s_{u-1}}$, by Lemma~\ref{lem:ed_divide},  we can partition $\bar{x}$ and $\bar{y}$ each into $\lceil d/\xi\rceil  $ parts such that 
		\begin{align*}
			\bar{x} = \bar{x}^1\circ \bar{x}^2 \circ \cdots \circ \bar{x}^{\lceil d/\xi\rceil },\\
			\bar{y} = \bar{y}^1\circ \bar{y}^2 \circ \cdots \circ \bar{y}^{\lceil d/\xi\rceil },
		\end{align*}
		and $\ed(\bar{x}^j, \bar{y}^j) \leq \xi$ for $j\in [\lceil d/\xi \rceil ]$. We denote $\bar{x}^j = x[\beta_j: \gamma_j]$ for $j\in [\lceil d/\xi \rceil ]$ so that $\bar{x}^j$ starts with $\beta_j$-th symbol and ends with $\gamma_j$-th symbol of $x$. We have $\beta_1 = a_{i_0}$ and $\gamma_i+1 = \beta_{i+1}$.

		We first show that for $i_0\leq i \leq T $, $a_i\geq \beta_{i-i_0+1}$. Assume  there is some $i$ such that $a_i\geq \beta_{i-i_0+1}$ and $a_{i+1} < \beta_{i-i_0+2}$. By Lemma~\ref{lem:findlongestsubstring}, $l_i$ is at least the largest integer $l$ such that there is a substring of $y$ within edit distance $\xi$ to $x[a_i:a_i+l-1]$. Since $x[a_i:a_{i+1}-1] = x[a_i:a_i+l_i-1]$ is a substring of $x[\beta_{i-i_0+1}: \gamma_{i-i_0+1}] = \bar{x}^{i-i_0+1}$ and $\ed(\bar{x}^{i-i_0+1}, \bar{y}^{i-i_0+1})\leq \xi$. There is a substring of $y$ within edit distance $\xi$ to $x[a_i:\gamma_{i-i_0+1}]$.  Thus, we must have $l_i \geq \gamma_{i-i_0+1} - a_i + 1$. Thus $a_{i+1} = a_i + l_i \geq \gamma_{i-i_0+1}  + 1 = \beta_{i-i_0+2}$. This is contradictory to our assumption that $a_{i+1} < \beta_{i-i_0+2}$. 
		
		We now show that $T\leq i_0+ \lceil d/\xi \rceil -1$. If $T >  i_0+\lceil d/\xi\rceil $, we have $a_{i_0+\lceil d/\xi \rceil -1}\geq \beta_{\lceil d/\xi \rceil } $. By Lemma~\ref{lem:findlongestsubstring}, we must have $a_{i_0+\lceil d/\xi \rceil } = n+1$. Since $a_{i_0+\lceil d/\xi \rceil } > n$, we will terminate the while loop and set $T = i_0+\lceil d/\xi \rceil -1$. Thus, we have $T\leq i_0+ \lceil d/\xi \rceil -1$.
		

		Meanwhile, by the assumption that $k^\delta_{u-1} > 4d^\delta$, we have $k^\delta_{u-1} - k^\delta_{u-2} = k^\delta_{u-2}  > 2d^\delta > 2d/\xi$. If $k$ is updated $u$ times, $T$ is at least $k^\delta_{u-1} -1$. Thus, $T\geq   k^\delta_{u-1}-1  > i_0 +  2 d/\xi -1$. This is contradictory to $T\leq i_0+\lceil d/\xi \rceil -1$.  
	\end{proof}

	By the above claim, we have $T = O(d^\delta)$. Thus, we can remember all $(p_i, q_i), l_i, d_i$ for $i\in [T]$ with $O(d^\delta\log n)$ bits of space. For convenience, let $\tilde{y} =  y[p_1:q_1]\circ y[p_2: q_2]\circ \cdots \circ y[p_T: q_T]$. We can use the space efficient algorithm from \cite{cheng2020space} (Lemma~\ref{lem:smallspace_ed}) to compute a $(1+\eps)$-approximation of $\ed(y, \tilde{y})$ with $\polylog(n)$ bits of space in polynomial time. By Lemma~\ref{lem:findlongestsubstring_space}, we can run $\mathsf{FindLongestSubstring}$ with $O(\frac{d^\delta}{\delta} \polylog(n))$ space since $s = O(d^\delta)$. The total amount of space used is $O(\frac{d^\delta}{\delta} \polylog(n))$.
	
	We run $\mathsf{FindLongestSubstring}$ $ O(d^\delta)$ times and compute a $(1+\eps)$-approximation of $\ed(y, \tilde{y})$ with polynomial time. The total running time is still a polynomial. 
\end{proof}

\begin{lemma}
	\label{lem:3-approx_ed}
	Assume $\ed(x,y) = d$, there is a one pass deterministic algorithm that outputs a $(3+\eps)$-approximation of $\ed(x,y)$ in asymmetric streaming model, with $O(\sqrt{d}\;\polylog(n))$ bits of space in polynomial time. 
\end{lemma}

\begin{proof}[Proof of Lemma~\ref{lem:3-approx_ed}]

	We run algorithm~\ref{algo:ed_1} with parameter $\delta = 1/2$. The time and space complexity follows from Lemma~\ref{lem:ed_time_space}. 
	
	Notice that algorithm~\ref{algo:ed_1} executes $\mathsf{FindLongestSubstring}$ $T$ times and records $T$ outputs $(p_i, q_i), l_i, d_i$ for $i\in [T]$. We also denote $a_i = 1 + \sum_{j = 1}^{i-1}l_j$ with $a_1 = 1$. Thus, we can partition $x$ into $T$ parts such that 
	\begin{equation*}
		x = x^1\circ x^2 \circ \cdots \circ x^T
	\end{equation*}
	where $x^i = x[a_i:a_{i+1}-1]$. Since $s = k^{1/2} = k/s$, by Lemma~\ref{lem:findlongestsubstring}, we know $\ed(x^i, y[p_i:q_i]) = d_i$. 

	We now show that the output $\bar{d} = \tilde{d}+ \sum_{i = 1}^{T}d_i$ is a $3+\eps$ approximation of $d = \ed(x,y)$. Let $\tilde{y} =  y[p_1:q_1]\circ y[p_2: q_2]\circ \cdots \circ y[p_T: q_T]$. Notice that $\ed(x, \tilde{y})\leq \sum_{i = 1}^T \ed(x^i, y[p_i: q_i])$ and $\ed(\tilde{y}, y)\leq \tilde{d}$. We have
	 
	\begin{align*}
		\ed(x,y) & \leq \ed(x, \tilde{y})+ \ed(\tilde{y}, y) &\text{By triangle inequality}\\
		& \leq \tilde{d} +\sum_{i = 1}^Td_i&\\
		& = \bar{d} &\\
	\end{align*}
	On the other hand, we can divide $y$ into $T$ parts such that $y = \hat{y}^1\circ\hat{y}^2 \circ \cdots \circ \hat{y}^T$ and guarantee that 
	\begin{equation}
		\label{ed:eq4}
		\ed(x, y) = \sum_{i = 1}^T \ed(x^i, \hat{y}^i).
	\end{equation}
	Also, by Lemma~\ref{lem:findlongestsubstring},  $y[p_i:q_i]$ is the substring of $y$ that is closest to $x^i$ in edit distance, we know
	\begin{equation}
		\label{ed:eq5}
		\ed(x^i, \hat{y}^i) \geq \ed(x^i, y[p_i: q_i]).
	\end{equation}
	We have
	\begin{align*}
		\bar{d} & = \tilde{d} +\sum_{i = 1}^Td_i & \\ 
			&\leq (1+\eps) \ed(y, \tilde{y})  + \sum_{i = 1}^T\ed\big(x^i, y[p_i: q_i]\big)  & \\
			& \leq (1+\eps) \Big(\ed(y,x)+ \ed(\tilde{y}, x)\Big)+ \sum_{i = 1}^T\ed\big(x^i, y[p_i: q_i]\big)   & \\
			& \leq (1+\eps) \sum_{i = 1}^T \ed\big(x^i, \hat{y}^i\big) + (2+\eps)\sum_{i = 1}^T\ed\big(x^i, y[p_i: q_i]\big)  &\text{By \ref{ed:eq4}} \\
			&\leq (3+2\eps) \sum_{i = 1}^T \ed\big(x^i, \hat{y}^i\big) & \text{By \ref{ed:eq5}}\\
			&  = (3+2\eps)\ed(x,y).
	\end{align*}

	Since we can pick $\eps$ to be any constant, we pick $\eps' = 2 \eps$ and the output $d$ is indeed a $3+\eps'$ approximation of $\ed(x,y)$. 
\end{proof}

\begin{lemma}
	\label{lem:2^delta-approx_ed}
	Assume $\ed(x,y) = d$, for any constant $\delta\in (0,1/2)$, there is an algorithm that outputs a $2^{O(\frac{1}{\delta})}$-approximation of $\ed(x,y)$ with $O(\frac{d^\delta}{\delta} \polylog(n))$ bits of space in polynomial time in the asymmetric streaming model.  
\end{lemma}

\begin{proof}[Proof of Lemma~\ref{lem:2^delta-approx_ed}]
	We run algorithm~\ref{algo:ed_1} with parameter $\delta\in(0,1/2)$. Without loss of generality,we can assume $1/\delta$ is an integer. The time and space complexity follows from Lemma~\ref{lem:ed_time_space}.

	Again we can divide $x$ into $T$ parts such that $x^i =  x[a_i:a_{i+1}-1]$ where $a_i = 1+ \sum_{j=1}^{i-1}l_j$. By Lemma~\ref{lem:findlongestsubstring}, we have $\ed(x^i, y[p_i:q_i])\leq d_i$.

	Let $y = \hat{y}^1\circ\hat{y}^2 \circ \cdots \circ \hat{y}^T$ be a partition of $y$ such that 
	\begin{equation*}
		\ed(x, y) = \sum_{i = 1}^T \ed\big(x^i, \hat{y}^i\big).
	\end{equation*}

	We now show that $\bar{d}$ is a $2^{O(1/\delta)}$ approximation of $d = \ed(x,y)$. Similarly, we let $\tilde{y} =  y[p_1:q_1]\circ y[p_2: q_2]\circ \cdots \circ y[p_T: q_T]$.  We have
	\begin{align*}
		\ed(x,y) & \leq \ed(x, \tilde{y})+ \ed(\tilde{y}, y) &\text{By triangle inequality}\\
		& \leq \tilde{d} +\sum_{i = 1}^Td_i&\\
		& = \bar{d} &
	\end{align*}
	
	Let $y[p^*_i : q^*_i]$ be the substring of $y$ that is closest to $x^i$ in edit distance. By Lemma~\ref{lem:findlongestsubstring}, we know
	
	\begin{equation}
		\label{ed:eq6}
		\ed(x^i, y[p^*_i : q^*_i])\leq \ed(x^i, y[p_i : q_i])\leq c\ed(x^i, y[p^*_i : q^*_i])
	\end{equation}
	where $c = 2^{O(\log_s(k^{1-\delta}))} = 2^{O(1/\delta)}$. Thus, we have  
	\begin{equation}
		\label{ed:eq7}
		\ed\big(x^i, y[p_i: q_i]\big)\leq  c\ed(x^i, y[p^*_i : q^*_i])\leq c\ed\big(x^i, \hat{y}^i\big).
	\end{equation}
	Thus, we can get a similar upper bound of $\bar{d}$ by
	\begin{align*}
		\bar{d} & = \tilde{d} +\sum_{i = 1}^Td_i & \\ 
		&\leq (1+\eps) \ed(y, \tilde{y})  + \sum_{i = 1}^T\ed\big(x^i, y[p_i: q_i]\big)  & \\
		& \leq (1+\eps) \Big(\ed(y,x)+ \ed(\tilde{y}, x)\Big)+ \sum_{i = 1}^T\ed\big(x^i, y[p_i: q_i]\big)   & \\
		&\leq (1+\eps)(1+2c)\sum_{i = 1}^T \ed\big(x^i, \hat{y}^i\big) &  \text{By \ref{ed:eq7}}\\
		&  = 2^{O(1/\delta)}\ed(x,y).
	\end{align*}
	
	This finishes the proof. 
\end{proof}

\subsection{Algorithm for LCS}

We show that the algorithm presented in \cite{rubinstein2020reducing} for approximating $\lcs$ to $(1/2+\eps)$ factor can be slightly modified to work in the asymmetric streaming model. The following is a restatement of Theorem~\ref{thm:LCS_Algo}

\begin{theorem}[Restatement of Theorem~\ref{thm:LCS_Algo}]
	\label{lem:lcs_algo_2}
	Given two strings $x,y\in \{0,1\}^n$, for any constant $\delta\in(0,1/2)$, there exists a constant $\eps>0$, such that there is a one-pass deterministic algorithm that outputs a $2-\eps$ approximation of $\lcs(x,y)$ with $\tilde{O}(n^\delta/\delta)$ bits of space in the asymmetric streaming model in polynomial time.
\end{theorem} 

The algorithm from \cite{rubinstein2020reducing} uses four algorithms as ingredients $\mathsf{Match}$, $\mathsf{BestMatch}$, $\mathsf{Greedy}$, $\mathsf{ApproxED}$. We give a description here and show why they can be modified to work in the asymmetric streaming model.

Algorithm $\mathsf{Match}$ takes three inputs: two strings $x, y\in \Sigma^*$ and a symbol $\sigma\in \Sigma$. Let $\sigma(x)$ be the number of symbol $\sigma$ in string $x$ and we similarly define $\sigma(y)$. $\mathsf{Match}(x,y,\sigma)$ computes the length of longest common subsequence between $x$ and $y$ that consists of only $\sigma$ symbols. Thus, $\mathsf{Match}(x,y,\sigma) = \min(\sigma(x), \sigma(y))$. 

Algorithm $\mathsf{BestMatch}$ takes two strings $x, y\in \Sigma^*$ as inputs. It is defined as $\mathsf{BestMatch}(x,y) = \max_{\sigma\in\Sigma}\mathsf{Match}(x,y,\sigma)$. 

Algorithm $\mathsf{Greedy}$ takes three strings $x^1, x^2, y\in \Sigma^*$ as inputs. It finds the optimal partition $y = y^1\circ y^2$ that maximizes $\mathsf{BestMatch}(x^1,y^1) +\mathsf{BestMatch}(x^2,y^2)$. The output is $\mathsf{BestMatch}(x^1,y^1) +\mathsf{BestMatch}(x^2,y^2)$ and $y^1, y^2$. 

Algorithm $\mathsf{ApproxED}$ takes two strings $x, y\in \Sigma^n$ as inputs. Here we require the input strings have equal length. $\mathsf{ApproxED}$ first computes a constant approximation of $\ed(x,y)$, denoted by $\tilde{\ed}(x,y)$. The output is $n-\tilde{\ed}(x,y)$.

In the following, we assume $\Sigma = \{0,1\}$ and input strings to the main algorithm $x,y$ both have length $n$. All functions are normalized with respect to $n$. Thus, we have $1(x)$,$0(x)$,$\lcs(x,y)$, $\mathsf{ApproxED}(x,y)\in (0,1)$.

The algorithm first reduce the input to the perfectly unbalanced case. 

We first introduce a few parameters. 

$\alpha = \min\{1(x), 1(y), 0(x), 0(y)\}$.

$\beta = \theta(\alpha)$ is a constant. It can be smaller than $\alpha$ by an arbitrary constant factor.  

$\gamma\in(0,1)$ is a parameter that depends on the accuracy of the approximation algorithm for $\ed(x,y)$ we assume.

\begin{definition}[Perfectly unbalanced]
	We say two string $x, y\in \Sigma^n$ are perfectly unbalanced if
	
	\begin{equation}
		\lvert 1(x)-0(y)\rvert \leq \delta \alpha, 
	\end{equation}
	and 
	\begin{equation}
		0(x)\notin [1/2-\beta', 1/2+\beta'].
	\end{equation}
	
	Here, we require $\delta$ to be a sufficiently small constant such that $\delta \alpha\leq \beta$ and $\beta' = 10\beta$. 
\end{definition}

To see why we only need to consider the perfectly unbalanced case, \cite{rubinstein2020reducing} proved the following two Lemmas. 

\begin{lemma}
	\label{lem:perfectly_unbalanced_case}
	If $\lvert 1(x)-0(y)\rvert \delta \leq \delta \alpha$, then
	\begin{equation*}
		\mathsf{BestMatch}(x,y) \geq (1/2 + \delta/2)\lcs(x,y).
	\end{equation*}
\end{lemma}

\begin{lemma}
	\label{lem:perfectly_balanced_case}
	Let $\beta', \gamma>0$ be sufficiently small constants. If $0(x)\in [1/2-\beta', 1/2+\beta']$, then
	\begin{equation*}
		\max\{\mathsf{BestMatch}(x,y), \mathsf{ApproxED}(x,y)\}\leq (1/2+\gamma)\lcs(x,y).
	\end{equation*}
\end{lemma}
 
If the two input string are not perfectly unbalanced, we can compute $\mathsf{BestMatch}(x,y)$ and $\mathsf{ApproxED}$ to get a $1/2+\eps$ approximation fo $\lcs(x,y)$ for some small constant $\eps>0$. 

Given two strings in the perfectly unbalanced case, without loss of generality, we assume $1(y) = \alpha$. The algorithm first both strings $x, y$ into three parts such that $x = L_x\circ M_x\circ R_x$ and $y = L_y\circ M_y\circ R_y$ where $\lvert L_x\rvert = \lvert R_x\rvert = \lvert L_y\rvert = \lvert R_y\rvert = \alpha n$. Then, the inputs are divided into six cases according to the first order statistics (number of 0's and 1's) of $L_x$, $R_x$, $L_y$, $R_y$. For each case, we can use the four ingredient algorithms to get a $(1/2+\eps)$ approximation of $\lcs(x,y)$ for some small constant $\eps$. We refer readers to \cite{rubinstein2020reducing} for the pseudocode and analysis of the algorithm. We omit the details here. 

We now prove Theorem~\ref{lem:lcs_algo_2}. 

\begin{proof}[Proof of Theorem~\ref{lem:lcs_algo_2}]
	As usual, we assume $x $ is the online string and $y$ is the offline string. If the two input strings are not perfectly unbalanced, we can compute $\mathsf{BestMatch}(x,y)$ in $O(\log n)$ space in the asymmetric streaming model since we only need to compare the first order statistics of $x$ and $y$. Also, for any constant $\delta$ we can compute a constant approximation (dependent on $\delta$) of $\ed(x,y)$ using $\tilde{O}(n^\delta)$ space by Lemma~\ref{lem:2^delta-approx_ed}. Thus, we only need to consider the case where $x$ and $y$ are perfectly balanced.

	Notice that the algorithm from \cite{rubinstein2020reducing} needs to compute $\mathsf{Match}$, $\mathsf{BestMatch}$, $\mathsf{Greedy}$, $\mathsf{ApproxED}$ with input strings chosen from  $x$, $L_x$, $L_x\circ M_x$, $R_x$, $M_x\circ R_x$ and $y$, $L_y$, $L_y\circ M_y$, $R_y$, $M_y\circ R_y.$
	
	If we know the number of 1's and 0's in $L_x, M_x$, and $R_x$, then we can compute $\mathsf{Match}$ and $\mathsf{BestMatch}$ with any pair of input strings from $x$, $L_x$, $L_x\circ M_x$, $R_x$, $M_x\circ R_x$ and $y$, $L_y$, $L_y\circ M_y$, $R_y$, $M_y\circ R_y.$.
	
	For $\mathsf{ApproxED}$, according to the algorithm in \cite{rubinstein2020reducing},  we only need to compute $\mathsf{ApproxED}(x,y)$, $\mathsf{ApproxED}(L_x,L_y)$, and  $\mathsf{ApproxED}(R_x,R_y)$. For any constant $\delta>0$, we can get a constant approximation (dependent on $\delta$) of edit distance with $\tilde{O}(n^\delta)$ space in the asymmetric streaming model by Lemma~\ref{lem:2^delta-approx_ed}.
	
	For $\mathsf{Greedy}$, there are two cases. For the first case, the online string $x$ is divided into two parts and the input strings are $x^1, x^2, y$ where $x = x^1\circ x^2$.  Notice that $x^1$ can only be $L_x $ or $L_x\circ M_x$. In this case, we only need to remember $1(L_x)$, $1(M_x)$, and $1(R_x)$. Since the length of $L_x$, $M_x$, and $R_x$ are all fixed. We know $0(L_x) = \lvert L_x\rvert - 1(L_x)$, and similar for $M_x$ and $R_x$. We know for $l\in [n]$
	\begin{align*}
		& \mathsf{BestMatch}(y[1:l], x^1)+ \mathsf{BestMatch}(y[l+1:n], x^2) \\
		= & \max_{\sigma\in \{0,1\}}\mathsf{Match}(y[1:l], x^1, \sigma) + \max_{\sigma\in \{0,1\}}\mathsf{Match}(y[l+1:n], x^2, \sigma)\\
		= & \max_{\sigma\in \{0,1\}}\big(\min\{\sigma(y[1:l]),\sigma(x^1)\}\big) + \max_{\sigma\in \{0,1\}}\big(\min\{\sigma(y[l+1:n]),\sigma(x^2)\}\big)\\
	\end{align*}
	Given $\alpha = 1(y)$, we know $1(y[l+1:n]) = \alpha-1(y[1:l])$, $0(y[1:l]) = l/n-0(y[1:l]) $, and $0(y[l+1:n]) = (n-l)/n-0(y[l+1:n]) $. Thus we only need to read $y$ from left to right once and remember the index $l$ that maximizes $\mathsf{BestMatch}(y[1:l], x^1)+ \mathsf{BestMatch}(y[l+1:n], x^2)$. 
	
	For the case when the input strings are $y^1, y^2, x$, if we know $0(x)$, similarly, we can compute $\mathsf{Greedy}(y^1, y^1, x)$ by reading $x$ from left to right once with $O(\log n)$ bits of space. Here, $0(x)$ is not known to us before computation. However, in the perfectly unbalanced case, we assume $|1(y)-0(x) |< \delta$ is a sufficiently small constant. We can simply assume $0(x) = 1(y) = \alpha$ and run $\mathsf{BestMatch}(y^1, y^2, x)$ in the asymmetric streaming model. This will add an error of at most $\delta$. The algorithm still outputs a $(1/2+\eps)$ approximation of $\lcs(x,y)$ for some small constant $\eps>0$. 

\end{proof}

%

\bibliographystyle{alpha}
\bibliography{references}

\appendix

\section{Lower Bound for ED in the Standard Streaming Model}

\begin{theorem}
	There exists a constant $\eps>0$ such that for strings $x,y\in \{0,1\}^n$, any deterministic $R$ pass streaming algorithm achieving an $\eps n$ additive approximation of $\ed(x,y)$ needs $\Omega(n/R)$ space. 
\end{theorem}

\begin{proof}
	Consider an asymptotically good insertion-deletion code $C\subseteq \{0,1\}^n$ over a binary alphabet (See \cite{schulman1999asymptotically} for example). Assume $C$ has rate $\alpha$ and distance $\beta$. Both $\alpha$ and $\beta$ are some constants larger than 0, and we have $|C| = 2^{\alpha n}$. Also, for any $x, y\in C$ with $x \neq y$, we have $\ed(x,y)\geq \beta n$. Let $\eps =\beta/2$ and consider the two party communication problem where player 1 holds $x\in C$ and player 2 holds $y\in C$. The goal is to decide whether $x = y$. Any deterministic protocol has communication complexity at least $\log |C| = \Omega(n)$. Note that any algorithm that approximates $\ed(x,y)$ within an $\eps n$ additive error can decide whether $x=y$. Thus the theorem follows.
\end{proof}

We note that the same bound holds for Hamming distance by the same argument.

\end{document}